\pdfoutput=1

\pdfoutput=1

\documentclass[11pt]{scrartcl}
\usepackage[a4paper, total={16cm, 24cm}]{geometry}
\usepackage{microtype} %
\usepackage{authblk}
\title{\vspace{-1cm} Approval-Based Multiwinner Voting with Candidate Qualities}
\date{\vspace{-1.5cm}}

\author[1]{Niclas Boehmer}
\author[1]{Chris Dong}
\author[1]{Luca Kreisel}
\author[2]{Markus Utke}
\affil[1]{Hasso Plattner Institute, University of Potsdam, Germany}
\affil[2]{Eindhoven University of Technology, Netherlands}
\affil[ ]{\texttt{niclas.boehmer@hpi.de,chrisshuyu.dong@hpi.de\\luca.kreisel@hpi.de,m.utke@tue.nl}\vspace{-0.83cm}}

\usepackage[table]{xcolor}
\usepackage{float}
\usepackage{graphicx}
\usepackage{pgfplots}
\pgfplotsset{compat=1.17}

\usepackage{amsmath, amssymb, amsthm, mathtools, nicefrac}

\usepackage[pagebackref]{hyperref}
\hypersetup{
	pdfencoding=auto,
	psdextra,
	colorlinks=true,
	citecolor=green!50!black,
	linkcolor=red!60!black,
}

\usepackage{thmtools}
\usepackage{thm-restate}
\newtheorem{theorem}{Theorem}[section]

\newtheorem{example}[theorem]{Example}
\newtheorem{observation}[theorem]{Observation}
\newtheorem{proposition}[theorem]{Proposition}
\newtheorem{lemma}[theorem]{Lemma}
\newtheorem{corollary}[theorem]{Corollary}
\theoremstyle{definition}
\newtheorem{definition}[theorem]{Definition}

\usepackage[nameinlink]{cleveref}
\crefname{observation}{Observation}{Observation}

\usepackage{natbib}
\usepackage[inline]{enumitem}
\usepackage{booktabs}
\usepackage{subcaption}
\usepackage{bigdelim}

\newcommand{\restatehere}[1]{%
	\marginline{\vspace{0.6cm}\footnotesize \hyperlink{original#1}{\hypertarget{restated#1}{[Main]}}}%
	\csname #1\endcsname*%
}

\usepackage{nicefrac}
\usepackage{xspace}

\usepackage{etoolbox}
\usepackage{comment}

\usepackage{tikz}
\usetikzlibrary{positioning,arrows.meta,calc,decorations.markings}

\definecolor{partyone}{RGB}{239,162,154}
\definecolor{partytwo}{RGB}{159,197,237}
\definecolor{partythree}{RGB}{169,215,173}

\usepackage{booktabs}
\usepackage{tabularx}
\usepackage{multirow}
\usepackage{wrapfig}
\makeatletter
\def\wrapfill{%
  \par
  \ifx\parshape\WF@fudgeparshape
    \nobreak
    \ifnum\c@WF@wrappedlines>\@ne
      \advance\c@WF@wrappedlines\m@ne
      \vskip\c@WF@wrappedlines\baselineskip
      \global\c@WF@wrappedlines\z@
    \fi
    \allowbreak
    \WF@finale
  \fi
}
\makeatother

\newcommand{\thresholdPJR}{\text{threshold-PJR}}
\newcommand{\JR}{\text{JR}}
\newcommand{\thresholdJR}{\text{threshold-JR}}
\newcommand{\valueJR}{\text{value-JR}}
\newcommand{\appcut}[1]{\ensuremath{C_{\cap #1}}}
\newcommand{\appjoin}[1]{\ensuremath{C_{\cup #1}}}

\DeclareMathOperator*{\argmax}{arg\,max}

\DeclareRobustCommand{\abbrevcrefs}{%
    \crefname{figure}{Fig.}{Figs.}%
    \crefname{equation}{Eq.}{Eqs.}%
    \crefname{corollary}{Cor.}{Cors.}%
    \crefname{proposition}{Prop.}{Props.}%
    \crefname{theorem}{Thm.}{Thms.}%
    \crefname{lemma}{Lem.}{Lems.}%
    \crefname{table}{Tab.}{Tabs.}%
    \crefname{section}{Sec.}{Secs.}%
    \crefname{observation}{Obs.}{Obs.}%
}
\DeclareRobustCommand{\cshref}[1]{{\abbrevcrefs\cref{#1}}}

\newenvironment{sideexample}[1]
  {%
    \def\sideexampletable{#1}%
    \refstepcounter{example}%
    \par\addvspace{\topsep}%
    \noindent
    \begin{minipage}[t]{0.68\linewidth}%
      \vspace{0pt}%
      \textbf{Example~\theexample.}\enspace
      \itshape
  }
  {%
    \end{minipage}\hfill%
    \begin{minipage}[t]{0.28\linewidth}%
      \vspace{0pt}%
      \raggedleft
      \normalfont
      \sideexampletable
    \end{minipage}%
    \par\addvspace{\topsep}%
  }

\NewDocumentEnvironment{profileproof}
  {O{6} O{0.63\linewidth} m}
  {%
    \par\addvspace{\topsep}%
    \setlength{\columnsep}{1em}%
    \begin{wraptable}[#1]{r}{#2}
      \vspace{-\intextsep}
      \centering
      \normalfont
      #3
      \vspace{-.4\intextsep}
    \end{wraptable}%
    \noindent
    \normalfont
    \textit{\proofname.}\enspace
    \ignorespaces
  }
  {%
    \unskip\nobreak\hfill\qedsymbol
    \wrapfill
    \addvspace{\topsep}%
  }

\begin{document}

\maketitle

\bigskip
{\footnotesize\tableofcontents}
\begin{abstract}
	\begin{center}
		\textbf{\textsf{Abstract}} \smallskip
	\end{center}
	We initiate the study of a new model of approval-based multiwinner voting in which each candidate carries an exogenous quality score, capturing, for instance, the reliability of the candidate or their relevance to the context of the selection. Quality scores break with the standard assumption of approval-based multiwinner voting that candidates are fully defined by the set of their supporters. We rethink what proportional representation means in the presence of quality scores. For this, we introduce a threshold-based and a value-based family of axioms, analyze their relationships, satisfiability, and computational complexity, and present rules that achieve the strongest jointly satisfiable combinations of our proportionality axioms. We then analyze the compatibility of proportionality with the natural goal of maximizing the summed quality of the selected candidates. While imposing standard proportionality notions can lead to an almost complete loss of quality, we show that under a new class of \emph{reciprocal} axioms, which scale a group's entitlement by the quality of its commonly approved candidate(s), proportional committees retaining a summed quality of $\frac{3}{4}$ of the optimum always exist and can be computed by our voting rules at no additional computational cost.
\end{abstract}

\section{Introduction}
\label{sec:introduction}

Many collective decisions involve selecting a fixed-size subset of a given set of candidates: a city funds a slate of proposed projects, an HR department puts together a shortlist for a job, a funding agency composes an expert review panel, a search engine returns a set of documents for a query, or an electorate elects a parliament. Approval-based multiwinner voting provides a natural framework for such decisions: given voters' approval preferences over a set of candidates and a target size $k$, it asks which set of $k$ candidates, called a committee, should be selected. A central normative desideratum in the literature on approval-based multiwinner voting is \emph{proportionality}: a $\theta$-fraction of the voters should control roughly a $\theta$-fraction of the committee. Extensive work has formalized this intuition through a hierarchy of axioms \citep{azizJustifiedRepresentationApprovalbased2017,sanchez2017proportional,BrPe23a} and developed rules satisfying them \citep{DBLP:series/sbis/LacknerS23}.

In line with the principle of neutrality, the standard multiwinner voting model treats candidates as interchangeable: a candidate is fully described by the set of voters who approve them, and two candidates approved by the same voters are indistinguishable. The starting point of this paper is the observation that in many subset-selection problems, candidates do have an identity beyond their supporter sets that is relevant to the selection. We focus on the case where ``candidate identity'' takes the form of a score, which we call the candidate's \emph{quality} score (or quality for short). Quality scores are exogenously given, typically produced by a central assessment, and while they may well be correlated with candidates' popularity among voters, they may equally be unknown to the voters altogether. Quality scores for candidates arise in conceptually distinct ways across applications of multiwinner voting:

\begin{description}

\item[Relevance.]
Quality scores may measure how well a candidate fits the task or context for which the subset is selected. In shortlisting for a position \citep{DBLP:journals/corr/abs-2601-21277}, candidates fit the advertisement to differing degrees, information that voters may or may not take into account when expressing their preferences. In information retrieval, the task is to select a subset of documents in response to a user query \citep{DBLP:conf/sigir/DangC12,DBLP:conf/sigmod/BeharC22}: the documents act as candidates and naturally carry their relevance to the query as a score, while the voters might capture different perspectives or information needs that the selection should represent.

\item[Reliability.]
Quality may capture how dependable a candidate will be once selected. In participatory budgeting \citep{peters2021proportional,brill2023proportionality}, proposed projects routinely undergo a technical review of their implementation readiness, cost realism, and legal feasibility, which can be summarized in a reliability score and interpreted as the probability of successful implementation after selection. Reliability also matters, for instance, when composing a review board: candidates differ in their capacity and in their record of completing the work assigned to them.

\item[Urgency.]

Quality scores may express how time-sensitive it is that a candidate is selected. For instance, in participatory budgeting or collective agenda setting, some projects or items may be more pressing because a deadline approaches (potentially rendering a project infeasible in the future) or the cost of inaction becomes prohibitive. In these settings, voters might be interested in getting the most urgent projects they like selected to ensure they are implemented before a deadline passes. %

\item[Merit.]
Finally, quality scores may capture a candidate's general excellence,
independent of the task at hand or the group they would represent. This connects to the idea of candidate \emph{valence} in political elections
\citep{Stokes_1963,groseclose2001model}, which captures a candidate's generally valued characteristics, such as competence,
experience, and integrity, that are appealing to voters independent of their political positions. Another example is the selection of grant
proposals: proposals are reviewed with scores assessing their scientific excellence.
\end{description}

Orthogonal to the standard goal of proportional representation,
candidate quality also induces another natural decision goal: maximize the summed quality of the selected candidates. In this paper, we rethink proportionality in the presence of quality scores and analyze its compatibility with quality maximization.

\subsection{Our Contributions}
We initiate the study of approval-based multiwinner voting with quality
scores, guided by two questions: how should candidate quality affect the representation
that groups of voters deserve, and how can proportional representation be reconciled with maximizing the summed quality of the selected candidates?

In \Cref{sec:two-approaches}, we address the first question. To see why simply ignoring the scores and applying standard proportionality notions falls short, note that this would allow a group of voters entitled to representation to be served only by low-quality candidates it approves, even when it agrees on much stronger ones. This is clearly undesirable, for instance, when quality captures relevance in shortlisting or the success probability of projects in participatory budgeting.
To address this gap, we introduce two families of
quality-aware proportionality axioms, mirroring the standard hierarchy of
justified representation (JR), proportional justified representation (PJR),
and extended justified representation (EJR). \emph{Value-based} axioms demand
that a cohesive group receives a sufficient \emph{sum} of
quality from approved selected candidates.\footnote{This is natural when quality accumulates. For instance, the sum of
reliability scores can be interpreted as the expected number of successfully implemented projects
or completed reviews, and the sum of merit scores can measure the total
scientific contribution of funded proposals.} \emph{Threshold-based}
axioms instead demand that groups are represented by enough candidates of
\emph{sufficient} quality, at every quality level at which the group is
cohesive: if a sufficiently large group commonly approves high-quality
candidates, it must also be represented by high-quality candidates. Value-based axioms would, in contrast, also allow it to be represented by a larger number of medium-quality candidates.\footnote{Such a requirement is, for instance, natural when weak candidates could be discounted downstream:
in information retrieval, post-processing routinely discards documents of low
relevance, and in shortlisting, strong candidates may fill the position before
weaker ones are even considered.} We show that our notions of value- and threshold-based axioms are canonical: combining
natural quality-based representation requirements on party-list profiles with the
well-behavedness properties of \citet{DoPe26a} characterizes exactly (refinements
of) our axioms, thereby also pinpointing the normative distinction between them. We further
determine, for each axiom, whether it is always satisfiable and how hard it is
to verify.

In \Cref{sec:landscape}, we present a complete picture of the relationships
between our axioms and the standard hierarchy. We identify two appealing always satisfiable combinations: threshold-, value-, and standard PJR can be satisfied simultaneously, and value-EJR can be satisfied
together with standard PJR. For the former, we present a quality-adapted
variant of an expanding approvals rule that runs in polynomial time. For the latter, we present the value-based greedy cohesive
rule value-GCR that runs in exponential time. We complement value-GCR with a matching hardness result: computing a
value-EJR committee is NP-hard, already when candidates take only two distinct
quality values.
This also answers open questions of \citet{PPS23equalshares,baharav2026complexityjustifiedrepresentationadditive,DBLP:journals/corr/abs-2303-00621} on the computational complexity of computing EJR committees for additive utilities (see \Cref{sec:designing_rules} for a more detailed discussion).

In \Cref{sec:reciprocal}, we turn towards the second question of reconciling  roportionality with the objective of maximizing summed quality. We first observe that quality scores suggest a natural weakening of voter entitlements: a group approving only low-quality candidates may be granted fewer seats than an equally large group approving high-quality ones, as each seat spent on the former diminishes the committee's quality.
This motivates \emph{reciprocal} variants of our axioms, in which a group's entitlement scales with the quality of its commonly approved candidates. Reciprocity gives rise to a compromise between proportional representation and quality maximization: while imposing even the weakest non-reciprocal proportionality notions can cost a factor of $k$ in summed quality, the reciprocal notions can be satisfied by committees that
retain a $\nicefrac{3}{4}$-fraction of the maximum summed quality, and our reciprocal rules compute such committees at no additional computational cost.

\subsection{Related Work}
While most works on approval-based multiwinner voting assume that a candidate is fully characterized by their supporters, some equip candidates with additional attributes to restrict which committees can be selected \citep{DBLP:conf/sigecom/MasarikP024,DBLP:journals/corr/abs-2602-08504}. Most prominently, in participatory budgeting, each candidate carries a cost, and the total cost of the selection may not exceed a given budget \citep{DBLP:journals/corr/abs-2303-00621}. Thereby, costs from participatory budgeting and qualities from our work play different roles: costs constrain what may be selected, whereas any fixed-size candidate set is admissible in our setting.
Another popular constraint type are diversity constraints, where candidates carry attributes, and the committee must contain a prescribed number of candidates per attribute \citep{DBLP:conf/aaai/BredereckFILS18,DBLP:conf/ijcai/CelisHV18}.

Closer in spirit, outside of computational social choice, spatial models of single-winner elections have explored the role of candidate valence \citep{ansolabehere2000valence,groseclose2001model,aragones2002mixed}: there, a voter's utility for a candidate combines their ideological proximity with a candidate's valence score shared among all voters. In contrast, we study \emph{multiwinner} voting and use a different voter utility model.

Formally, our model is related to the special case of multiwinner voting with additive
utilities, in which every voter derives from each candidate either a candidate-specific utility or no utility.
In participatory budgeting, \citet{KRE25a} have termed this utility domain \emph{uniform utilities}  and analyzed how its structure can be exploited for faster computation of the Equal Shares voting rule. The uniform utilities model of  \citet{KRE25a} generalizes the widely used
\emph{cost utilities} in participatory budgeting, where one assumes that a voter's utility for an approved project is
its cost \citep{brill2023proportionality}.
However, previous work has not studied
proportionality in this domain in its own right when the score
is independent of the cost.

In terms of axioms, our value-based axioms align with proportionality notions developed for additive utilities
\citep{peters2021proportional,brill2023proportionality,DBLP:conf/ijcai/LosCG22,DBLP:conf/sigecom/MasarikP024,DBLP:journals/corr/abs-2602-08504} in the uniform utilities setting. %
Our threshold-based axioms are related to
proportionality axioms for multiwinner voting with ordinal preferences, which rely on measuring cohesiveness and satisfaction relative to a specific rank cutoff
\citep{dummett1984voting,DBLP:journals/scw/AzizL20,BrPe23a}, and to axioms
from proportional clustering,
where representation is demanded within distance cutoffs
\citep{DBLP:conf/nips/Kellerhals024,DBLP:conf/wine/AzizLCV24}. Unlike
these works, our cutoffs in threshold-based axioms are global quality values: a candidate lies above the cutoff for all or none of their supporters, and, relatedly, a fixed quality cutoff may correspond to different ranks in different voters' preferences.
As value- and threshold-based axioms have in the past been considered in different contexts, their relationship has, to the best of our knowledge, not been studied before. In our model, both families are defined over the same global candidate quality scores, allowing us to study their relationship and their compatibility with quality maximization.

Our analysis of the tension between proportionality and quality maximization
is related to a line of work on proportionality and utilitarian welfare (i.e., the total number of approvals that the selected candidates receive) in approval-based multiwinner voting.
Previous work has established welfare guarantees of proportional voting rules
\citep{DBLP:journals/ai/LacknerS20,DBLP:conf/sigecom/Skowron21,DBLP:conf/aaai/Brill024},
designed rules trading off proportionality and welfare
\citep{DBLP:conf/aaai/BaychkovBP26,DBLP:journals/corr/abs-2606-23320}, and
quantified the welfare loss of imposing proportionality
\citep{10.1145/3676953}. We differ in that our objective is the summed
quality of the committee, which is independent of the approval profile, and in
that we additionally design quality-aware relaxations of the proportionality
axioms.

\section{Model}
\label{sec:model}
For $x\in \mathbb N$, we write $[x] \coloneqq \{1,\dots, x\}$.\ Let $V$ be a set of $n$ \emph{voters} and $C$ a set of $m$ \emph{candidates}. Each voter $i\in V$ is equipped with an \emph{approval set} $A_i \subseteq C$ containing all candidates they approve. An \emph{approval profile} $A\coloneqq (A_i)_{i\in V}$ contains each voter's approval set.  %
For $c\in C$, we write $V[c]\coloneqq \{i \in V \mid c \in A_i\}$ to denote the \emph{supporters} of candidate $c$.
A profile is a \emph{party-list profile} if the set of voters can be partitioned into $V_1,\dots, V_r$, such that all voters within one voter set $V_j$ have the same, possibly empty, approval set $C_j$, and the sets $C_j$ and $ C_{j'}$ have empty intersection whenever $j\neq j'$.
Given a target committee size $k \in [m]$, we refer to sets of candidates of this size as \emph{committees}. In approval-based multiwinner voting, given a profile and target size $(A,k)$, the goal is to select a committee $W \subseteq C$ representing the voters.

\paragraph{Approval-Based Multiwinner Voting  with Candidate Qualities.}
We extend the standard  approval-based multiwinner voting  setting by equipping each candidate $c \in C$ with a \emph{quality score} $h_c \in (0,1]$. We call $\mathcal I = (A,k,h)$ an \emph{instance} of the quality-aware approval-based multiwinner voting  problem. An example instance is depicted in \Cref{tab:running_example}.
For a set of candidates $X\subseteq C$ and $\tau\in [0,1]$, %
we define $X^{\ge \tau}\coloneqq \{c \in X \mid h_c \geq \tau\}$ to denote the subset of candidates with quality at least $\tau$. Similarly, we denote by $X^{>\tau}\coloneqq \{c \in X \mid h_c > \tau\}$ and $X^{<\tau}\coloneqq \{c \in X \mid h_c < \tau\}$ the subsets of candidates with quality strictly greater than or strictly smaller than $\tau$, respectively. Given any subset $X \subseteq C$, the total quality of $X$ is $h(X) \coloneqq \sum_{c \in X} h_c$. We further denote the maximum quality achievable with at most $\ell\le m$ candidates from $X$ as $h(X\mid\ell) \coloneqq \max_{\substack{T\subseteq X\\ |T|\le \ell}} h(T)$.

\paragraph{Proportionality Notions.}
A \emph{proportionality notion} $\mathcal X$ maps each instance $\mathcal I$ to a (possibly empty) set of committees $\mathcal X(\mathcal I)$ fulfilling the notion. For two proportionality notions $\mathcal X, \mathcal Y$ we write $\mathcal Y\subseteq \mathcal X$ if $\mathcal Y(\mathcal I) \subseteq \mathcal X(\mathcal I)$ for all $\mathcal I$.
In formulating proportionality requirements, we follow a common approach from classical (approval-based) multiwinner voting: we first determine when a group of voters is entitled to representation, and then specify what exactly such a group is
entitled to.
Addressing the first point, a standard idea is that a group of voters is entitled to representation if it is sufficiently large and has sufficiently cohesive approval preferences. For a group of voters $S\subseteq V$, we write $\appcut{S} \coloneqq \bigcap_{i\in S} A_i$ for their shared approval set and $\appjoin{S} \coloneqq \bigcup_{i\in S}A_i$ for their joint approval set.
For $\ell \in [k]$, a group of voters $S \subseteq V$ is called \emph{$\ell$-large} if $\lvert S\rvert  \geq \ell \cdot \frac{n}{k}$, and \emph{$\ell$-cohesive} if it is $\ell$-large and $|\appcut{S}| \geq \ell$.
In the standard setting, the axioms justified representation (JR), proportional justified
representation (PJR), and extended justified representation (EJR) instantiate
this approach as follows: for every $\ell \in [k]$ and every $\ell$-cohesive
group $S \subseteq V$, JR requires $\appjoin{S} \cap W \neq \emptyset$, PJR
requires $\lvert \appjoin{S} \cap W \rvert \ge \ell$, and EJR requires
$\max_{i\in S} \lvert A_i\cap W\rvert \geq \ell$
\citep{azizJustifiedRepresentationApprovalbased2017,sanchez2017proportional}. %

\section{Defining Quality-Aware Proportionality}
\label{sec:two-approaches}
In this section, we introduce our proportionality notions that incorporate the quality scores of the candidates. We follow the general approach and ``template'' of the JR, PJR, and EJR hierarchy defined in \Cref{sec:model}.
While these standard notions can be applied directly in our setting, they fall short:
they guarantee groups of voters \emph{some number} of approved candidates in
the committee, but do not impose any requirements on the \emph{quality} of
these candidates. We therefore propose adaptations of JR
(\Cref{ax:jr}) as well as PJR and EJR (\Cref{ax:pjr_ejr}), which strengthen
the representation entitlements of groups based on candidate qualities. \Cref{tab:proportionality-notions}
provides a summary of all notions. In \Cref{sec:verify}, we discuss which of our axioms are satisfiable and the computational complexity of deciding whether a given committee fulfills the axiom.

\subsection{Adapting JR}\label{ax:jr}
\begin{table}
    \caption{Party-list profile with $n = 9$ voters. Voters are rows and candidates are columns, with \checkmark~indicating an approval. Each party consists of three voters approving four candidates. }
    \label{tab:intro-party-list-blocks}
    \centering
    \begin{tabular}{c @{\hspace{-10pt}} c c c c c c c c c c c c c}
        \toprule
        & & $a_1$ & $a_2$ & $a_3$ & $a_4$ & $b_1$ & $b_2$ & $b_3$ & $b_4$ & $c_1$ & $c_2$ & $c_3$ & $c_4$\\
        & quality & $1$ & $1$ & $1$ & $0.1$ & $\frac{1}{3}$ & $\frac{1}{3}$ & $\frac{1}{3}$ & $\frac{1}{3}$ & $1$ & $\frac{1}{2}$ & $\frac{1}{2}$ & $\frac{1}{2}$\\
        \midrule
        \ldelim\{{3}{*}[$V_1$] & $v_1$ & \checkmark & \checkmark & \checkmark & \checkmark & & & & & & & &\\
        & $v_2$ & \checkmark & \checkmark & \checkmark & \checkmark & & & & & & & &\\
        & $v_3$ & \checkmark & \checkmark & \checkmark & \checkmark & & & & & & & &\\
        \ldelim\{{3}{*}[$V_2$] & $v_4$ & & & & & \checkmark & \checkmark & \checkmark & \checkmark & & & &\\
        & $v_5$ & & & & & \checkmark & \checkmark & \checkmark & \checkmark & & & &\\
        & $v_6$ & & & & & \checkmark & \checkmark & \checkmark & \checkmark & & & &\\
        \ldelim\{{3}{*}[$V_3$] & $v_7$ & & & & & & & & & \checkmark & \checkmark & \checkmark & \checkmark\\
        & $v_8$ & & & & & & & & & \checkmark & \checkmark & \checkmark & \checkmark\\
        & $v_9$ & & & & & & & & & \checkmark & \checkmark & \checkmark & \checkmark\\
        \bottomrule
    \end{tabular}
\end{table}

We begin by adapting JR and discussing the representation a $1$-cohesive group of voters deserves. To build intuition and a normative basis for the representation entitlements in our quality-aware proportionality notions, we first consider the restricted class of party-list profiles (cf.\ \Cref{sec:model}). An example of a party-list profile is given in \Cref{tab:intro-party-list-blocks}, which we use as a running example.
Here, for $k=3$,
the ideal of proportional representation intuitively demands that each party, making up $\nicefrac{1}{3}$ of the voters, controls how to fill one of the three seats in the committee (as there are no fractional seats, the entitlement stays the same for $k=4$ and $k=5$). Standard JR implements this by demanding that for each of the three parties, one of the candidates is part of the committee; however, it only requires that each party is represented by some candidate, possibly only by its lowest-quality candidate, e.g.,  $a_4$ for $V_1$. This is unsatisfactory, as the parties also approve candidates of significantly higher quality, e.g., $a_1$ for $V_1$.
There are two natural ways to strengthen JR by incorporating candidate qualities.

First, we might assume that voters only care about the summed quality of the
selected candidates they approve. Then, for instance in our example, as party $V_3$ commonly supports the
quality-$1$ candidate $c_1$, $V_3$ is entitled to
representation with a total quality of $1$ for $k\geq 3$, i.e., candidates from $C_3$ of summed quality of at least $1$ need to be selected. More generally, each party deserving a seat is entitled to a total representation matching the quality of its commonly approved
candidate of highest quality. We call this the \emph{total quality} approach.
In the profile from \Cref{tab:intro-party-list-blocks} for $k=4$, the entitlement of the parties could, for instance, be met by selecting $\{a_1,b_1,c_2,c_3\}$.

Given its focus on total quality, this approach permits a party to be represented by multiple medium-quality candidates, even when it commonly approves candidates of high quality. In some settings, this
can be undesirable for the party because weak candidates may be discarded downstream: for instance, in the
information retrieval setting discussed in the introduction,
post-processing of the selected documents often drops those of insufficient
relevance. This motivates us to strengthen the total quality approach and require that party seats are filled
\emph{best-in-slot}, i.e., the seat a party deserves should be filled with its highest-quality
candidate.
In our running example, party $V_3$ would then need to be represented by its
highest-quality candidate $c_1$.

To formulate our quality-aware JR notions, we generalize the two
principles established on party-list profiles to general profiles, following again a standard approach from classical multiwinner voting.
When voter
preferences do not fall into the party list structure, we can emulate parties as follows: for
each group of voters $S\subseteq V$, we treat the voters as a party and the commonly approved
candidates $\appcut{S}$ as the party's candidates.
However, straightforward implementations of the party-list ideas, e.g., requiring each $1$-cohesive voter group $S$ to be represented via a candidate from its emulated party $\appcut{S}$ turn out to be too restrictive (cf.\ strong JR \citep{azizJustifiedRepresentationApprovalbased2017}).\footnote{Consider for example a profile with only quality-$1$ candidates and three voters with ballots $\{c_1,c_2\}$, $\{c_2,c_3\}$, and
$\{c_1,c_3\}$ for $k=2$. Each pair of voters forms a $1$-cohesive group, and
the three pairs emulate the pairwise distinct parties $\{c_2\}$, $\{c_3\}$,
and $\{c_1\}$, respectively. Serving each pair from within its emulated party
thus requires all three candidates, but $k=2$.}
Both our notions therefore allow entitlements to be met through any selected candidates that some voter in the group approves.

Generalizing the total quality approach, we thus require that a $1$-cohesive group should receive a summed quality of at least the quality of the best candidate in its emulated party $\appcut{S}$ through the union of all candidates they approve $\appjoin{S}$. We call the resulting notion \emph{value-JR}.
\begin{definition}
    A committee $W$ satisfies \emph{value-JR}\footnote{Alternatively, we could require in value-JR that the summed quality entitlement is met for a single voter, i.e., $h(A_i \cap W) \ge h(\appcut{S}\mid 1)$ for some $i\in S$. However, this variant would not be implied by value-PJR (cf. \Cref{tab:proportionality-notions} for a formal definition) and would thus break the JR to EJR hierarchy.} on instance $\mathcal I$, if
    for each $1$-cohesive $S\subseteq V$, we have that
    $h(\appjoin{S} \cap W) \ge h(\appcut{S}\mid 1)$.
\end{definition}

To adapt the best-in-slot idea, we require for each $1$-cohesive group $S\subseteq V$ that the best selected candidate approved by \emph{some} voter in $S$ has at least the quality of the best candidate in $\appcut{S}$. As it will be easier to generalize to the PJR- and EJR-level, we phrase this requirement
equivalently via quality cutoffs: for every cutoff $\tau$ such
that the emulated party $\appcut{S}$ contains a candidate of quality at least $\tau$, some voter in the group must approve a selected  candidate of quality at least $\tau$.
We call the resulting notion \emph{threshold-JR}.
\begin{definition}
    A committee $W$ satisfies \emph{threshold-JR} on instance $\mathcal I$,
    if for each $\tau\in [0,1]$ and each $1$-large $S\subseteq V$ with
    $\appcut{S}^{\ge \tau} \neq \emptyset$, we have that
    $\lvert \appjoin{S} \cap W^{\ge \tau} \rvert \geq 1$.
\end{definition}

\paragraph{Axiomatics.}
In \Cref{app:axiomatic_proportionality}, we analyze the two quality-aware JR notions through an axiomatic lens. We first adapt an axiom of \citet{delemazure2023strategyproofness} to formally capture the two normative intuitions discussed above on party-list profiles as the \emph{weak total representation} and \emph{weak best-in-slot representation} axioms. We then show that value- and threshold-JR are in some mathematically exact sense the canonical extensions of these two axioms to general profiles when considering $1$-cohesive groups.
That is, we prove that every proportionality notion satisfying one of the two party-list axioms together with some natural well-behavedness axioms of \citet{DoPe26a} is a refinement of value- or threshold-JR, respectively.

\subsection{Adapting PJR and EJR}\label{ax:pjr_ejr}

\begin{table*}[t]
    \centering
    \small
    \renewcommand{\arraystretch}{1.25}
    \setlength{\tabcolsep}{4pt}

    \caption{Entitlement and representation requirements of the standard, value, and threshold proportionality notions. All axioms except threshold-EJR are always satisfiable.}
    \label{tab:proportionality-notions}
    \begin{tabularx}{\textwidth}{
    @{}
    l
    >{\raggedright\arraybackslash}X
    >{\raggedright\arraybackslash}X
    >{\centering\arraybackslash}p{1.6cm}
    >{\centering\arraybackslash}p{2.4cm}
    @{}
    }
        \toprule
        Axiom
          & Entitlement condition\newline
            \emph{For every \(\ell\in[k]\) and every \(\ell\)-large
            \(S\subseteq V\), \(\ldots\)}
          & Required representation
          & Verifiability
          & Computability
          \\
        \midrule

        JR
          & \multirow{3}{=}{%
              \[
                \text{if } |\appcut{S}|\ge\ell
              \]}
          & \(|\appjoin{S} \cap W|\ge 1\).
          & P
          & P
          \\

        PJR
          &
          & \(|\appjoin{S} \cap W|\ge\ell\).
          & coNP-h.
          & P
          \\

        EJR
          &
          & \(\displaystyle
              \max_{i\in S}|A_i\cap W|\ge\ell\).
          & coNP-h.
          & P
          \\

        \midrule

        value-JR
          & \multirow{3}{=}{%
              \[
                \text{if } |\appcut{S}|\ge\ell
              \]}
          & \(\displaystyle
              h\bigl(\appjoin{S} \cap W\bigr)\ge h(\appcut{S}\mid 1)\).
          & coNP-h.
          & P
          \\

        value-PJR
          &
          & \(\displaystyle
              h\bigl(\appjoin{S} \cap W\bigr)\ge h(\appcut{S}\mid\ell)\).
          & coNP-h.
          & P
          \\

        value-EJR
          &
          & \(\displaystyle
              \max_{i\in S}h(A_i\cap W)
              \ge h(\appcut{S}\mid\ell)\).
          & coNP-h.
          & NP-h.
          \\

        \midrule

        threshold-JR
          & \multirow[c]{3}{=}{%
              \makebox[\linewidth][c]{%
                \shortstack[c]{%
                  \emph{For every \(\tau\in[0,1]\)}:\\[0.5ex]
                  \(\text{if } |\appcut{S}^{\ge\tau}|\ge\ell\)
                }}}
          & \(\displaystyle
              \left|\appjoin{S} \cap W^{\ge\tau}\right|\ge 1\).
          & P
          & P
          \\

        threshold-PJR
          &
          & \(\displaystyle
              \left|\appjoin{S} \cap W^{\ge\tau}\right|\ge\ell\).
          & coNP-h.
          & P
          \\

        threshold-EJR
          &
          & \(\displaystyle
              \max_{i\in S}|A_i\cap W^{\ge\tau}|
              \ge\ell\).
          & coNP-h.
          & NP-h.
          \\

        \bottomrule
  \end{tabularx}

\end{table*}

We proceed by extending these ideas to the more restrictive notions of PJR and EJR. Imposing standard PJR and EJR again runs into the problem that candidates a group is entitled to may be chosen sub-optimally in terms of quality. We discuss extensions of EJR here and refer to \Cref{tab:proportionality-notions} for definitions of the PJR variants, where we as usual require that entitlements are fulfilled by the union of candidates approved by a group of voters instead of a specific voter from the group as in EJR. Extending the total quality approach, value-EJR requires that every $\ell$-cohesive group of voters receives a summed quality of at least what it can achieve with $\ell$ of its commonly approved candidates. Formally, at the EJR-level, this yields the following notion.
\begin{definition}
    A committee $W$ satisfies \emph{value-EJR} on instance $\mathcal I$, if
    for each $\ell \in [k]$ and each $\ell$-cohesive group $S\subseteq V$, we have
    $\max_{i\in S} h(A_i \cap W) \ge h(\appcut{S}\mid \ell)$.
\end{definition}

This definition of value-EJR (and similarly value-PJR) coincides with the definition of EJR (and PJR) in participatory budgeting with additive utilities (see \citet{peters2021proportional} and \citet{brill2023proportionality}) for the special case of unit-costs and uniform utilities (as defined in \citet{KRE25a}).

To give some intuition for this notion, we consider the example instance from \Cref{tab:running_example} and identify some demands. The voter group $S = \{v_1, v_2, v_3, v_4\}$ is $2$-large and commonly approves $\appcut{S} = \{c_1, c_2\}$. Rather than requiring a certain number of approved candidates in the committee, value-EJR requires some voter in $S$ to approve candidates of total quality at least $h(\appcut{S} \mid 2) = h_{c_1} + h_{c_2} = \frac{5}{3}$ in the committee. Moreover, every size-two subgroup of $S$ is $1$-large and commonly approves the quality-$1$ candidate $c_1$, hence value-EJR requires one of the two voters from each of these size-two subgroups to approve candidates with total quality of $1$ in the committee. All value-EJR demands in this instance are, for example, satisfied by the committee $W = \{c_2, c_4, c_5\}$. In particular, voters $v_1$ and $v_2$ satisfy the $\frac{5}{3}$-demand, while voter $v_4$ approves candidates in the committee with total quality $\frac{4}{3}$. Thus, the group $\{v_3, v_4\}$ meets its value-EJR demand through $v_4$, even though neither voter approves a selected quality-$1$ candidate. This contrasts with threshold-EJR, which we introduce next.

For threshold-EJR, we require for each group of voters which is $\ell$-cohesive ``above some quality cutoff'', i.e., the group is $\ell$-large and commonly approves $\ell$ candidates above the cutoff, that some voter in the group is represented by $\ell$ candidates above this cutoff. This yields the following notion:
\begin{definition}
    A committee $W$ satisfies \emph{threshold-EJR} on instance $\mathcal I$,
    if for each $\ell \in [k]$, each $\tau\in [0,1]$, and each $\ell$-large
    $S\subseteq V$ with $\lvert \bigcap_{i\in S} A_{i}^{\ge \tau}\rvert \ge \ell$,
    we have $\max_{i\in S}\lvert A_i \cap W^{\ge \tau} \rvert \geq \ell$.
\end{definition}

We again consider the instance from \Cref{tab:running_example} and identify some demands arising at different quality thresholds. Since all candidates have quality at least $\frac{2}{3}$, threshold-EJR at threshold $\tau = \frac{2}{3}$ coincides with standard EJR, implying, e.g., that some voter from the $2$-cohesive group $S = \{v_1, v_2, v_3, v_4\}$ must approve at least two candidates in the committee. At threshold $\tau = 1$, we restrict attention to candidates of quality $1$. Every size-two subgroup of $S$ is $1$-large and commonly approves $c_1$, and hence must contain a voter who approves a selected quality-$1$ candidate. The requirements for the different thresholds have to hold simultaneously. For example, the committee $\{c_2, c_4, c_5\}$ satisfies all threshold-EJR requirements at $\tau = \frac{2}{3}$, and hence EJR, but violates threshold-EJR at $\tau = 1$, as the group $\{v_3, v_4\}$ commonly approves $c_1$ with quality $1$, while neither voter approves a selected candidate of quality $1$.

\begin{table}
    \caption{Example instance with $n = 6$ voters.}
    \label{tab:running_example}
    \centering
    \begin{tabular}{c c c c c c}
        \toprule
        $k = 3$ & $c_1$ & $c_2$ & $c_3$ & $c_4$ & $c_5$\\
        quality & $1$ & $\frac{2}{3}$ & $1$ & $1$ & $\frac{2}{3}$\\
        \midrule
        $v_1$ & \checkmark & \checkmark & & \checkmark &\\
        $v_2$ & \checkmark & \checkmark & & \checkmark &\\
        $v_3$ & \checkmark & \checkmark & & &\\
        $v_4$ & \checkmark & \checkmark & & & \checkmark \\
        $v_5$ & \checkmark & \checkmark & \checkmark & & \checkmark  \\
        $v_6$ & & & \checkmark & &\\
        \bottomrule
    \end{tabular}
\end{table}

\paragraph{Axiomatics and an Efficiently Verifiable Strengthening.}
In Appendix \ref{app:axiomatic_proportionality}, we analyze value-EJR/PJR and threshold-EJR/PJR axiomatically.
As a normative basis, we generalize an axiom called lower quota for party lists to the best-in-slot and total quality approaches, and prove that threshold- and value-PJR/EJR adhere to these principles. Further, we show that under the well-behavedness axioms of \citet{DoPe26a}, the following strengthening of threshold-PJR is the canonical requirement for proportional representation when generalizing the best-in-slot approach on party-list profiles to general profiles and larger cohesive groups.

\begin{definition}
    A committee $W$ satisfies \emph{threshold-PJR+}, if for every $\ell\in [k]$, $\tau \in [0,1]$, and $S\subseteq V$ such that $\lvert S\rvert \ge \frac{\ell n}{k}$ and $\lvert \appcut{S}^{\ge \tau} \setminus W \rvert \ge 1$, we have $ \left| \appjoin{S} \cap W^{\ge\tau}\right|\ge\ell$.\footnote{There is also an equivalent formulation of threshold-PJR+ closer in spirit to the PJR+ notion from classical approval-based multiwinner voting \citep{BrPe23a}: $W$ satisfies threshold-PJR+, if for every $\ell \in [k]$ and candidate $c\in C\setminus W$, and every $S\subseteq V[c]$ of size at least $\ell \frac nk$, we have that $\lvert C_{\cup S} \cap W^{\ge h_c}\rvert \ge \ell$.}
\end{definition}
Beyond its axiomatic appeal, as we argue in \Cref{sec:verify}, threshold-PJR+ has the advantage that it can be verified in polynomial time.
A similar approach using total lower quota leads to a technical refinement of value-PJR, which we omit here.\footnote{Roughly speaking, this refinement states that a committee $W$ satisfies value-PJR+ if it satisfies value-PJR and under all additions of candidates from $W$ to some approval ballots, $W$ still satisfies value-PJR. This axiom is, arguably, unintuitive. Further, one can show that this version remains coNP-hard to verify.}

\subsection{Satisfiability and Verification}\label{sec:verify}
From a practical perspective, notions that are not always satisfiable may be too restrictive to use. In \Cref{sec:designing_rules}, we develop algorithms showing that all so-far discussed proportionality notions are always satisfiable, with the exception of threshold-EJR. For threshold-EJR consider the following counterexample:

\begin{sideexample}{%
  \begin{tabular}{@{}c c c c c@{}}
    \toprule
     & $c_1$      & $c_2$      & $d_1$      & $d_2$      \\
    quality & $1$        & $1$        & $0.1$      & $0.1$      \\
    \midrule
    $v_1$     & \checkmark &             & \checkmark & \checkmark \\
    $v_2$     &            & \checkmark  & \checkmark & \checkmark \\
    \bottomrule
  \end{tabular}%
}
Consider the instance depicted in the table with $k=n=2$.
At $\tau=1$, each voter forms a $1$-cohesive group and deserves representation from a candidate with quality $1$, thus $c_1$ and $c_2$ must be selected. At $\tau=0.1$, both voters together form a $2$-cohesive group and therefore some voter must be represented by two
candidates, which is only possible if $d_1$ or $d_2$ is selected.
\end{sideexample}

Determining whether a given instance admits a threshold-EJR committee turns out to be intractable. We provide a proof for this in Appendix \ref{app:designing_rules} (see \Cref{prop:threshold-EJR_existence_hardness}).
\begin{restatable}{proposition}{thresholdEJRExistsIsHard}
    It is NP-hard to determine whether a given instance admits a  committee satisfying threshold-EJR.
\end{restatable}

\paragraph{Verification.}
To audit whether a selected committee is proportional, efficient algorithms to verify whether a given committee fulfills some proportionality axiom are needed. This is possible in polynomial time for all JR-level definitions, except for value-JR, as we show in \Cref{thm:value-JR_NP_Verification}. %
Unfortunately, for all PJR- and EJR-level definitions, verification is
intractable: on instances in which all candidates have quality $1$, our
notions reduce to their standard analogs, for which verification is
already coNP-complete \citep{azizJustifiedRepresentationApprovalbased2017,DBLP:conf/aaai/0001EHLFS18}.
In contrast, as for PJR+ \citep{BrPe23a}, threshold-PJR+ is polynomial-time verifiable: we can iterate over all $c\in C\setminus W$ and only need to check that $W^{\ge h_c}$ satisfies PJR+ on the reduced instance with candidates of quality at least $h_c$, which is known to be possible in polynomial time via submodular optimization.

\section{Axiom Relations and Designing Proportional Rules}
\label{sec:landscape}
In this section, we first analyze implications and (in-)compatibilities between our axioms.
Based on this, we then design rules which provide the best possible jointly satisfiable proportionality guarantees. A summary of all our results can be found in \Cref{fig:relations}.
\tikzset{
  implication/.style={
    ->,
    >=Stealth,
    thick
  },
  nonimplication/.style={
    black,
    dotted,
    >=Stealth,
    thick,
  },
  incompatibility/.style={
    red,
    dashed,
    semithick
  },
}

\newcommand{\arrowref}[6][0.28]{%
  \begingroup

  \path
    (#3) coordinate (arrowref-start)
    (#4) coordinate (arrowref-end);

  \pgfmathanglebetweenpoints
    {\pgfpointanchor{arrowref-start}{center}}
    {\pgfpointanchor{arrowref-end}{center}}%
  \edef\arrowrefangle{\pgfmathresult}%

  \pgfmathparse{%
    (\arrowrefangle > 90 && \arrowrefangle < 270)
      ? \arrowrefangle - 180
      : \arrowrefangle
  }%
  \edef\arrowreflabelangle{\pgfmathresult}%

  \ifstrequal{#5}{below}
    {\pgfmathsetmacro{\arrowrefshiftangle}{\arrowreflabelangle-90}}
    {\pgfmathsetmacro{\arrowrefshiftangle}{\arrowreflabelangle+90}}

  \draw[#2]
    (#3) -- (#4)
    node[
      pos=#1,
      shift={({\arrowrefshiftangle}:6pt)}
    ]
    {\rotatebox{\arrowreflabelangle}{
        \scriptsize #6}
    };
  \endgroup
}

\begin{figure*}[t]
  \centering

  \begin{tikzpicture}[
      x=5cm,
      y=1.8cm,
      every node/.style={align=center},
      implication/.style={-{Stealth},thick},
      nonimplication/.style={dotted,thick},
      incompatibility/.style={red,dashed,semithick}
    ]

    \node (EJRl) at (-1,2) {EJR};
    \node (PJRl) at (-1,1) {PJR};
    \node (JRl)  at (-1,0) {JR};

    \node (tEJR) at (0,2)
    {threshold-EJR\\[-1mm]\scriptsize unsatisfiable};
    \node (tPJR) at (0,1)
    {threshold-PJR\\[-1mm]\scriptsize polynomial-time satisfiable};
    \node (tJR) at (0,0)
    {threshold-JR};

    \node (vEJR) at (1,2)
    {value-EJR\\[-1mm]\scriptsize satisfiable};
    \node (vPJR) at (1,1)
    {value-PJR};
    \node (vJR) at (1,0)
    {value-JR};

    \node (EJRr) at (2,2) {EJR};
    \node (PJRr) at (2,1) {PJR};
    \node (JRr)  at (2,0) {JR};

    \draw[implication]
    (tEJR.south) -- (tPJR.north);

    \draw[implication]
    (tPJR.south) -- (tJR.north);

    \draw[implication]
    (vEJR.south) -- (vPJR.north);

    \draw[implication]
    (vPJR.south) -- (vJR.north);

    \draw[implication] (EJRl.south) -- (PJRl.north);
    \draw[implication] (PJRl.south) -- (JRl.north);

    \draw[implication] (EJRr.south) -- (PJRr.north);
    \draw[implication] (PJRr.south) -- (JRr.north);

    \arrowref[0.25]
      {incompatibility}
      {EJRl.south east}
      {tJR.north west}
      {above}
        {\cshref{prop:EJR_tJR_incompatible}}

    \arrowref[0.33]
        {nonimplication}
        {vEJR.east}
        {PJRr.west}
        {above}
        {\cshref{prop:vEJR_PJR}, \cshref{prop:vEJR_PJR_compatible}}

    \arrowref[0.3]
        {incompatibility}
        {vJR.north east}
        {EJRr.south west}
        {above}
        {\cshref{prop:EJR_vJR_incompatible}}

    \arrowref[0.5]
        {incompatibility}
        {tEJR.east}
        {vEJR.west}
        {above}
        {\cshref{thm:threshold_ejr_not_value_ejr}*}

    \arrowref[0.23]
        {incompatibility}
        {tJR.north east}
        {vEJR.south west}
        {below}
        {\cshref{prop:vEJR_tJR_incompatible}}

    \arrowref[0.5]
        {implication}
        {tEJR.west}
        {EJRl.east}
        {below}
        {\cshref{obs:threshold_implies_basic}}

    \arrowref[0.65]
        {implication}
        {tPJR.west}
        {PJRl.east}
        {below}
        {\cshref{obs:threshold_implies_basic}}

    \arrowref[0.5]
        {implication}
        {tJR.west}
        {JRl.east}
        {below}
        {\cshref{obs:threshold_implies_basic}}

    \arrowref[0.5]
        {implication}
        {vJR.east}
        {JRr.west}
        {below}
        {\cshref{obs:JR_vJR}}

    \arrowref[0.7]
        {implication}
        {tPJR.east}
        {vPJR.west}
        {below}
        {\cshref{prop:tPJR_implies_vPJR}}

    \arrowref[0.5]
        {implication}
        {tJR.east}
        {vJR.west}
        {below}
        {\cshref{cor:tJR_implies_vJR}}

  \end{tikzpicture}
  \caption{Relationships of threshold-, value- and JR notions. Arrows denote implications, red dashed lines incompatibilities and dotted lines compatibilities without implications. The incompatibility between threshold-EJR and value-EJR (marked with \textcolor{red}{*}) also holds on instances where threshold-EJR committees exist. Only the implications that can be inferred from the depicted arrows hold.
  }\label{fig:relations}\label{fig:relationships}
\end{figure*}
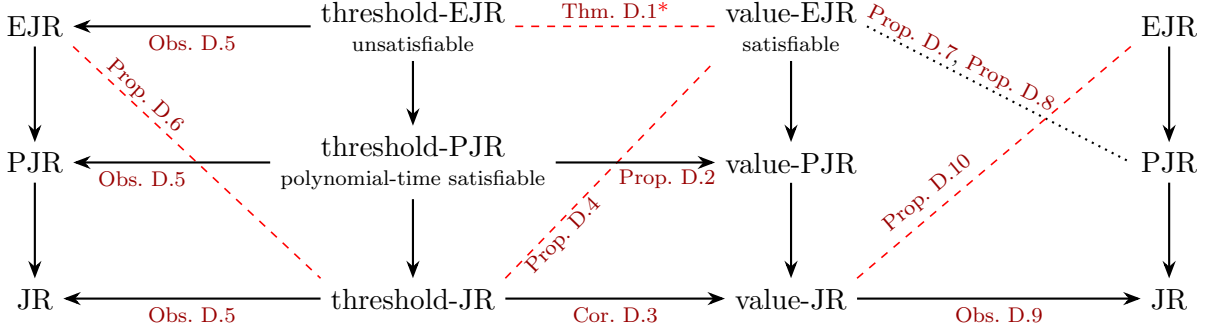

\subsection{Axiom Relations} \label{sec:axioms_relationships}
We first analyze the implications between our axioms in
\Cref{fig:relations}
(formal proofs can be found in \Cref{app:relations}; an intuitive discussion is included here). Within each
family, x-EJR implies x-PJR implies x-JR for $x\in\{
\text{value}, \text{threshold}\}$, and threshold-PJR+ implies threshold-PJR.
As all our notions reduce to their standard analogs when every candidate has
quality $1$, counterexamples from the standard setting carry over, and therefore no converse implication holds. Moreover, each threshold-notion implies its standard counterpart, as can be seen by choosing a cutoff below the minimum candidate quality. Value-JR implies JR, as being represented with positive quality implies that the voter approves at least one selected candidate. However, value-PJR and even value-EJR do not imply standard PJR, as it may be that an $\ell$-cohesive group is represented with sufficient total quality while only approving a single (high-quality) selected candidate.

Turning to the relations between value and threshold notions, at the JR- and PJR-level, the threshold notions imply their value counterparts: intuitively, threshold-(P)JR requires for every $\ell$-cohesive group that the $\ell$ selected candidates of highest quality approved by at least one voter from the group  are pointwise at least as good as the $\ell$ best candidates they commonly approve. This implies that the group's summed quality entitlement is also met.
Interestingly, at the EJR-level, this relation breaks down: here threshold-EJR and value-EJR are incompatible, even on instances that admit a threshold-EJR committee (\Cref{thm:threshold_ejr_not_value_ejr}).

Turning from implications to incompatibilities, it turns out that all PJR-level notions are jointly
satisfiable, while EJR notions are incompatible with various JR ones.
\begin{proposition}
\label{thm:relationship-map}
 EJR (and threshold-EJR) is incompatible with threshold- and value-JR.
Further, value-EJR is incompatible with threshold-JR.
\end{proposition}

\begin{profileproof}{%
  \setlength{\tabcolsep}{3pt}%
  \begin{tabular}{@{}c@{\hspace{1.2em}}c@{}}
    \begin{tabular}{@{}c
        >{\columncolor{gray!20}}c
        >{\columncolor{gray!20}}c
        c c@{}}
      \toprule
      $k=2$
        & $c_1$      & $c_2$      & $d_1$      & $d_2$      \\
      quality
        & $1$        & $1$        & $0.1$      & $0.1$      \\
      \midrule
      $v_1$
        & \checkmark &             & \checkmark & \checkmark \\
      $v_2$
        &            & \checkmark  & \checkmark & \checkmark \\
      \bottomrule
    \end{tabular}
    &
    \begin{tabular}{@{}c
        >{\columncolor{gray!20}}c
        >{\columncolor{gray!20}}c
        c c@{}}
      \toprule
      $k=2$
        & $c_1$      & $c_2$      & $d_1$      & $d_2$      \\
      quality
        & $1$        & $1$        & $0.6$      & $0.6$      \\
      \midrule
      $v_1$
        & \checkmark &             & \checkmark & \checkmark \\
      $v_2$
        &            & \checkmark  & \checkmark & \checkmark \\
      \bottomrule
    \end{tabular}
  \end{tabular}%
}
Consider the following two instances. In both profiles, the grey columns identify the committee $\{c_1,c_2\}$. In the first instance, this is the only committee satisfying threshold-JR or value-JR, but it violates EJR for the voter group $\{v_1,v_2\}$. In the second instance, it is the only committee satisfying threshold-JR, but it violates value-EJR.
\end{profileproof}

\subsection{Designing Rules} \label{sec:designing_rules}
We now turn to designing rules that achieve appealing combinations of our axioms.
Having identified the boundary of simultaneously satisfiable axioms in the previous section, two always satisfiable design goals remain: satisfying all PJR-level axioms simultaneously, or satisfying value-EJR and standard PJR. For the first goal, we modify an existing class of rules due to \citet{DBLP:journals/scw/AzizL20} as follows.
\begin{definition}
    We say that $f$ is an \emph{expanding approvals rule (EAR)}, if it can be described as follows.
    Start with an empty committee $W$ and budgets $b_i = \frac{k}{n}$ for each voter $i \in V$. Let $r$ be the number of distinct quality scores and assume $h^1 > h^2 > \dots > h^r$ to be the distinct quality scores. For $j = 1, \ldots, r$, consider $C^j \coloneqq \{c\in C\setminus W \mid h_c \geq h^j \,\land\, \sum_{i\in V[c]} b_i \geq 1\}$, i.e., the set of candidates of quality at least $h^j$, whose supporters have a total remaining budget of at least $1$. While $C^j \neq \emptyset$, add a candidate $c \in C^j$ to $W$ and let voters in $V[c]$ pay a cost of $1$ for $c$, where the cost is arbitrarily distributed among $V[c]$ without overdrawing any voter's budget and $p_i(c)$ is the amount voter $i$ paid for candidate $c$. Update each $b_i$ by subtracting $p_i(c)$ and update $C^j$.
    Once no candidate in $C^r$ can be afforded by their supporters, add arbitrary candidates to $W$ until it contains $k$ candidates.
\end{definition}

All rules from this class satisfy threshold-PJR+, and therefore also threshold-, value-, and standard PJR.

\begin{restatable}{proposition}{EARThesholdPJR}\label{prop:EAR_threshold-PJR+}
   If $f$ is an EAR, it satisfies threshold-PJR+.
\end{restatable}

Among the EAR rules, there are several appealing, computationally tractable ones. For instance, there is a natural adaptation of the Method of Equal Shares (MES) \citep{peters2020proportionality} to our setting, which we term \emph{threshold-MES} and which is the EAR such that among $C^j$, we choose the candidate $c$ minimizing $\rho^*(c)$, where $\rho^*(c)$ is the smallest non-negative $\rho$ such that $\sum_{i\in V[c]} \min(b_i, \rho) = 1$, and we deduct from the budget of each $i\in V[c]$ precisely $\min(b_i, \rho^*(c))$. Clearly, a committee selected by threshold-MES can be output in polynomial time, and the rule satisfies threshold-PJR+ by \Cref{prop:EAR_threshold-PJR+}.

For the second goal, we present an exponential-time rule which satisfies both value-EJR and PJR, which is based on a rule due to \citet{peters2021proportional}:

\begin{restatable}{definition}{valueGCR}
   The \emph{value-based greedy cohesive rule with price completion (value-GCR)}
   starts with the empty committee $W = \emptyset$ and budgets
   $b_i = \frac{k}{n}$. For some set of voters $S \subseteq V$ and $\ell\in [k]$, we say that $(S, \ell)$ \emph{witnesses a value-EJR
   violation} if $S$ is $\ell$-cohesive but we
   have $h(A_i\cap W)< h(\appcut{S}\mid \ell)$ for all $i \in S$. In the first phase, as long as some value-EJR violation witness exists, we
   iteratively choose the witness $(S,\ell)$ maximizing
   $h(\appcut{S}\mid \ell)$. Let $T\subseteq \appcut{S}$ denote a set of
   candidates with $h(T)=h(\appcut{S}\mid \ell)$ and $|T|= \ell$. We add all
   candidates from $T\setminus W$ to $W$ and deduct
   $\frac{\lvert T\setminus W \rvert}{\lvert S\rvert}$ from the budget of
   each $i\in S$. We proceed until no witness exists. In the second phase, as
   long as some candidate can be afforded by their supporting voters, i.e.,
   $c\in C\setminus W$ satisfies $\sum_{i \in V[c]} b_i \ge 1$, we add $c$ to
   $W$ and deduct a total budget of $1$ from the voters in $V[c]$ arbitrarily
   without charging any voter more than their remaining budget. Finally, if
   $|W| < k$, we complete $W$ arbitrarily.
\end{restatable}

\begin{restatable}{theorem}{valueGCRValueEJR} \label{thm:value-GCR_value-EJR_PJR}
    Value-GCR satisfies value-EJR and PJR.
\end{restatable}

We give the proof in \Cref{app:value-GCR_value-EJR_PJR}. Note, however, that value-GCR runs in exponential time, as it needs to check exponentially many groups of voters to identify a value-EJR violation witness. This is not incidental, as we show that, assuming $\mathrm{P}\neq\mathrm{NP}$, no polynomial-time rule guaranteeing value-EJR exists:

\begin{restatable}{theorem}{valueEJRExistsIsHard}\label{thm:value-EJR-hard}
    If some polynomial-time algorithm outputs a committee satisfying value-EJR for every instance
    $\mathcal I$, then
    $\mathrm{P}=\mathrm{NP}$. This holds already on instances in which there
    are only candidates of quality $1$ and one other, unary-encoded quality.
\end{restatable}
Recall that our setting is a special case of unit-cost participatory budgeting with additive utilities (cf.\ \Cref{ax:pjr_ejr}).
Assuming $\mathrm{P}\neq\mathrm{NP}$, \Cref{thm:value-EJR-hard} thus negatively answers an open question of \citet{PPS23equalshares} whether  an EJR outcome can be computed in strongly polynomial time for additive utilities and unit costs in participatory budgeting, a question whose complexity \citet{baharav2026complexityjustifiedrepresentationadditive} also recently posed as an open problem. Moreover, since our hardness holds already for unary-encoded utilities and unit costs, it also rules out pseudo-polynomial-time algorithms for computing
an EJR outcome, negatively answering a question posed by
\citet[Sec.~5.1.1.2]{DBLP:journals/corr/abs-2303-00621} whether an EJR
outcome can be found in time polynomial in the budget, costs, and utilities
of an instance in participatory budgeting.
Somewhat surprisingly, our result provides negative answers to these questions already in a---in terms of cost and utility values strongly restricted---special case.

Lastly, in \Cref{app:MES_value-EJR-1}, we observe that a MES variation defined by \cite{peters2021proportional}  adapted to our setting satisfies standard PJR and a relaxation of value-EJR while being computationally tractable.

\section{Quality-Aware Proportionality and Quality Maximization}
\label{sec:reciprocal}
So far, we have only considered how the representation demands of voter groups may increase in the presence of quality scores. However, an orthogonal perspective on qualities is that a group approving only low-quality candidates may deserve less entitlement: every seat spent on representing such a group is filled by a candidate contributing little quality, and thus diminishes the overall quality of the committee.
As a specific example, suppose quality measures how likely the selected candidate, e.g., a project, is to be successfully implemented. A quality-$0.5$ project would then need twice as many supporters as a quality-$1$ project to benefit the same number of voters in expectation. Since both occupy one slot, it seems natural to require more support for the less reliable project.
In \Cref{rec:prop}, we define \emph{reciprocal} proportionality notions implementing this idea, in which a group's entitlement scales with the quality of its commonly approved candidates.

The idea behind reciprocal proportionality notions is in line with the natural goal to maximize the summed quality of the selected candidates; a natural desideratum in many application domains, e.g., when selecting the most relevant documents in information retrieval or the most reliable projects in participatory
budgeting.
In \Cref{pop}, we explore the trade-off between this standard objective of maximizing committee quality and our notions of proportionality, showing that moving to reciprocal notions allows us to bound the quality loss of imposing proportionality.

\subsection{Reciprocal Proportionality}\label{rec:prop}
As motivated above, another way to think about proportionality in the presence of quality scores is to incorporate them by requiring that a group only approving low-quality candidates needs to be larger to be eligible for representation.
Consider again the example profile in \Cref{tab:intro-party-list-blocks} for $k = 6$. Under standard EJR, every party is entitled to two candidates. In particular, two of party~2's low-quality candidates must be selected, while several high-quality candidates are not selected.

Our reciprocal axioms instead scale each group's entitlement with the quality of its approved candidates: an $\ell$-large group whose $\ell$
highest-quality commonly approved candidates all have quality $h$ is only
entitled to $\ell \cdot h$ candidates.
We start by establishing some terminology:
\begin{definition}
    We say that a group of voters $S$ can \emph{reciprocally afford} a set of     candidates $T$ if $T \subseteq \appcut{S}$ and $\lvert S \rvert \ge \sum_{c \in T} \frac{1}{h_c} \cdot \frac{n}{k}$. We just say a set $T\subseteq C$ is \emph{reciprocally affordable} if $\sum_{c \in T} \frac{1}{h_c} \le k$. Further, for $X\subseteq C$, we write $q_\ell(X)\coloneqq \min_{T\subseteq X,\, |T|=\ell} \sum_{c\in T}\frac{1}{h_c}$, where $q_\ell(X)=\infty$ if $X$ contains fewer than $\ell$ candidates.
\end{definition}

Note that a group $S$ can reciprocally afford $\ell$ of its commonly approved candidates if and only if $\lvert S \rvert \ge q_\ell(\appcut{S}) \cdot \frac{n}{k}$. Replacing the $\ell$-largeness condition on a voter group $S$ by the condition that $S$ can reciprocally afford a candidate set of size $\ell$ leads to reciprocal adaptations of all previously considered JR-, PJR-, and EJR-type definitions. Note that each reciprocal axiom is a weakening of its non-reciprocal version. At the JR-level, we, for instance, get the following notion:
\begin{definition}
    A committee $W$ satisfies \emph{reciprocal JR} on instance $\mathcal I$,
    if for each $S\subseteq V$ with
    $\lvert S \rvert \ge q_1(\appcut{S}) \cdot \frac{n}{k}$, we have that
    $\appjoin{S} \cap W \neq \emptyset$.
\end{definition}
The formal definitions of all other reciprocal notions can be found in \Cref{tab:reciprocal-proportionality-notions}.
Reciprocal affordability can also be linked to participatory budgeting via a cost-based interpretation: viewing each candidate $c$ as having a cost of $\nicefrac{1}{h_c}$ and granting each voter a budget of $\nicefrac{k}{n}$, the reciprocal affordability requirement corresponds exactly to cohesiveness in participatory budgeting with budget $k$ \citep{peters2021proportional}.

We return to the running example from \Cref{tab:running_example} to illustrate reciprocal EJR. The group $S = \{v_1, v_2, v_3, v_4\}$ is $2$-large and $2$-cohesive over $T = \{c_1, c_2\}$ and therefore has a two-candidate claim under standard EJR. However, since $c_1$ and $c_2$ have quality $1$ and $\frac{2}{3}$, respectively, making the same claim under reciprocal EJR would require the group to contain at least $(\frac{1}{h_{c_1}} + \frac{1}{h_{c_2}}) \frac{n}{k} = (1 + \frac{3}{2}) \frac{6}{3} = 5$ voters. Hence, $S$ cannot reciprocally afford $T$ and no longer deserves the same amount of representation. Adding voter $v_5$ resolves this, since the group $\{v_1, v_2, v_3, v_4, v_5\}$ still commonly approves $\{c_1, c_2\}$ and has exactly the five voters required for this two-candidate reciprocal EJR claim. In contrast, reciprocal EJR does not weaken claims based on quality-$1$ candidates. For example, the group $\{v_3, v_4\}$ is $1$-large and commonly approves the quality-$1$ candidate $c_1$. Since $\frac{1}{h_{c_1}} \frac{n}{k} = 2$, the group can also reciprocally afford $\{c_1\}$ and therefore retains the same one-candidate claim under reciprocal EJR.

We analyze satisfiability and the relations between the reciprocal notions in \Cref{app:reciprocal_notions_and_relations}, finding that the relations  exactly match the relationships of their non-reciprocal versions, and that reciprocal threshold-EJR remains unsatisfiable.
Further, we define reciprocal threshold-PJR+ (cf. \Cref{def:rtPJR+}), an efficiently verifiable and axiomatically grounded strengthening of reciprocal threshold-PJR. We also provide axiomatic characterizations of all JR-level reciprocal axioms.

\subsection{Price of Representation} \label{pop}
As noted above, a natural desideratum in many application domains where candidate qualities emerge is to maximize the total quality of the selection $W$. To what extent can we achieve this while simultaneously maintaining proportionality? We formalize this question as follows. In this section, to avoid trivial counterexamples, we assume that each candidate has at least one supporter.

\begin{definition} \label{def:quality_price_of_representation}
For an instance $\mathcal I$, let $W^*(\mathcal I)$ be a committee consisting of the $k$ candidates with the highest quality.
    The \emph{quality-price of representation} of an (always satisfiable) proportionality notion $\mathcal X$ is the worst-case ratio
    $$
        \sup_{\mathcal I} \frac{h(W^\star(\mathcal I))}{ \max_{W\in \mathcal X(\mathcal I)}h(W)}
    $$
    over all instances $\mathcal I$ between a quality-optimal committee $W^\star(\mathcal I)$ and the quality of a quality-optimal $\mathcal X$-satisfying committee $W$.
    A rule $R$ provides a \emph{quality approximation} of at least $\alpha\le 1$ if $h(R(\mathcal I))\geq \alpha \cdot h(W^\star(\mathcal I))$ for all instances $\mathcal I$, where $R(\mathcal I)$ is any committee returned by the rule on instance $\mathcal I$.
\end{definition}

The quality-price of representation turns out to be large for all non-reciprocal notions we consider.

\begin{proposition}
\label{thm:unscaled-cost}
For JR, value-JR, and threshold-JR, the quality-price of representation is at least $k$.
\end{proposition}

\begin{profileproof}[8][0.48\linewidth]{%
    \setlength{\tabcolsep}{2.2pt}%
    \renewcommand{\arraystretch}{1.08}%
    \begin{tabular}{
      @{}c
      >{\columncolor{gray!20}}c
      >{\columncolor{gray!20}}c
      >{\columncolor{gray!20}}c
      >{\columncolor{gray!20}}c
      cccc@{}
    }
        \toprule
        & $a_1$
        & $a_2$
        & $\cdots$
        & $a_k$
        & $d_1$
        & $d_2$
        & $\cdots$
        & $d_k$
        \\
        quality
        & $1$
        & $1$
        & $\cdots$
        & $1$
        & $\varepsilon$
        & $\varepsilon$
        & $\cdots$
        & $\varepsilon$
        \\
        \midrule
        $v_1$
        & \checkmark
        & \checkmark
        & $\cdots$
        & \checkmark
        & \checkmark
        & & &
        \\
        $v_2$
        & & & & &
        & \checkmark
        & &
        \\
        $\vdots$
        & & & & &
        & & $\ddots$
        &
        \\
        $v_k$
        & & & & &
        & & &
        $\checkmark$\\
        \bottomrule
    \end{tabular}%
}
    Consider the profile on the right.
    Among all committees satisfying either JR, value-JR, or threshold-JR, the quality-maximizing committee is $W = \{a_1\}\cup \{d_2,\dots,d_k\}$, with $h(W) = 1 + (k-1) \varepsilon$. However, the  committee with highest summed quality is $W^\star = \{a_1,\dots,a_k\}$ with $h(W^\star) = k$. For arbitrarily small $\varepsilon$, we therefore get $\frac {h(W^\star)}{h(W)} = \frac {k}{1+ (k-1) \varepsilon} \xrightarrow{\varepsilon \to 0} k$.
\end{profileproof}

Note that when allowing for unsupported candidates, the quality-price becomes unbounded, as can be seen by removing all approvals for the candidates $a_1, \dots,a_k$ in the profile from \Cref{thm:unscaled-cost}.

Remarkably, it turns out that many of our reciprocal proportionality notions escape this incompatibility and allow for a constant-factor quality
guarantee. The intuition is that if reciprocal proportionality forces the inclusion of low-quality candidates, only a few of them will need to be included in the committee, so that we can use the remaining seats for candidates with high quality. Recall that for a given target size $k$, we say that a set of candidates $X \subseteq C$ of size at most $k$ is reciprocally affordable, if $\sum_{c\in X} \frac 1 {h_c} \le k$. We say that $X^+$ is a \emph{quality-greedy completion} of $X$, if $X\subseteq X^+$, $X^+$ is of size $k$, and $h(X^+)$ is maximal among all size-$k$ supersets of $X$. We  obtain the following result.

\begin{restatable}{theorem}{QualityPriceOfRepresentation} \label{thm:quality_approximation}
    Let $W^\star$ have optimal quality $h(W^\star)$ in a given instance $\mathcal I$.
    If $X$ is reciprocally affordable, then any quality-greedy completion $X^+$ satisfies $h(X^+) \ge \frac{3}{4} h(W^\star)$.
\end{restatable}

\begin{proof}[Proof sketch]
    Let $t^\star = \lvert X \cap W^\star \rvert$ be the number of candidates shared between the reciprocally affordable and the optimal committee. Intuitively, a large $t^\star$ results in a high approximation ratio, as some of the optimal candidates are already in the committee, even before the greedy completion. In this proof sketch, we therefore consider the $t^\star = 0$ ``extreme'' case and additionally, for ease of exposition, assume unique qualities for all candidates in $C$, i.e., $h_c\neq h_{c'}$ for all $c\neq c'\in C$.

    Let $t = |X|$. Now, as $X \cap W^\star = \emptyset$, the quality-greedy completion $X^+ \setminus X$ must contain exactly the $(k-t)$ candidates with the highest qualities from $W^\star$. These candidates need to have a total quality of at least $(k-t)$ times the average quality of $W^\star$, i.e., $h(X^+ \setminus X) \ge \frac{k-t}{k} h(W^\star)$.
    Furthermore, by the AM-HM-inequality and the reciprocal affordability of $X$, we can bound the total quality of $X$ by $h(X) = \sum_{c \in X} h_c \ge \frac{t^2}{\sum_{c \in X} \frac{1}{h_c}} \ge \frac{t^2}{k}$.
    The total quality ratio can then be expressed as
    $$
        f(t) := \frac{h(X^+)}{h(W^\star)} \ge \frac{\frac{t^2}{k} + \frac{k-t}{k} h(W^\star)}{h(W^\star)} = 1 + \frac{t^2 - t \cdot h(W^\star)}{k \cdot h(W^\star)}.
    $$
    We can find the minimum of that function at $t = \frac{h(W^\star)}{2}$, giving the final lower bound of
    $$
        f\Big(\frac{h(W^\star)}{2}\Big) \ge 1 + \frac{\big(\frac{h(W^\star)}{2}\big)^2 - \frac{h(W^\star)}{2} \cdot h(W^\star)}{k \cdot h(W^\star)} = 1 - \frac{h(W^\star)}{4k} \ge 1 - \frac{k}{4k} = \frac{3}{4}.
    $$
\end{proof}

In \Cref{app:quality-price_of_representation}, we define reciprocal versions of MES, EAR-rules and GCR. We prove that these satisfy reciprocal EJR, reciprocal threshold-PJR+ and reciprocal value-EJR, respectively, and always output quality-greedy completions of reciprocally affordable candidate sets, thus achieving the optimal $\frac{3}{4}$ quality approximation. The quality-prices of these three notions and all implied notions are therefore at most $\frac{4}{3}$. This bound is tight, as shown by the following statement (recall that all considered reciprocal notions imply  reciprocal JR):
\begin{restatable}{proposition}{recJRupperbound}\label{obs:rec_JR_upper_bound}
    The quality-price of reciprocal JR is at least $\frac{4}{3}$.
\end{restatable}

Analogously to \Cref{def:quality_price_of_representation}, we consider in \Cref{app:welfare_tradeoffs} the trade-off between proportional representation and the overall welfare of the committee, i.e., the sum $\sum_{i \in V} h(A_i \cap W)$.
Here, a notable difference between value- and threshold- notions emerges: due to its restrictive ``lexicographic'' nature, threshold-JR cannot guarantee more than a $\nicefrac{1}{k}$-fraction of the optimal welfare, even when all qualities are arbitrarily close to $1$.
In contrast,  a value-JR committee with welfare approximation $\nicefrac{h^-}{8}\cdot \nicefrac{1}{\sqrt{k}}$ always exists and can be computed in polynomial time, where $h^-$ denotes the minimum candidate quality. This suggests that while threshold notions provide stronger proportionality guarantees, these stronger guarantees come at the cost of other desiderata.

\section{Conclusion}
\label{sec:conclusion}
We initiated the study of approval-based multiwinner voting with candidate quality scores and studied notions of  proportionality in this new setting, proposing two distinct ways to define quality-aware proportionality notions through a threshold- and a value-based family of JR, PJR, and EJR adaptations. While we showed that all satisfiable threshold notions imply their value-based counterparts, their stronger demands are not always satisfiable at the EJR level. We provided rules, most notably threshold-MES and value-GCR, which achieve the strongest possible combinations of proportionality guarantees under our notions. On the negative side, we showed that computing a value-EJR committee is NP-hard, thereby also resolving open questions on the complexity of EJR in participatory budgeting with additive utilities and unit costs. As maximizing the quality of the committee is a natural desideratum in many application settings, we analyzed the trade-off between our notions and total quality. We proposed reciprocal versions of our axioms that remain compatible with quality maximization, allowing for an optimal constant-factor approximation of the total quality while maintaining representation guarantees at no additional computational cost.

There are several directions for future work. First, it would be interesting to study further quality-aware axioms for voting rules beyond quality-aware proportionality notions, e.g., capturing that increasing a candidate's quality should never remove them from the committee. Second, since candidate qualities are externally assessed, studying the robustness of rules and axioms with respect to noisy scores is well-motivated. Lastly, there are various application domains, e.g., information retrieval and participatory budgeting, in which our model could be applied and empirically studied.

\paragraph{Acknowledgements}
Parts of this work were developed with the assistance of large language models. Specifically, the proof of \Cref{prop:minRelWelfareBound} was developed in collaboration with Claude Fable 5 (Anthropic). Initial proof ideas and versions for \Cref{thm:value-EJR-hard}, \Cref{thm:value-JR_NP_Verification}, and \Cref{prop:threshold-EJR_existence_hardness} were obtained by querying ChatGPT 5.6 Sol (OpenAI); the authors verified and substantially rewrote these proofs, including drastic simplifications in the case of \Cref{thm:value-EJR-hard}. All results were checked by the authors, who take full responsibility for the content and correctness of the paper.

\newpage

\appendix
\section{Hardness of Finding a Value-EJR Committee}
In this section, we show the following theorem. We note that our construction is similar to a construction by \citet{baharav2026complexityjustifiedrepresentationadditive}. Since the original proof idea was found by ChatGPT Sol 5.6, it is unclear whether the LLM had access to and was inspired by the paper of \citeauthor{baharav2026complexityjustifiedrepresentationadditive} while finding the proof idea.
\valueEJRExistsIsHard*

\begin{proof}
    We reduce from \textsc{Balanced Biclique}: given a bipartite graph
$G=(L\cup R,E)$ and an integer $t$, decide whether $G$
contains a biclique $K_{t,t}=(I,J)$ of size $t$, which is a pair of sets $I\subseteq L$ and $J\subseteq R$ with
$|I|=|J|=t$ so that $\{u,v\}\in E$ for all $u\in I$ and $v\in J$. This problem is
NP-complete~\citep{GareyJohnson1979}. We assume without loss of generality that $2\le t\le\min\{|L|,|R|\}$. We write $N_G(v) \coloneqq \{u\in L\cup R \mid \{v,u\} \in E\}$ for the (open)
neighborhood of $v\in L\cup R$ in $G$.

\paragraph{Construction}
Given an instance $(G=(L\cup R,E), t)$ of \textsc{Balanced Biclique}, we construct a multiwinner election $f(G,t)$ with candidate qualities as follows.
Fix
\[
  B\coloneqq(t-1)\bigl(|L|-t\bigr)+1,
  \quad
  \ell^{*}\coloneqq Bt,
  \quad
  D\coloneqq\ell^{*}-t,
  \quad
  \beta\coloneqq\frac{2}{2\ell^{*}-1}.
\]
The voters of $f(G,t)$ are $V\coloneqq L\cup U$, where $U$ is a set of
$D$ \emph{filler} voters disjoint from $L$; we refer to $L$ as the \emph{vertex} voters.
The set $C$ consists of two kinds of candidates:
\begin{itemize}[nosep]
  \item for every $i\in V$, we add a \emph{private} candidate $p_i$ of quality $1$,
        approved by $i$ and by nobody else;
  \item for
        every $r\in R$, we add $B$ \emph{block} candidates $c_{r,1},\dots,c_{r,B}$ of quality $\beta$, each approved exactly by the voters in $N_G(r)\cup U$.\footnote{Note that $t\ge2$ implies $\ell^{*}\ge2$ and hence $\beta\le\tfrac23$, so all qualities lie in $(0,1]$. Since $t\leq |L|$, we have $B\leq |L|^2+1\leq |L|^3$, so $\ell^*\leq|L|^4$. Hence, encoding the numerator $2$ and the denominator $2\ell^{*}-1\le2|L|^4-1$ of $\beta$ in unary requires only polynomially many bits in the input size.}
\end{itemize}
The target size is
\[
  k\coloneqq n=|L|+D=\ell^{*}+\bigl(|L|-t\bigr).
\]
Let $W_0\coloneqq\{p_i\mid i\in V\}$ be the committee consisting of all private candidates, and observe that $|W_0|=n=k$. All voters $i\in V$ have $h(A_i\cap W_0)=1$.
Since $n=k$, a group $S\subseteq V$ is $\ell$-large exactly when $\lvert S \rvert\ge\ell$.

The idea behind the construction is that $W_0$ is a---and in fact the unique---value-EJR committee if and only if there is no $K_{t,t}$ in $G$: If there is a biclique $(I,J)$, then the vertex voters from $I$ together with all filler voters $U$ form a cohesive group with respect to the block candidates corresponding to $J$, and this group has an entitlement strictly larger than one, preventing $W_0$ from satisfying value-EJR.
Thus, to decide the given \textsc{Balanced Biclique} instance we can apply the construction, run an algorithm for our problem and return No if and only if $W_0$ is returned by our algorithm.

\begin{lemma}\label{lem:vejr-baseline}
If $G$ contains a biclique $K_{t,t}=(I,J)$, then $W_0$ violates value-EJR in
$f(G,t)$.
\end{lemma}

\begin{proof}
Set $S\coloneqq I\cup U$ and
$T\coloneqq\{c_{u,\alpha}\mid u\in J,\ \alpha\in[B]\}$. Every voter of $S$ approves every
candidate of $T$: the filler voters $U$ approve all block candidates, and each $v\in I\subseteq L$
is a neighbour of each $u\in J \subseteq R$. Hence $T\subseteq \appcut{S}$ and
$|\appcut{S}|\ge|T|=Bt=\ell^{*}$. Moreover, $\lvert S \rvert=t+D=\ell^{*}$, so $S$ is
$\ell^{*}$-large and therefore $\ell^{*}$-cohesive. Its entitlement is
\[
  h(\appcut{S}\mid\ell^{*})\ \ge\ h(T)\ =\ \ell^{*}\beta\ >\ 1,
\]
since $\ell^{*}\beta=\frac{2\ell^{*}}{2\ell^{*}-1}>1$. As $h(A_i\cap W_0)=1$ for every $i\in S$, no member of $S$ meets the entitlement, so $W_0$ violates value-EJR.
\end{proof}

\begin{lemma}\label{lem:vejr-unique}
If $G$ does not contain a biclique $K_{t,t}$, then $W_0$ is the unique value-EJR committee of
$f(G,t)$.
\end{lemma}

\begin{proof}
By \Cref{thm:value-GCR_value-EJR_PJR}, the instance $f(G,t)$ has at least one value-EJR
committee, so it suffices to show that if $G$ does not contain a $K_{t,t}$ then no committee $W\neq W_0$ satisfies value-EJR. Suppose therefore,
for contradiction, that some $W\neq W_0$ satisfies value-EJR.
Let $Y\coloneqq W\cap \{c_{r,\alpha}\mid r\in R,\ \alpha\in[B]\}$ be the block candidates in $W$
and let $O\coloneqq\{i\in V\mid p_i\notin W\}$ be the voters whose private candidate has not been selected. As $n=k$ and $|W|=k$, we get $|O|=|Y|\eqqcolon x$, with $x>0$ because
$W\neq W_0$. For $i\in V$, let
$d(i)\coloneqq|\{c_{r,\alpha}\in Y\mid i\in N_G(r)\cup U\}|$ be the number of selected
block candidates that $i$ approves.

The following argument is structured in four steps. First, we show that every voter whose private candidate has not been selected needs to approve $\ell^*$ block candidates in $W$ (Claim~1). As $W\neq W_0$, at least one such voter needs to exist, which since $k=n$ in fact implies that there are at least $\ell^*$ voters whose private candidate has not been selected (Claim~2). We can then show (Claim~3) that there are at least $t$ vertex voters among them. In Claim~4, we use the existence of these $t$ vertex voters each approving $\ell^*$ block candidates in $W$ together with the fact that $x$ exceeds $\ell^*$ only by $(|L|-t)$ to conclude that the intersection of their approval sets needs to be large. This then implies that the intersection needs to contain block candidates corresponding to at least $t$ right vertices. Since by construction a vertex voter (corresponding to a left side vertex) approves a block candidate only if they are adjacent to the corresponding right side vertex, these vertices form a biclique of size $t$.

\smallskip
\noindent\emph{Claim 1: $d(i)\ge\ell^{*}$ for every $i\in O$.}
Intuitively stated, every voter whose private candidate is not selected needs to approve sufficiently many block candidates in $W$.
Fix $i\in O$. The singleton group $S=\{i\}$ is $1$-large and $1$-cohesive as $p_i\in A_i$, so $i$ has entitlement at least $h(\{p_i\})=1$. Since $p_i\notin W$, all of $i$'s utility comes from block candidates, i.e., it needs to hold that
$\beta\cdot d(i)\ge1$. As $d(i)$ is an integer,
\[
  d(i)\ \ge\ \Bigl\lceil\tfrac1\beta\Bigr\rceil
        \ =\ \Bigl\lceil\ell^{*}-\tfrac12\Bigr\rceil
        \ =\ \ell^{*}.
\]

\medskip
\noindent\emph{Claim 2: $x\ge\ell^{*}$.}
Since at least one voter has lost their private candidate, $W$ needs to contain at least $\ell^*$ block candidates. Since $|O|=x>0$, such a voter $i\in O$ exists. By Claim~1, we have $\ell^{*}\le d(i)$. Since the block candidates in $W$ approved by $i$ are a subset of $Y$, it follows that $d(i)\le|Y|=x$, so $x\geq \ell^*$.

\medskip
\noindent\emph{Claim 3: $|O\cap L|\ge t$.}
It holds that $|O|=x\geq \ell^*$ by Claim~2. Moreover, as $V$ consists only of the (disjoint) sets of filler voters $U$ and vertex voters $L$, we have $|O\cap U|+|O\cap L|=|O|\geq \ell^*$. As there are only $D=\ell^{*}-t$ filler voters $U$, we have $|O\cap L| \geq \ell^* -|O\cap U| \geq \ell^* -(\ell^{*}-t)=t$.

\medskip
\noindent\emph{Claim 4: $G$ contains a biclique $K_{t,t}$.}
By Claims~1 and~3, any $t$ distinct vertex voters in $O\cap L$ each approve at least $\ell^{*}$
of the $x$ block candidates in $W$. Intuitively, we will now show that since $x$ is sufficiently close to $\ell^*$, these voters all need to share more than $B(t-1)$ block candidates in $W$, implying that they share $t$ neighbors in the graph.

Formally, by Claim~3, there exist $t$ distinct voters $i_1,\dots,i_t\in O\cap L$. For $a\in [t]$, denote by $Y_a\coloneqq Y\cap A_{i_a}$ the block candidates in $W$ approved by $i_a$, so that $|Y_a|=d(i_a)\ge\ell^{*}$ by Claim~1.
We start by observing that as $Y_a\subseteq Y$ we have $|Y\setminus Y_a|=x-|Y_a|$ for all $a\in [t]$, and
$Y\setminus\bigcap_{a\in [t]} Y_a=\bigcup_{a\in [t]}(Y\setminus Y_a)$, so by the union bound
\[
  x-\Bigl|\bigcap_{a\in [t]}Y_a\Bigr|
  \ =\ \Bigl|\bigcup_{a\in [t]}(Y\setminus Y_a)\Bigr|
  \ \le\ \sum_{a\in [t]}\bigl(x-|Y_a|\bigr)
  \ =\ t\,x-\sum_{a\in [t]}|Y_a|,
\] which by rearranging gives the following lower bound on the intersection, $$\Bigl|\bigcap_{a\in [t]}Y_a\Bigr|\geq \sum_{a\in [t]}|Y_a|-(t-1)\,x.$$
We can further lower bound the right hand side, using $x=|O|\le|V|=|L|+D=\ell^{*}-t+|L|$,
$$\sum_{a\in [t]}|Y_a|-(t-1)\,x\ \ge\ \sum_{a\in [t]}|Y_a|-(t-1)\bigl(\ell^{*}-t+|L|\bigr).$$
Thus, as $\sum_{a\in [t]}|Y_a|\geq t\ell^*$, we have
\[
  \Bigl|\bigcap_{a\in [t]}Y_a\Bigr|
  \ \ge\ t\ell^{*}-(t-1)\bigl(\ell^{*}-t+|L|\bigr)
  \ =\ \ell^{*}-(t-1)\bigl(|L|-t\bigr)
  \ =\ Bt - (B-1)
  \ =\ B(t-1)+1.
\]
Note that each $r\in R$ contributes at most $B$ candidates to $Y$, while $\bigcap_{a\in [t]}Y_a\subseteq Y$ contains more than $B(t-1)$ candidates. By the pigeonhole principle, there are $t$ distinct vertices $r_1,\dots,r_t\in R$ such that for each $b\in [t]$ there is an $\alpha_b\in[B]$ with $c_{r_b,\alpha_b}\in\bigcap_{a\in [t]}Y_a$. In particular, for every $a\in [t]$, $c_{r_b,\alpha_b}\in Y_a$ and hence $c_{r_b,\alpha_b}$ is approved by $i_a$, which by construction implies $i_a\in N_G(r_b)$. Hence $\{i_a,r_b\}\in E$ for all $a,b\in[t]$, so $\{i_1,\dots,i_t\}$ and $\{r_1,\dots,r_t\}$ form a biclique $K_{t,t}$ of $G$, contradicting the assumption that $G$ contains no biclique of size $t$.
\end{proof}

\noindent

Assume now there is a polynomial-time algorithm that constructs a committee satisfying value-EJR for a given multiwinner election with quality scores.
Then, we can decide \textsc{Balanced Biclique} in polynomial time as follows.
Given an instance $(G,t)$ of \textsc{Balanced Biclique}, we construct the instance $f(G,t)$ in polynomial time and run our polynomial-time algorithm returning a committee $W$ on $f(G,t)$.
Answer
\textsc{yes} if and only if $W\neq W_0$. If $G$ has no $K_{t,t}$, then $W_0$ is
the only value-EJR committee by \Cref{lem:vejr-unique}, so $W=W_0$ and we answer
\textsc{no}. If $G$ has a $K_{t,t}$, then $W_0$ is not
value-EJR by \Cref{lem:vejr-baseline}, so $W\neq W_0$ and we answer
\textsc{yes}.
Since \textsc{Balanced Biclique} is NP-complete, the existence of such a polynomial-time algorithm implies $\mathrm{P}=\mathrm{NP}$.
\end{proof}

\section{Axiomatic Proportionality}\label{app:axiomatic_proportionality}

\subsection{Justified Representation}
We start by providing the characterizations of our JR-level notions.

\paragraph{Non-zero Representation} When there are no quality scores, a minimal requirement on party list instances is that every sufficiently large voter group is represented by some candidate. Formally, \emph{weak representation} \citep{delemazure2023strategyproofness} requires for every party list instance that every group of voters with $\lvert V_j\rvert \ge \frac{n}{k}$ and $C_j\neq \emptyset$ is represented by at least one candidate.

A natural way to include qualities is to choose these candidates best-in-slot: we say that a committee $W$ satisfies \emph{weak best-in-slot representation}, if  $W\cap \arg\max\{ h_c\mid c\in C_j \}\neq \emptyset$ whenever $\lvert V_j\rvert \ge \frac{n}{k}$ and $C_j\neq \emptyset$. From the voters' perspective, it may be that only the sum of scores matters. Formally, $W$ satisfies \emph{weak total representation}, if $h(W\cap C_j) \ge \max_{c\in C_j} h_c$ whenever $\lvert V_j\rvert \ge \frac{n}{k}$ and $C_j\neq \emptyset$.

We say that a proportionality notion $X$ satisfies weak (best-in-slot or total) representation, if all $W\in X(\mathcal I)$ satisfy weak (best-in-slot or total) representation whenever $\mathcal I$ is a party-list instance.
Clearly, weak best-in-slot representation implies weak total representation implies weak representation.
Indeed, our JR-level definitions can easily be seen to meet these requirements.
\begin{observation}
     JR satisfies weak representation, value-JR satisfies weak total representation, and
    threshold-JR satisfies weak best-in-slot representation.
\end{observation}

While there are many notions satisfying the weak representation axioms, remarkably, the JR-level definitions are canonical among these when we require the well-behavedness conditions used by \citet{DoPe26a}.

\begin{itemize}
    \item \emph{Independence of losers} requires that if $W$ is proportional and we remove some $c\notin W$ from the instance, then $W$ remains proportional.
    \item \emph{Robustness to fully satisfied voters} requires that if $W$ is proportional and some voter changes their ballot from $A_i$ to a subset of $A_i\cap W$, then $W$ remains proportional.
    \item \emph{Monotonicity} requires that if $W$ is proportional and some voter changes their ballot from $A_i$ to $A_i\cup \{c\}$ with $c\in W$, then $W$ remains proportional.
\end{itemize}

Indeed, with these three natural well-behavedness conditions, we see that the weak representation axioms can canonically be extended to correspond to our JR-level definitions.

\begin{theorem}
    Let $X$ be a proportionality notion satisfying independence of losers, robustness to fully satisfied voters, and monotonicity. The following statements hold.
    \begin{itemize}
        \item If $X$ satisfies weak representation, then $X\subseteq \JR$.
        \item If $X$ satisfies weak total representation, then $X\subseteq \valueJR$.
        \item If $X$ satisfies weak best-in-slot representation, then $X\subseteq \thresholdJR$.
    \end{itemize}
    Further, \JR, \valueJR, and \thresholdJR{} satisfy independence of losers, robustness to fully satisfied voters, and monotonicity, as well as the respective weak representation axiom.
\end{theorem}
\begin{proof}
It is easy to verify that $\JR, \valueJR, \thresholdJR$ satisfy the claimed properties.
We prove the remaining statements by contraposition. Let $X$ satisfy all three well-behavedness axioms. Our goal is to show that if $X$ does not refine one of the JR-level notions, then it violates the corresponding weak representation axiom.
\begin{itemize}
    \item First, let $W\in X(\mathcal I) \setminus \JR(\mathcal I)$. Then, there is $S\subseteq V$ that is $\ell$-cohesive such that $\appjoin{S} \cap W= \emptyset$. Take any $c\in \appcut{S}\setminus W$. Apply independence of losers to remove all candidates in $C\setminus W$ from the instance, except for $c$. Further, apply robustness to fully satisfied voters to let each $i\in V\setminus S$ approve of $A_i\cap W$ (i.e., potentially removing the approval of $c$), and then apply monotonicity so that $A_i\cap W$ becomes $W$ for all $i\in V\setminus S$. In total, we arrived at  a party-list profile in which $V_1 = S$, $V_2 = V\setminus S$. By assumption, weak representation is violated for $W$ in this resulting party-list instance $\mathcal I^*$. However, by our operations, $W\in X(\mathcal I^*)$.

    \item Next, let $W\in X(\mathcal I) \setminus \valueJR(\mathcal I)$. Then, there is $S\subseteq V$ that is $\ell$-cohesive such that $h(\appjoin{S} \cap W) < h(\appcut{S}\mid 1)$. Take any $c\in \arg\max\{h_d\mid d\in \appcut{S}\}$, then $c\notin W$. We proceed with similar operations as for JR, but one extra step inbetween: apply independence of losers to remove all candidates in $C\setminus W$ from the instance, except for $c$;
    monotonicity to let all voters in $S$ approve of $(\appjoin{S} \cap W) \cup\{c\}$;
    robustness to fully satisfied voters and monotonicity to let each $i\in V\setminus S$ approve of precisely $W \setminus (\appjoin{S} \cap W)$. In the resulting party-list instance $\mathcal I^*$ with $V_1 = S$, $C_1 = (\appjoin{S} \cap W) \cup\{c\}$, $V_2 = V\setminus S$, $C_2 = W \setminus (\appjoin{S} \cap W)$ we still have $h(W \cap C_1)= h(\appjoin{S} \cap W) < h(\appcut{S}\mid 1) \le h_{\mathcal I^*}(C_1\mid 1)$.
    Therefore, on instance $\mathcal I^*$, the committee $W$ violates weak total representation. However, since we only applied well-behaved profile changes, $W\in X(\mathcal I^*)$.

    \item Finally, let $W\in X(\mathcal I) \setminus \thresholdJR(\mathcal I)$. Then, there is $S\subseteq V$ and $\tau \in (0,1]$ such that $S$ is $\ell$-large and $\lvert \appcut{S}^{\ge \tau}\rvert \ge 1$, but $\lvert \appjoin{S} \cap W^{\ge \tau} \rvert = 0$.
    Take any $c\in \appcut{S}^{\ge \tau}\setminus W$, which exists as $\appcut{S}^{\ge \tau}\neq \emptyset$ and $\appcut{S}^{\ge \tau} \cap W \subseteq \appjoin{S} \cap W^{\ge \tau} = \emptyset$. We proceed with similar operations as for value-JR: apply independence of losers to remove all candidates in $C\setminus W$ from the instance, except for $c$;
    monotonicity to let all voters in $S$ approve of $(\appjoin{S} \cap W) \cup\{c\}$;
    robustness to fully satisfied voters and monotonicity to let each $i\in V\setminus S$ approve of precisely $W\setminus (\appjoin{S} \cap W)$. In the resulting party-list instance $\mathcal I^*$ with $V_1 = S$, $C_1 = (\appjoin{S} \cap W) \cup\{c\}$, $V_2 = V\setminus S$, $C_2 = W\setminus (\appjoin{S} \cap W)$ we still have $(W \cap C_1)^{\ge \tau} = \appjoin{S} \cap W^{\ge \tau} =\emptyset$.
    Therefore, on instance $\mathcal I^*$, the committee $W$ violates weak best-in-slot representation. However, since we only applied well-behaved profile changes, $W\in X(\mathcal I^*)$.\qedhere
\end{itemize}
\end{proof}

\paragraph{Proportional Representation} Without quality scores given, a broadly accepted approach to party list instances is \emph{lower quota}, requiring that each set of voters $V_j$ deserves representation by at least $q_j = \min(\lfloor \frac{\lvert V_j \rvert }{n} k \rfloor, \lvert C_j\rvert)$ candidates. We present two natural extensions of this:
firstly, we require that if $V_j$ is eligible for $\ell$ candidates, these are chosen best-in-slot. Formally, we say that a committee $W$ satisfies \emph{best-in-slot lower quota}, if for each $
V_j$, there exists $S_j \subseteq W \cap C_j$ of size $q_j$, such that no candidate in $C_j \setminus S_j$ has a strictly higher quality than any candidate in $S_j$.
Secondly, we can instead focus on the voter-utilities: if $V_j$ is eligible for $\ell$ candidates, the voters should be at least compensated with the achievable utility of $\ell$ candidates. Formally, we say that a committee $W$ satisfies \emph{total lower quota}, if for each $V_j$, we have that $h(W\cap C_j)\ge h(C_j\mid q_j)$.

\begin{observation}
Classic PJR and EJR satisfy lower quota. Value-PJR and -EJR satisfy total lower quota. Threshold-PJR and -EJR satisfy best-in-slot lower quota.
\end{observation}

As \citet{DoPe26a} have shown, the three well-behavedness axioms combined with lower quota characterize a refinement of PJR called PJR+ \citep{BrPe23a}.
We generalize this result:
\begin{theorem}
    Let $X$ satisfy independence of losers, robustness to fully satisfied voters, and monotonicity. If $X$ satisfies best-in-slot lower quota, $X$ refines threshold-PJR+. Vice versa, threshold-PJR+ satisfies all these axioms.
\end{theorem}
\begin{proof}
    It is easy to verify that $\thresholdPJR+$ satisfies the claimed properties.
We prove the other direction by contraposition. Let $X$ satisfy all three well-behavedness axioms. Our goal is to show that if $X$ does not refine threshold-PJR+, then it violates best-in-slot lower quota.
    Let $W\in X(\mathcal I) \setminus \thresholdPJR+(\mathcal I)$. Then, there is $S\subseteq V$ and $\tau \in (0,1]$ such that $S$ is $\ell$-large and $\lvert \appcut{S}^{\ge \tau}\setminus W\rvert \ge 1$, but $\lvert \appjoin{S} \cap W^{\ge \tau} \rvert <  \ell$.
    Take any $c\in \appcut{S}^{\ge \tau}\setminus W$. Since $\tau\le h_c$ and increasing $\tau$ makes $\appjoin{S} \cap W^{\ge \tau}$ weakly smaller, wlog choose $\tau = h_c$.
    We proceed with the exact same operations as for threshold-JR: apply independence of losers to remove all candidates in $C\setminus W$ from the instance, except for $c$;
    monotonicity to let all voters in $S$ approve of $(\appjoin{S} \cap W) \cup\{c\}$;
    robustness to fully satisfied voters and monotonicity to let each $i\in V\setminus S$ approve of precisely $W\setminus (\appjoin{S} \cap W)$. In the resulting party-list instance $\mathcal I^*$ with $V_1 = S$, $C_1 = (\appjoin{S} \cap W) \cup\{c\}$, $V_2 = V\setminus S$, $C_2 = W\setminus (\appjoin{S} \cap W)$ we still have $(W \cap C_1)^{\ge \tau} = \appjoin{S} \cap W^{\ge \tau}$, which contains strictly less than $\ell$ candidates.
    Therefore, on instance $\mathcal I^*$, the committee $W$ violates best-in-slot lower quota, as $c$ is not chosen. However, since we only applied well-behaved profile changes, $W\in X(\mathcal I^*)$.\qedhere
\end{proof}

\section{Hardness of Value-JR Verification}\label{sec:hardness_value-JR_verification}

We prove the following hardness result.
\begin{theorem}\label{thm:value-JR_NP_Verification}
    Given an instance $\mathcal{I}=(A,k,h)$ and a committee $W$, deciding
    whether $W$ satisfies value-JR is coNP-complete. This holds even if
    there are only two distinct candidate qualities.
\end{theorem}

\begin{proof}
    Note that a violation exists iff there exists $S\subseteq V$ with $\lvert S \rvert\geq\frac{n}{k}, \appcut{S}\neq\varnothing, h(\appjoin{S} \cap W)<h(\appcut{S}\mid1)$. Therefore, a valid certificate for $W$ violating value-JR is a set of voters $S$ and a candidate $c\in \appcut{S}$ maximizing $h$, and the problem lies in coNP.

    For the converse direction, we reduce from the following problem: \textsc{Minimum $t$-Union} takes as input a finite universe $U$, a collection $\mathcal{B}=\{B_1,\dots,B_r\}\subseteq 2^U$ of subsets of $U$, the target integer $t$ and the total, integral budget $b\ge 1$. The question to answer is whether there exists a selection of $t$ sets from $\mathcal B$ with at most $b$ elements in total. Formally, a yes instance is given if there exist $J\subseteq[r]$ of size $t$ such that $\lvert \bigcup_{j\in J}B_j\rvert \leq b$. \citet{chlamtac2018densest} show that this problem is NP-hard.

    We construct a voting instance from a \textsc{Minimum $t$-Union} instance $(U, \mathcal{B}, t, b)$. First, set $k = \max\left\{|U|,\left\lceil\frac rt\right\rceil\right\}$ and $n = tk$. By definition, $\frac{n}{k}= t$ and $n\ge r$.
    Further, we set $C = X \cup F\cup \{c^*\}$, with $X= \{x_u\mid u\in U\}$ being the \emph{universe candidates}, $F= \{f_1,\dots,f_{k-|U|}\}$ being the \emph{filler candidates}, and $c^*$ being the \emph{potential witness candidate}. The two occurring qualities are $1$ and $\beta \coloneqq \frac{2}{2b+1}\le \frac 23$, with $h_{c^*}\coloneqq 1$ and $h_c = \beta$ for all $c\in C \setminus \{c^*\}$. We partition the voter set into $r$ \emph{set-voters} and $n-r$ \emph{dummy voters}. For the $r$ set-voters, let voter $i$ correspond to the set $B_i$ and additionally approve the potential witness candidate, i.e., $A_i = \{c^*\}\cup\{x_u\mid u\in B_i\}$, whereas the $n-r$ dummy voters submit $A_i = \emptyset$. The committee we will check for value-JR is $W= F\cup X$. Indeed, since $\lvert F \rvert= k-\lvert U\rvert$, we have $\lvert W \rvert = k$.

    We claim that \textsc{Minimum $t$-Union} has answer yes if and only if in the corresponding instance, $W$ violates value-JR.

    ``$\Longrightarrow$'' First, let the answer to the problem be yes. Then, there exist $J\subseteq [r]$ of size $t$ with $\lvert \bigcup_{j\in J} B_j \rvert\le b$. Let $S$ be the set of corresponding set-voters. Then, $S$ has cardinality $t= \frac nk$, and is cohesive w.r.t. $c^*\notin W$. We have $h(C_{\cap S} \mid 1) = 1$ and $\appjoin{S} \cap W=\left\{x_u\;\middle|\;u\in\bigcup_{j\in J}B_j\right\}$.
    Since each $x_u$ has quality $\beta$ and the size of $\appjoin{S} \cap W$ corresponds to the size of $\bigcup_{j\in J}B_j$ which is at most $b$, we have in total that $h(\appjoin{S} \cap W) \le b \beta < 1$. Therefore, $W$ indeed violates value-JR.

    ``$\Longleftarrow$'' Vice versa, let $W$ violate value-JR in the corresponding instance.
    Then, we realize that since $C\setminus W= \{c^*\}$, the $1$-cohesive set of voters $S$ witnessing the violation must be witnessing it due to $c^*$, i.e. $S$ consists only of set-voters and we have $h(\appjoin{S} \cap W)<h(\appcut{S}\mid 1) = h(c^*) = 1$. Let $J\subseteq [r]$ be the set of indices such that each set voter in $S$ corresponds to some index in $J$. Then, it holds that $\appjoin{S} \cap W=\{x_u\mid u\in\bigcup_{j\in J}B_j\}$. Therefore, $\frac{2}{2b+1}\left|\bigcup_{j\in J}B_j\right|=h(\appjoin{S} \cap W)<1$. This is only possible if $\lvert\bigcup_{j\in J}B_j \rvert\le b$. Since $S$ has size at least $\frac nk = t$, choose any subset $J^*$ of $J$ of size $t$. Then, the \textsc{Minimum $t$-union} problem has answer yes due to $J^*$.

    NP-hardness of finding a violation now implies coNP-hardness of verification.
\end{proof}

\section{Axiom Relationships}\label{app:relations}

\subsection{Threshold- and Value-JR Notions}

\begin{theorem}\label{thm:threshold_ejr_not_value_ejr}
  Threshold-EJR and value-EJR are incompatible, even on instances on which threshold-EJR is satisfiable.
\end{theorem}

\begin{proof}
    Consider the instance in \Cref{tab:threshold_ejr_unique_no_value} with $n = 10$ voters, $k = 5$, and $m = 8$ candidates with $h_a = h_{e_1} = h_{e_2} = h_{e_3} = 1$ and $h_b = h_{f_1} = h_{f_2} = h_{f_3} = 0.5$. Since $\frac{n}{k}  = 2$, a group is $\ell$-large iff it has at least $2\ell$ voters; in particular singletons raise no claims. Write $S = \{v_1, v_2, v_3, v_4\}$, let $B_1 = \{v_5, v_6\}$, $B_2 = \{v_7, v_8\}$, $B_3 = \{v_9, v_{10}\}$ with union $B = B_1 \cup B_2 \cup B_3$, and set $W = \{e_1, e_2, e_3, f_1, f_2\}$. The five ballot types are $\{a, b, e_1\}$ (voters $v_1$ to $v_3$), $\{a, b, f_1, f_2\}$ (voter $v_4$), and $\{e_j, f_1, f_2, f_3\}$ (for $B_j$) for $j \in [3]$.

    We first show that no other committee than $W$ can satisfy threshold-EJR. For that, let $W'$ be a threshold-EJR committee.
    For each $j \in [3]$ the group $B_j$ is $1$-large and both members approve $e_j$, the only quality-$1$ candidate on their ballot. The claim with $\ell = 1$ at level $\tau = 1$ can only be met if some member has an approved committee candidate of quality $1$, i.e. if $e_j \in W'$. Hence $\{e_1, e_2, e_3\} \subseteq W'$.
    The union $B$ is $3$-large, and commonly approves $\appcut{B} = \{f_1, f_2, f_3\}$, all of quality $0.5$. The claim with $\ell = 3$ at level $\tau = 0.5$, thus requires a voter $i \in B$ with $|A_i \cap W'| \geq 3$. Since each of these voters already has one approved candidate of quality at least $0.5$ in the committee, this forces $\lvert \{f_1, f_2, f_3\} \cap W' \rvert \ge 2$.
    Since $|W'| = 5$, it contains \emph{exactly} two of these candidates and neither $a$ nor $b$.
    Finally, $S$ is $2$-large with $\bigcap_{i \in S} A_i = \{a, b\}$, both of quality at least $0.5$. The claim with $\ell = 2$ at level $\tau = 0.5$ requires a voter $i \in S$ with $|A_i \cap W'| \geq 2$. Voters $v_1$ to $v_3$ only approve $e_1$ from the committee $W'$, thus the witness must be $v_4$, giving $\{f_1, f_2\} \subseteq W'$. This yields $W' = \{e_1, e_2, e_3, f_1, f_2\} = W$.

    Next, we argue that $W$ does indeed satisfy threshold-EJR. Only the qualities $1$ and $0.5$ occur, so it suffices to verify all claims at these thresholds.

    \emph{Level $\tau = 1$.} The only approval sets of size $\geq 2$ are those of $v_1, v_2, v_3$, who are not $2$-large together, so every level-$1$ claim has $\ell = 1$. Each $1$-large group with a commonly approved quality-$1$ candidate contains a member approving some candidate in $W^1$. If the common candidate is $e_j$, every member approves $e_j \in W^1$. If it is $a$, the group lies in $S$ and, being $1$-large (size $\geq 2$), contains some voter who approves $e_1 \in W^1$. Thus all level-$1$ claims are met.

    \emph{Level $\tau = 0.5$.} All candidates have qualities at least $0.5$, so this reduces to checking EJR without considering the qualities. We enumerate the possible common approval sets of a group $G$ of voters, each being an intersection of ballot types. The non-empty possibilities are the following.
    \begin{itemize}
        \item $G \subseteq S$ and $2 \le \lvert G \rvert \le 3$: then $\ell = 1$, met by $e_1 \in W$.
        \item $G = S$: then $\ell = 2$, met by $v_4$ via $\{f_1, f_2\} \subset W$.
        \item $G \subseteq \{v_1, v_2, v_3, v_5, v_6\}$: then $\ell = 1$, met by $e_1 \in W$.
        \item $G \subseteq \{v_5, v_6, v_7, v_8, v_9, v_{10}\}$: then $\ell \leq 3$, met by any voter in $G$, who approves one of $e_1, e_2, e_3 \in W$ and $\{f_1, f_2\} \subseteq W$.
        \item $G \subseteq \{v_4, v_5, v_6, v_7, v_8, v_9, v_{10}\}$ and $v_4 \in G$: then $\ell \le 2$, met by any voter in $G$ by $\{f_1, f_2\} \subseteq W$.
    \end{itemize}
    Since every possible claim at every threshold is met, $W$ satisfies threshold-EJR.

    It remains to show that $W$ fails to satisfy value-EJR.
    $S$ is $2$-cohesive, over $\appcut{S} = \{a, b\}$, which gives the demand $h_a + h_b = 1.5$. However, $h(A_{v_i} \cap W) = h_{e_1} = 1 < 1.5$ for $i \in \{1, 2, 3\}$ and $h(A_{v_4} \cap W) = h_{f_1} + h_{f_2} = 1 < 1.5$, thus $W$ violates value-EJR.
\end{proof}

\begin{table}
  \caption{Instance for \Cref{thm:threshold_ejr_not_value_ejr} ($n = 10$,
    $k = 5$). The unique threshold-EJR committee
    $W = \{e_1, e_2, e_3, f_1, f_2\}$ (shaded) violates value-EJR for the $2$-cohesive
  group $S = \{v_1, \dots, v_4\}$ with $\appcut{S} = \{a, b\}$.}
  \label{tab:threshold_ejr_unique_no_value}
  \centering
  \begin{tabular}{c c c >{\columncolor{gray!20}} c >{\columncolor{gray!20}} c >{\columncolor{gray!20}} c >{\columncolor{gray!20}}c >{\columncolor{gray!20}}c c}
    \toprule
    $k = 5$ & $a$ & $b$ & $e_1$ & $e_2$ & $e_3$ & $f_1$ & $f_2$ & $f_3$\\
    quality & $1$ & $0.5$ & $1$ & $1$ & $1$
    & $0.5$ & $0.5$ & $0.5$\\
    \midrule
    $v_1$    & \checkmark & \checkmark & \checkmark & & & & & \\
    $v_2$    & \checkmark & \checkmark & \checkmark & & & & & \\
    $v_3$    & \checkmark & \checkmark & \checkmark & & & & & \\
    $v_4$    & \checkmark & \checkmark & & & & \checkmark & \checkmark & \\
    $v_5$    & & & \checkmark & & & \checkmark & \checkmark & \checkmark \\
    $v_6$    & & & \checkmark & & & \checkmark & \checkmark & \checkmark \\
    $v_7$    & & & & \checkmark & & \checkmark & \checkmark & \checkmark \\
    $v_8$    & & & & \checkmark & & \checkmark & \checkmark & \checkmark \\
    $v_9$    & & & & & \checkmark & \checkmark & \checkmark & \checkmark \\
    $v_{10}$ & & & & & \checkmark & \checkmark & \checkmark & \checkmark \\
    \bottomrule
  \end{tabular}
\end{table}

\begin{proposition}\label{prop:tPJR_implies_vPJR}
  Threshold-PJR implies value-PJR.
\end{proposition}

\begin{proof}
    Let $W$ be a threshold-PJR committee.
    Consider any group of voters $S$ that is $\ell$-cohesive, commonly approving candidates $\appcut{S}$ and consider a set $$T^\star \in \argmax_{\substack{T\subseteq \appcut{S}\\ |T| = \ell}} \,\, \sum_{c \in T} h_c$$
    of $\ell$ commonly approved candidates with the highest quality. Define $c_1, \dots, c_\ell$ as the candidates in $T^\star$ in order of decreasing quality. Let $W_S=\appjoin{S} \cap W$. Note that threshold-PJR with $\tau = 0$ implies that $|W_S|\geq \ell$ as $S$ is $\ell$-cohesive. Again, define $d_1, \dots, d_\ell$ as the first $\ell$ candidates in $W_S$ in order of decreasing quality. In the following, we show that $h_{d_x}\geq h_{c_x}$ for all $x \in [\ell]$.
    Let $x \in [\ell]$. Note that $S$ is $x$-cohesive for the quality threshold $\tau_x = h_{c_x}$. Thus, threshold-PJR implies that $|\appjoin{S} \cap W^{\ge \tau_x}| \ge x$. Therefore, the candidate with the $x$-th highest quality in $W_S$ needs to have quality at least $h_{c_x}$, so $h_{d_x} \ge h_{c_x}$. Since this holds for all $x \in [\ell]$, we have that
    $$
        h(\appjoin{S} \cap W)
        \ge \sum_{x \in [\ell]} h_{d_x}
        \geq \sum_{x \in [\ell]} h_{c_x} =
        \sum_{c \in T^\star} h_{c}
        = h(\appcut{S} \mid \ell)
    $$
    which completes the proof.
\end{proof}

The above proof works in particular for the special case of $\ell = 1$, giving the following corollary.

\begin{corollary}\label{cor:tJR_implies_vJR}
  Threshold-JR implies value-JR.
\end{corollary}

The incompatibility of value- and threshold-EJR can be strengthened to value-EJR and threshold-JR.

\begin{proposition}\label{prop:vEJR_tJR_incompatible}
  Value-EJR and threshold-JR are incompatible.
\end{proposition}
\begin{table}[H]
    \caption{Example instance showing that value-EJR and threshold-JR are incompatible. The only threshold-JR committee is $\{c_1, c_2\}$ (choose threshold $\tau = 1$), which violates value-EJR for the voter group $\{v_1, v_2\}$.}
    \centering
    \label{tab:vEJR_tJR_incompatible}
    \begin{tabular}{c >{\columncolor{gray!20}} c >{\columncolor{gray!20}}  c c c}
        \toprule
        $k = 2$ & $c_1$ & $c_2$ & $d_1$ & $d_2$ \\
        quality & $1$ & $1$ & $0.75$ & $0.75$\\
        \midrule
        $v_1$ & \checkmark &  & \checkmark & \checkmark \\
        $v_2$ & & \checkmark  & \checkmark & \checkmark \\
        \bottomrule
    \end{tabular}
\end{table}

\begin{proof}
      Consider the instance from \Cref{tab:vEJR_tJR_incompatible} with $k = 2$, voters $V = \{v_1, v_2\}$, and candidates $C = \{c_1, c_2, d_1, d_2\}$ where $h_{c_i} = 1$ for $i \in [2]$ and $h_{d_j} = 0.75$ for $j \in [2]$. Let $A_i = \{c_i\} \cup \{d_1, d_2\}$ for each $i \in [2]$.

      Any threshold-JR committee $W$ must contain $c_i$ for every $i \in [2]$, since $\{v_i\}$ is $1$-large (as $\frac{n}{k} = 1$) and approves only $c_i$ at quality threshold $1$. Hence $W = \{c_1, c_2\}$ is the unique threshold-JR committee.
      However, $V$ is $2$-cohesive at $\{d_1, d_2\}$, so value-EJR demands that some voter $i \in V$ has $h(A_i \cap W) \ge h_{d_1} + h_{d_2} = 1.5$. Yet each voter $v_i$ has $h(A_i \cap W) = 1 < 1.5$, violating value-EJR.
\end{proof}

\subsection{Threshold- and Basic JR Notions}

It is easy to see that all threshold-JR notions imply the corresponding JR notions, as choosing the threshold $\tau = 0$ gives exactly the corresponding definition.

\begin{observation}\label{obs:threshold_implies_basic}
    Threshold-EJR implies EJR, threshold-PJR implies PJR, and threshold-JR implies JR.
\end{observation}

Clearly, the other direction does not hold, in fact we can show that EJR is incompatible with threshold-JR.

\begin{proposition}\label{prop:EJR_tJR_incompatible}
  EJR is incompatible with threshold-JR.
\end{proposition}

\begin{table}[H]
  \caption{Example instance showing that threshold-JR and EJR (also value-JR and EJR) are incompatible. The only threshold-JR (value-JR) committee is $\{c_1, c_2\}$, which violates EJR for the voter group $\{v_1, v_2\}$.}
  \label{tab:tJR_EJR_incompatible}
  \centering
  \begin{tabular}{c >{\columncolor{gray!20}} c >{\columncolor{gray!20}} c c c c}
    \toprule
    $k = 2$ & $c_1$ & $c_2$ & $d_1$ & $d_2$ \\
    quality & $1$ & $1$ & $0.1$ & $0.1$ \\
    \midrule
    $v_1$ & \checkmark & & \checkmark & \checkmark \\
    $v_2$ & & \checkmark & \checkmark & \checkmark \\
    \bottomrule
  \end{tabular}
\end{table}

\begin{proof}
    Consider the instance from \Cref{tab:tJR_EJR_incompatible} with $k = 2$, voters $V = \{v_1, v_2\}$, and candidates $C = \{c_1, c_2, d_1, d_2\}$ where $h_{c_1} = h_{c_2} = 1$ and $h_{d_1} = h_{d_2} = 0.1$. Let $A_i = \{c_i\} \cup \{d_1, d_2\}$ for each $i \in [2]$.

    Any threshold-JR committee $W$ must contain $c_i$ for every $i \in [2]$, since $\{v_i\}$ is $1$-large (as $\frac{n}{k} = 1$) and approves only $c_i$ at quality threshold $1$. Hence $W = \{c_1, c_2\}$ is the unique threshold-JR committee.
    However, $V$ is $2$-cohesive at $\{d_1, d_2\}$, so EJR demands that some voter $i \in V$ has $|A_i \cap W| \geq 2$. Therefore, $W$ violates EJR.
\end{proof}

\subsection{Value- and Basic JR Notions}

\begin{proposition}\label{prop:vEJR_PJR}
    Value-EJR does not imply PJR.
\end{proposition}

\begin{table}[H]
    \caption{Example instance showing that value-EJR does not imply PJR. The committee $\{c_3, c_4, c_5\}$ satisfies value-EJR, but fails PJR for the voter group $\{v_1, v_2, v_3\}$ which is $2$-large and $2$-cohesive.}
    \label{tab:vEJR_PJR}
    \centering
    \begin{tabular}{c c c >{\columncolor{gray!20}} c >{\columncolor{gray!20}} c >{\columncolor{gray!20}} c }
        \toprule
        $k = 3$ & $c_1$ & $c_2$ & $c_3$ & $c_4$ & $c_5$\\
        quality & $0.5$ & $0.5$ & $1$ & $1$ & $1$\\
        \midrule
        $v_1$ & \checkmark & \checkmark & \checkmark & & \\
        $v_2$ & \checkmark & \checkmark & \checkmark & & \\
        $v_3$ & \checkmark & \checkmark & & & \\
        $v_4$ & & & & \checkmark & \checkmark \\
        \bottomrule
    \end{tabular}
\end{table}

\begin{proof}
    Consider the instance from \Cref{tab:vEJR_PJR} with $k = 3$, voters $v_1$ to $v_4$, and candidates $c_1$ to $c_5$ where $h_{c_1} = h_{c_2} = 0.5$ and $h_{c_3} = h_{c_4} = h_{c_5} = 1$. Let $A_1 = A_2 = \{c_1, c_2, c_3\}$, $A_3 = \{c_1, c_2\}$ and $A_4 = \{c_4, c_5\}$.

    The committee $W = \{c_3, c_4, c_5\}$ satisfies value-EJR. To verify this, for each $\ell \in [k]$ we check all $\ell$-large, $\ell$-cohesive voter sets. Any $1$-large, $1$-cohesive group $S$ of voters needs to be of size at least $\frac{n}{k} = \frac{4}{3} > 1$ and requires one voter to have utility at least $\max_{c \in \appcut{S}} h_c \le 1$. This is the case, since every size-$2$ group of voters has one voter $i$ with $h(A_i \cap W) \ge 1$. The only $2$-large, $2$-cohesive group of voters is $S = \{v_1, v_2, v_3\}$ commonly approving $T = \{c_1, c_2\}$ and voter $v_1$ has utility $h(A_1 \cap W) = 1 \ge h_{c_1} + h_{c_2} = 1$. There is no $3$-cohesive group of voters. Therefore $W$ satisfies value-EJR.

    However, $W$ does not satisfy PJR since the voter group $S = \{v_1, v_2, v_3\}$ is $2$-large and $2$-cohesive, but $|W \cap \appjoin{S}| = |W \cap \{c_1, c_2, c_3\}| = |\{c_3\}| = 1 < 2$.
\end{proof}

Although value-EJR and PJR do not imply each other, we can show that they are compatible, i.e., there always exists a committee satisfying both value-EJR and PJR.

\begin{proposition}\label{prop:vEJR_PJR_compatible}
  Value-EJR and PJR are simultaneously satisfiable in every profile.
\end{proposition}

We show this in \Cref{app:value-GCR_value-EJR_PJR} by defining a rule that always satisfies both value-EJR and standard PJR.

\begin{observation}\label{obs:JR_vJR}
  Value-JR implies JR.
\end{observation}
\begin{proof}
    Let $W$ be a committee that satisfies value-JR. Then for every $1$-large, $1$-cohesive group of voters $S$ we have $h(\appjoin{S} \cap W) = \sum_{c \in \appjoin{S} \cap W} h_c \ge \max_{c \in \appcut{S}} h_c > 0$ and thus $\appjoin{S} \cap W \neq \emptyset$.
\end{proof}

\begin{proposition}\label{prop:EJR_vJR_incompatible}
  EJR is incompatible with value-JR.
\end{proposition}

\begin{proof}
    Consider the instance from \Cref{tab:tJR_EJR_incompatible} with $k = 2$, voters $V = \{v_1, v_2\}$, and candidates $C = \{c_1, c_2, d_1, d_2\}$ where $h_{c_1} = h_{c_2} = 1$ and $h_{d_1} = h_{d_2} = 0.1$. Let $A_i = \{c_i\} \cup \{d_1, d_2\}$ for each $i \in [2]$.

    Any value-JR committee $W$ must contain $c_i$ for every $i \in [2]$, since $\{v_i\}$ is $1$-large and cohesive over $ \{c_i\}$ with quality $1$. Hence $W = \{c_1, c_2\}$ is the unique value-JR committee.
    However, $V$ is $2$-cohesive at $\{d_1, d_2\}$, so EJR demands that some voter $i \in V$ has $|A_i \cap W| \geq 2$. Therefore, $W$ violates EJR.
\end{proof}

\section{Axiom Satisfaction} \label{app:designing_rules}

\subsection{Threshold-PJR+ via EAR}\label{app:EAR_threshold-PJR+}

In this section we give a proof for \Cref{prop:EAR_threshold-PJR+}, which says that if $f$ is an EAR, it satisfies threshold-PJR+.

We first define the following price system notion.

\begin{definition}
    For a given instance and set of candidates $X$, a price system $(p, (p_i)_{i \in V})$ is a tuple of a price $p \le 1$ and payment functions $p_i: C \rightarrow [0,1]$ for every voter $i \in V$ such that
    \begin{enumerate}
        \item $i\notin V[c]\Rightarrow p_i(c) = 0$,
        \item $\sum_{c\in X} p_i(c) \le \frac{k}{n}$ for all $i\in V$,
        \item $\sum_{i\in V} p_i(c) = p$ for all $c\in X$, and
        \item $\sum_{i\in V} p_i(d) = 0$ for all $d\in C\setminus X$.
    \end{enumerate}
\end{definition}

We further define the notion of threshold-maximal affordability, which is comparable to priceability from \citet{peters2020proportionality} and the rank-priceability axiom of \citet{BrPe23a}.

\begin{definition}[Threshold-maximal affordability]\label{def:threshold_affordable}
    For a given instance, a set of candidates $X$ is called \emph{threshold-maximally affordable} iff there exists a price system $(p, (p_i)_{i \in V})$, such that for all $h \in [0,1]$ and $d \in C^{\ge h} \setminus X$ it holds that
    $$
        \sum_{i\in V[d]} (\frac{k}{n} - \sum_{c\in X^{\geq h}} p_i(c)) < p.
    $$
\end{definition}

\begin{proposition} \label{prop:theshold-maximally-affordable_threshold-PJR+}
    Any completion of a threshold-maximally affordable set of candidates satisfies threshold-PJR+.
\end{proposition}
\begin{proof}
    Let $X \subseteq C$ of size $\le k$ be threshold-maximally affordable certified by some price system $(p, (p_i)_{i \in V})$ and let $W$ be an arbitrary completion of $X$.
    Assume for contradiction that there exist a threshold $h$, candidate $d \notin W$ with $h_d \ge h$ and a group of supporting voters $S\subseteq V[d]$ such that $S$ is $\ell$-large and $|\appjoin{S} \cap W^{\geq h}|<\ell$.
    Then, the total money that the voters in $S$ can have spent on candidates in $X^{\geq h}$ is at most $(\ell-1) \cdot p \le (\ell-1)$. Therefore, the voters collectively have a remaining budget of $ \lvert S \rvert \frac{k}{n} - (\ell-1) \cdot p \ge \frac{n}{k} \ell \frac{k}{n} - (\ell-1) = 1 \ge p$, a violation of threshold-maximal affordability.
\end{proof}

We can now show \Cref{prop:EAR_threshold-PJR+}.

\begin{proof}[Proof of \Cref{prop:EAR_threshold-PJR+}.]
    Let $W$ be the committee obtained by an EAR. Let $X$ be the uncompleted set of candidates from the EAR and $(1, (p_i)_{i \in V})$ the corresponding price system. We show that $X$ is threshold-maximally affordable. If for some $h$, a candidate $d \in C\setminus X$ with $h_d \ge h$ violates the condition, i.e., $\sum_{i\in V[d]} (\frac{k}{n} - \sum_{c\in X^{\geq h}} p_i(c)) \ge 1$, then we claim that the algorithm would not have stopped the first phase with $X$: after each step during the algorithm with working committee $X'$, we have that $b_i = \frac{k}{n} - \sum_{c\in X'} p_i(c)$. Therefore, consider step $j$ of the algorithm where $h_d = h^j$. Since $h^j = h_d \ge h$, we have $\sum_{i\in V[d]} {b_i} = \sum_{i\in V[d]} (\frac{k}{n} - \sum_{c\in X^{\geq h_d}} p_i(c)) \ge \sum_{i\in V[d]} (\frac{k}{n} - \sum_{c\in X^{\geq h}} p_i(c)) \ge 1$. However, since $d \in C^j$, step $j$ would not have terminated.
    Since $W$ is a completion of the threshold-maximally affordable set $X$, it follows from \Cref{prop:theshold-maximally-affordable_threshold-PJR+} that $W$ satisfies threshold-PJR+.
\end{proof}

While EAR satisfies threshold-PJR+ in polynomial time, the next natural step is to look for a rule which satisfies this axiom and, additionally, outputs a threshold-EJR committee whenever such a committee exists. We show that this is not possible in polynomial time by reduction from the vertex cover problem.

\begin{proposition}\label{prop:threshold-EJR_existence_hardness}
    If $f$ is a rule returning only committees satisfying threshold-EJR whenever such committees exist, then computing $f$ is NP-hard. Similarly, determining whether a given instance admits a committee satisfying threshold-EJR is NP-hard.
\end{proposition}
\begin{proof}
    We reduce from \textsc{Vertex Cover}, which is NP-complete. Given an undirected graph $G=(U,E)$ and an integer $q$, the problem asks whether there exists a set $X \subseteq U$ of size at most $q$ such that $X \cap e \neq \emptyset$ for every edge $e \in E$. We may assume that $q \leq |U|$. Since every superset of a vertex cover is again a vertex cover, a yes-instance then admits a vertex cover of size exactly $q$.

    Let $f$ be a rule that only outputs committees satisfying threshold-EJR whenever such committees exist, and arbitrary committees otherwise.
    We use $f$ as a black-box for vertex cover: Let $(G=(U,E), q)$ be a vertex cover instance without isolated vertices.
    We construct a voting instance $(A,k,h)$ in which we have to select $2\lvert E \rvert + 1$ dummy-candidates and at most $q$ candidates corresponding to vertices. Intuitively any committee satisfying threshold-EJR should correspond to a vertex cover.
    \begin{itemize}
        \item Let $n = k = 2\lvert E \rvert + q + 1$.
        \item For each vertex $x \in U$, define a corresponding candidate $x \in C$ with quality $h_x = 0.5$.
        \item For each $e \in E$, we define two voters $v^1_e$, $v^2_e$, and corresponding dummy candidates $d_e^1$, $d_e^2$ with quality $h_{d_e^1} = h_{d_e^2} = 1$.
        \item Each voter $v_e^j$ for $e=\{x,y\} \in E$ and $j \in \{1,2\}$ has approval ballot $A_{v_e^j} = \{d^j_e, x, y\}$
        \item Define a dummy candidate $d^\star$ with quality $h_{d^\star} = 1$ and create $q + 1$ dummy voters $v^\star_j$ for $j \in [q+1]$ who approve only $A_{v^\star_j} = \{d^\star\}$.
    \end{itemize}

    We claim that in this voting instance $(A,k,h)$, a committee $W$ satisfies threshold-EJR if and only if it contains all dummy candidates and $W \cap U$ is a size $q$ vertex cover for $G$.

    Note first that each voter $v_e^j$ is 1-cohesive at threshold $1$ due to their dummy candidate $d_e^j$. All dummy voters $v^\star_j$ are also 1-cohesive at threshold $1$ for the dummy candidate $d^\star$. Therefore, any committee satisfying threshold-EJR must contain all $2\lvert E\rvert + 1$ dummy candidates.
    Now, let $W$ be a committee satisfying threshold-EJR. Then, since it contains all dummy candidates, it contains exactly $q$ vertex candidates. Consider any edge $e = \{x,y\}$. Note that in the voting instance, $v_e^1, v_e^2$ are 2-cohesive at threshold $\tau = 0.5$, as they both approve of $x$ and $y$. $W$ must therefore contain at least two candidates from either $A_{v^1_e} = \{d^1_e, x, y\}$ or $A_{v^2_e} = \{d^2_e, x, y\}$, which implies that either $x$ or $y$ has to be contained in $W$. Since this holds for all edges $W\cap U$ is a vertex cover.

    For the other direction, let $X$ be a size $q$ vertex cover and let $W$ of size $2\lvert E \rvert + q + 1$ be the union of $X$ and all dummy candidates. Since all candidates with quality $1$ are contained in $W$, there can be no EJR-violation at this threshold. We therefore consider the threshold 0.5. Each individual voter is 1-cohesive and thus deserves one approved candidate in $W$. This demand is covered by the dummy candidates, as each voter approves exactly one of them. To be at least two-cohesive, a group must consist of multiple voters who share multiple approvals. The only such groups are of the form  $v_e^1,v_e^2$ for some $e = \{x,y\}$, who share the two approvals $x, y$. Since they are two-cohesive, at least one of the voters must be represented by two candidates in $W$. Indeed, since $W \cap U$ is a vertex cover and $e = \{x,y\}\in E$, we have that $\lvert \{x,y\} \cap (W\cap U) \rvert \ge 1$. Together with the corresponding dummy candidate, this means that each of the 2-cohesive groups is covered. Thus, $W$ satisfies threshold-EJR.

    Therefore, $f$ can be used as a blackbox to solve vertex cover as follows: Take $(G,q)$ as input, form the corresponding voting instance $(A,k,h)$ and compute any $W = f(A,k,h)$. Verify if $W$ contains all dummy candidates and if $W \cap U$ is a vertex cover in $G$. If yes, return that there exists a vertex cover of size $q$, if no, return that no such vertex cover exists. If a vertex cover $X$ of size $q$ exists, then a threshold-EJR committee $W$ exists and then $W \cap U$ is a valid vertex cover of size $q$, so our algorithm returns yes. If no such cover exists, then our algorithm will return no by definition. NP-hardness now follows from NP-hardness of vertex cover.

	To prove that deciding whether a given instance admits any committee satisfying threshold-EJR is NP-hard, consider the same induced instance. By the above proven claim, the instance admits a committee satisfying threshold-EJR if and only if the original vertex cover problem admits a size $q$ vertex cover.
\end{proof}

\subsection{PJR and Value-EJR via Value-GCR} \label{app:value-GCR_value-EJR_PJR}

In this section, we show that value-GCR satisfies value-EJR and standard PJR. For convenience, we restate the definition and theorem here.

\valueGCR*

\valueGCRValueEJR*

\begin{proof}
    Consider any profile $A$ and target size $k$. We first show that value-GCR terminates with a committee $W$ of size at most $k$.

    Let $r$ be the number of iterations of phase 1.
    Let $(S_x, \ell_x)$ be the witness and $T_x$ the set of candidates maximizing $h(\appcut{S_x}\mid \ell_x)$ that was added in iteration $x \in [r]$. Moreover, let $W_x$ denote the set of candidates after iteration $x$. We claim that all $S_x$ are disjoint. Indeed, once $T_x \setminus W_{x-1}$ is added to $W_{x-1}$, for all $x' > x$ it holds that $h(A_i \cap W_{x'-1}) \ge h(A_i \cap W_x) \ge h(T_x) \ge h(T_{x'})$ for all $i \in S_x$, and therefore $i \notin S_{x'}$ for all $i \in S_x$.

    We now consider the budgets from this process. Each voter $i$ starts with a budget of $b_i = \frac{k}{n}$ and since the voter sets are disjoint, before iteration $x$ the budget of all voters in $S_x$ is still $\frac{k}{n}$. We split the cost of $|T_x \setminus W_{x-1}| \le \ell_x$ over $|S_x| \ge \ell_x \frac{n}{k}$ voters, thus each voter is charged $\frac{|T_x \setminus W_{x-1}|}{|S_x|} \le \frac{\ell_x}{\ell_x \frac{n}{k}} = \frac{k}{n}$. Thus no voter can have a negative budget after iteration $r$. Phase 2 distributes the cost of selected candidates over supporting voters without exceeding their budget. Let $W'$ be the committee after phase 2. Then for each candidate in $W'$ their cost of $1$ was distributed over the voters without exceeding any voter's budget. The size of $W'$ is therefore bounded by the sum of initial budgets $\sum_{i \in V} b_i = k$. The completion step does not increase the size of $W'$ above $k$ by definition, therefore $|W| \le k$.

    We now argue that $W$ satisfies both value-EJR and standard PJR.
    The satisfaction of value-EJR follows immediately from the termination condition of phase 1 and $W \supseteq W_r$.
    Suppose that $W$ does not satisfy PJR, i.e., suppose there is some $\ell \in [k]$ and an $\ell$-large group $S$ with $|\appcut{S}| \ge \ell$ but $|\appjoin{S} \cap W| < \ell$. Then there must be a candidate $c \in \appcut{S}$ with $c \notin W$ and the voters in $S$ cannot have spent more than a budget of $\ell-1$, while their joint budget is at least $\ell$. Therefore, they still would have sufficient budget to purchase $c$, and phase 2 would not have terminated. \qedhere
\end{proof}

\subsection{Value-EJR-1 via Quality-utility MES} \label{app:MES_value-EJR-1}

In this section we demonstrate that while computing a value-EJR committee is NP-hard, we can compute a committee satisfying a relaxation of value-EJR efficiently. The results of this section are mostly translations of results from \citet{peters2021proportional} into our setting.
\begin{definition}[\citep{peters2021proportional}]
    A committee $W$ satisfies \emph{value-EJR-up-to-one-candidate} (value-EJR-1) on instance $\mathcal I$, if for each $\ell \le k$ and each $\ell$-cohesive $S\subseteq V$, we have $\max_{i\in S} h(W \cap A_i) \ge h(\appcut{S} | \ell)$ or there is a candidate\footnote{We can restrict $c$ to this set, see Footnote 9 in the full version of \citet{peters2021proportional}.} $c^\star\in \appcut{S}$ with $\max_{i\in S} h((W \cup \{c^\star\}) \cap A_i) > h(\appcut{S} | \ell).$
\end{definition}

We denote by \emph{quality-utility MES} the standard MES (with arbitrary completion) for cardinal utilities \citep{peters2021proportional} with unit costs and utilities $u_i(c) = h_c$ for $c \in A_i$. The results from \citet{peters2021proportional} then immediately imply that quality-utility MES runs in polynomial time and satisfies value-EJR-1 and PJR.

\begin{proposition}[\citealp{peters2021proportional}]
    Quality-utility MES satisfies value-EJR-1 and PJR.
\end{proposition}

Since quality-utility MES runs in polynomial time, it cannot satisfy value-EJR unless $\mathrm{P}=\mathrm{NP}$ (see \Cref{thm:value-EJR-hard}). We show that it does not even satisfy value-JR.

\begin{proposition}
  Quality-utility MES does not satisfy value-JR with any completion method.
\end{proposition}

\begin{table}
  \caption{Example instance (from \citet{peters2021proportional}) showing that quality-utility MES with any completion method fails value-JR.}
  \label{tab:MES_vEJR}
  \centering
  \begin{tabular}{c c c c}
    \toprule
    $k = 2$ & $c_1$ & $c_2$ & $c_3$\\
    quality & $2/3$ & $1$ & $1$\\
    \midrule
    $v_1$ & \checkmark & \checkmark & \\
    $v_2$ & \checkmark & & \checkmark \\
    \bottomrule
  \end{tabular}
\end{table}

\begin{proof}
    Consider the instance in \Cref{tab:MES_vEJR}. Quality-utility MES selects $c_1$ first and afterwards no project is affordable. Without loss of generality, assume that the completion chooses $c_3$. Then, $S = \{v_1\}$ is 1-cohesive over $T = \{c_2\}$, but $h(\appjoin{S} \cap W) = \frac{2}{3} < 1 = h(\appcut{S} \mid 1)$, violating value-JR.
\end{proof}

\section{Reciprocal Proportionality}\label{app:reciprocal}
We first define reciprocal notions and discuss their relations, then provide axiomatic results.

\subsection{Reciprocal Notions and Relations} \label{app:reciprocal_notions_and_relations}

We first give definitions of all reciprocal axioms in \Cref{tab:reciprocal-proportionality-notions}, except for reciprocal threshold-PJR+, which we define below.
\begin{definition} \label{def:rtPJR+}
    A committee $W$ satisfies  \emph{reciprocal threshold-PJR+} if, for every candidate $d \in C \setminus W$ and every nonempty group of voters $S\subseteq V[d]$, it holds that $\lvert S \rvert \frac{k}{n} < \frac{1}{h_d} + \sum_{c\in \appjoin{S} \cap W^{\ge h_d}}\frac{1}{h_c}$.
\end{definition}

\begin{table*}[t]
    \centering
    \small
    \renewcommand{\arraystretch}{1.25}
    \setlength{\tabcolsep}{4pt}

    \caption{Entitlement and representation requirements of the reciprocal proportionality notions.  }
    \label{tab:reciprocal-proportionality-notions}
    \begin{tabularx}{\textwidth}{
    @{}
    l
    >{\raggedright\arraybackslash}X
    >{\raggedright\arraybackslash}X
    >{\centering\arraybackslash}p{1.6cm}
    >{\centering\arraybackslash}p{2.4cm}
    @{}
    }
        \toprule
        Axiom
          & Entitlement condition\newline
            \emph{For every
            \(S\subseteq V\), \(\ldots\)}
          & Required representation
          & Verifiable in polynomial time
          & One committee computable in polynomial time
          \\
        \midrule

        rec. JR
        & \multirow[c]{3}{=}{%
        \centering
        if \(S\) can reciprocally afford
        \(T\subseteq\appcut{S}\) of size $\ell = \lvert T \rvert \ge 1$\par
        }
          & \(|\appjoin{S} \cap W|\ge 1 \).
          & Yes
          & Yes
          \\

        rec. PJR
          &
          & \(|\appjoin{S} \cap W|\ge \ell \).
          & No
          & Yes
          \\

        rec. EJR
          &
          & \(\displaystyle
              \max_{i\in S}|A_i\cap W|\ge \ell\).
          & No
          & Yes
          \\

        \midrule

        rec. value-JR
        & \multirow[c]{3}{=}{%
        \centering
        if \(S\) can reciprocally afford
        \(T\subseteq\appcut{S}\) of size $\ell = \lvert T \rvert \ge 1$\par
        }
          & \(\displaystyle
              h\bigl(\appjoin{S} \cap W\bigr)\ge h(T\mid 1)\).
          & ?
          & Yes
          \\

        rec. value-PJR
          &
          & \(\displaystyle
              h\bigl(\appjoin{S} \cap W\bigr)\ge h(T \mid \ell)\).
          & No
          & Yes
          \\

        rec. value-EJR
          &
          & \(\displaystyle
              \max_{i\in S}h(A_i\cap W)
              \ge h(T\mid\ell)\).
          & No
          & ?
          \\

        \midrule

        rec. threshold-JR
        & \multirow[c]{3}{=}{%
        \centering
        \emph{For every \(\tau\in[0,1]\):}\par
        if \(S\) can reciprocally afford
        \(T\subseteq\appcut{S}^{\ge\tau}\) of size $\ell = \lvert T \rvert \ge 1$\par
        }
          & \(\displaystyle
              \left|\appjoin{S} \cap W^{\ge\tau}\right|\ge 1\).
          & Yes
          & Yes
          \\

        rec. threshold-PJR
          &
          & \(\displaystyle
              \left|\appjoin{S} \cap W^{\ge\tau}\right|\ge \ell\).
          & No
          & Yes
          \\

        rec. threshold-EJR
          &
          & \(\displaystyle
              \max_{i\in S}|A_i\cap W^{\ge\tau}|
              \ge \ell\).
          & No
          & Not guaranteed to exist.
          \\

        \bottomrule
    \end{tabularx}

\end{table*}

The relationships between the reciprocal axioms exactly match the relationships of their non-reciprocal versions. \Cref{fig:relations_reciprocal} summarizes these relations and links to the corresponding results and proofs below.

As for the threshold-EJR, we can show that reciprocal threshold-EJR, though weaker, remains not always satisfiable.

\begin{sideexample}{%
  \begin{tabular}{@{}c c c c c c@{}}
    \toprule
    $k=3$     & $c_1$      & $c_2$      & $c_3$      & $d_1$      & $d_2$      \\
    quality & $1$        & $1$        & $1$        & $\frac{2}{3}$      & $\frac{2}{3}$      \\
    \midrule
    $v_1$     & \checkmark &             &            & \checkmark & \checkmark \\
    $v_2$     &            & \checkmark  &            & \checkmark & \checkmark \\
    $v_3$     &            &             & \checkmark & \checkmark & \checkmark \\
    \bottomrule
  \end{tabular}%
}
Consider the following instance with $n=k=3$.
At $\tau=1$, each voter can reciprocally afford one candidate and deserves representation from a candidate with quality $1$, thus $c_1$, $c_2$ and $c_3$ must be selected. At $\tau=\frac{2}{3}$, all voters together can reciprocally afford $\{d_1, d_2\}$ and therefore some voter must be represented by two candidates, which is only possible if $d_1$ or $d_2$ is selected.
\end{sideexample}

\begin{figure}
  \centering
  \begin{tikzpicture}[
      x=5.0cm,
      y=1.8cm,
      every node/.style={align=center}
    ]

    \node (EJRl) at (-1,2) {rec.\,EJR};
    \node (PJRl) at (-1,1) {rec.\,PJR};
    \node (JRl)  at (-1,0) {rec.\,JR};

    \node (tEJR) at (0,2)
    {rec.\,threshold-EJR\\[-1mm]\scriptsize unsatisfiable};
    \node (tPJR) at (0,1)
    {rec.\,threshold-PJR\\[-1mm]\scriptsize polynomial-time satisfiable};
    \node (tJR) at (0,0)
    {rec.\,threshold-JR};

    \node (vEJR) at (1,2)
    {rec.\,value-EJR\\[-1mm]\scriptsize satisfiable};
    \node (vPJR) at (1,1)
    {\,\,rec.\,value-PJR};
    \node (vJR) at (1,0)
    {rec.\,value-JR};

    \node (EJRr) at (2,2) {rec.\,EJR};
    \node (PJRr) at (2,1) {rec.\,PJR};
    \node (JRr)  at (2,0) {rec.\,JR};

    \draw[implication]
    (tEJR.south) -- (tPJR.north);

    \draw[implication]
    (tPJR.south) -- (tJR.north);

    \draw[implication]
    (vEJR.south) -- (vPJR.north);

    \draw[implication]
    (vPJR.south) -- (vJR.north);

    \draw[implication] (EJRl.south) -- (PJRl.north);
    \draw[implication] (PJRl.south) -- (JRl.north);

    \draw[implication] (EJRr.south) -- (PJRr.north);
    \draw[implication] (PJRr.south) -- (JRr.north);

    \arrowref[0.25]
      {incompatibility}
      {EJRl.south east}
      {tJR.north west}
      {above}
        {\cshref{prop:rEJR_rtJR_incompatible}}

    \arrowref[0.30]
        {nonimplication}
        {vEJR.east}
        {PJRr.west}
        {above}
        {\cshref{cor:rvEJR_rPJR_compatible}, \cshref{prop:rvEJR_rPJR}}

    \arrowref[0.3]
        {incompatibility}
        {vJR.north east}
        {EJRr.south west}
        {above}
        {\cshref{prop:rEJR_rvJR_incompatible}}

    \arrowref[0.5]
        {incompatibility}
        {tEJR.east}
        {vEJR.west}
        {above}
        {\cshref{thm:rec_threshold_ejr_not_rec_value_ejr}*}

    \arrowref[0.23]
        {incompatibility}
        {tJR.north east}
        {vEJR.south west}
        {below}
        {\cshref{prop:rvEJR_rtJR_incompatible}}

    \arrowref[0.5]
        {implication}
        {tEJR.west}
        {EJRl.east}
        {below}
        {\cshref{obs:rthreshold_implies_basic}}

    \arrowref[0.65]
        {implication}
        {tPJR.west}
        {PJRl.east}
        {below}
        {\cshref{obs:rthreshold_implies_basic}}

    \arrowref[0.5]
        {implication}
        {tJR.west}
        {JRl.east}
        {below}
        {\cshref{obs:rthreshold_implies_basic}}

    \arrowref[0.5]
        {implication}
        {vJR.east}
        {JRr.west}
        {below}
        {\cshref{obs:rJR_rvJR}}

    \arrowref[0.79]
        {implication}
        {tPJR.east}
        {vPJR.west}
        {below}
        {\cshref{prop:rtPJR_implies_rvPJR}}

    \arrowref[0.5]
        {implication}
        {tJR.east}
        {vJR.west}
        {below}
        {\cshref{cor:rtJR_implies_rvJR}}

  \end{tikzpicture}
  \caption{Relationships of reciprocal threshold-, reciprocal value- and reciprocal JR notions. Arrows denote implications, (red) dashed lines incompatibilities and dotted lines compatibilities without implications. Only the implications that can be inferred from the depicted arrows hold.\\
  \textcolor{red}{*} incompatibility on an instance where reciprocal threshold-EJR committees exist. }\label{fig:relations_reciprocal}
\end{figure}
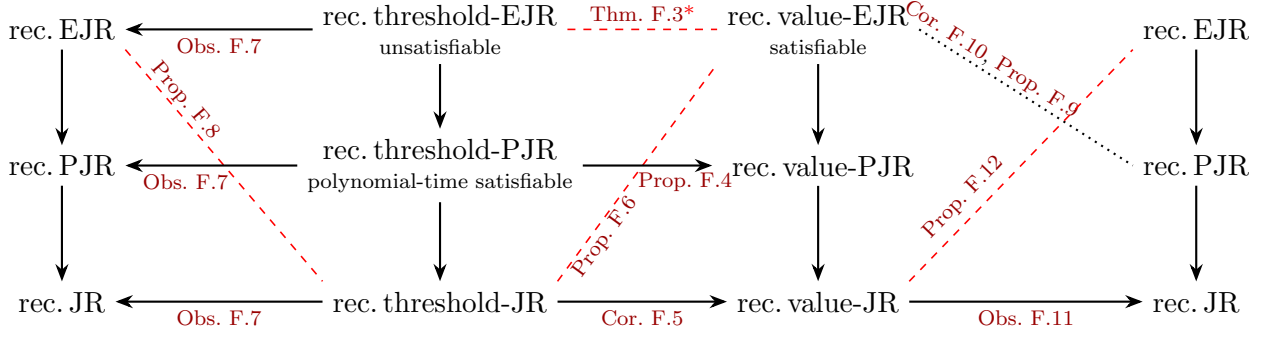

\subsubsection{Threshold- and Value-JR Notions}

\begin{theorem}\label{thm:rec_threshold_ejr_not_rec_value_ejr}
  Reciprocal threshold-EJR and reciprocal value-EJR are incompatible, even on an instance that admits a reciprocal threshold-EJR committee.
\end{theorem}

\begin{table}[H]
    \caption{An instance with $n=16$ and $k=8$ on which reciprocal threshold-EJR is satisfiable by the committees $\{e_1,\dots,e_5,f_1,f_2,f_3\}$ and $\{e_1,\dots,e_5,f_1,f_2,f_4\}$, but none of them satisfies reciprocal value-EJR.}
    \label{tab:rec_threshold_ejr_not_rec_value_ejr}
    \centering
        \begin{tabular}{c c c >{\columncolor{gray!20}} c >{\columncolor{gray!20}} c >{\columncolor{gray!20}} c >{\columncolor{gray!20}} c >{\columncolor{gray!20}} c >{\columncolor{gray!20}} c >{\columncolor{gray!20}} c >{\columncolor{gray!10}} c >{\columncolor{gray!10}} c}
        \toprule
        $k=8$ & $a$ & $b$ & $e_1$ & $e_2$ & $e_3$ & $e_4$ & $e_5$ & $f_1$ & $f_2$ & $f_3$ & $f_4$\\
        quality & $1$ & $0.8$ & $1$ & $1$ & $1$ & $1$ & $1$ & $0.8$ & $0.8$ & $0.8$ & $0.8$\\
        \midrule
        $v_1$ & \checkmark & \checkmark & \checkmark & & & & & & & & \\
        $v_2$ & \checkmark & \checkmark & \checkmark & & & & & & & & \\
        $v_3$ & \checkmark & \checkmark & & \checkmark & & & & & & & \\
        $v_4$ & \checkmark & \checkmark & & \checkmark & & & & & & & \\
        $v_5$ & \checkmark & \checkmark & & & & & & \checkmark & \checkmark & & \\
        $v_6$ & & & \checkmark & & & & & \checkmark & \checkmark & \checkmark & \checkmark \\
        $v_7$ & & & \checkmark & & & & & \checkmark & \checkmark & \checkmark & \checkmark \\
        $v_8$ & & & \checkmark & & & & & \checkmark & \checkmark & \checkmark & \checkmark \\
        $v_9$ & & & & \checkmark & & & & \checkmark & \checkmark & \checkmark & \checkmark \\
        $v_{10}$ & & & & \checkmark & & & & \checkmark & \checkmark & \checkmark & \checkmark \\
        $v_{11}$ & & & & & \checkmark & & & \checkmark & \checkmark & \checkmark & \checkmark \\
        $v_{12}$ & & & & & \checkmark & & & \checkmark & \checkmark & \checkmark & \checkmark \\
        $v_{13}$ & & & & & & \checkmark & & \checkmark & \checkmark & \checkmark & \checkmark \\
        $v_{14}$ & & & & & & \checkmark & & \checkmark & \checkmark & \checkmark & \checkmark \\
        $v_{15}$ & & & & & & & \checkmark & \checkmark & \checkmark & \checkmark & \checkmark \\
        $v_{16}$ & & & & & & & \checkmark & \checkmark & \checkmark & \checkmark & \checkmark \\
        \bottomrule
    \end{tabular}
\end{table}

\begin{proof}
    Consider the instance in \Cref{tab:rec_threshold_ejr_not_rec_value_ejr}. We have $n = 16$, $k = 8$, and hence $\frac{n}{k} = 2$. Let $S_1=\{v_1,v_2\}$, $S_2=\{v_3,v_4\}$, $Q=\{v_5\}$, and let $R_1, \dots, R_5$ denote the five remaining identical voter blocks, of sizes $3, 2, 2, 2, 2$, respectively. Write $R = R_1 \cup \dots \cup R_5$.

    We first identify all reciprocal threshold-EJR committees. For every $j \in [5]$, the voters in $R_j$ commonly approve $e_j$, which is the only quality-$1$ candidate on their ballots. Since $|R_j| \ge 2 \ge \frac{1}{h_{e_j}} \frac{n}{k} = 2$, reciprocal threshold-EJR forces each $e_j$ for $j \in [5]$ to be in the committee. Moreover, the $11$ voters in $R$ commonly approve $f_1, \dots, f_4$, and $|R| = 11 \ge \sum_{x=1}^4 \frac{1}{h_{f_x}} \frac{n}{k} = 10$. They therefore require some voter to approve four selected candidates of quality at least $0.8$. Since each voter in $R_j$ already approves $e_j$, at least three candidates from $\{f_1, \dots, f_4\}$ must be selected. The five candidates $e_1, \dots, e_5$ and these three candidates already fill the committee, so neither $a$ nor $b$ can be in it.
    Now consider the group $S = S_1 \cup S_2 \cup Q$. The voters in $S$ commonly approve $\{a, b\}$, and $\lvert S \rvert = 5 \ge (\frac{1}{h_a} + \frac{1}{h_b}) \frac{n}{k} = \frac{9}{2}$. Hence some voter in $S$ must approve two selected candidates of quality at least $0.8$. The voters in $S_1$ and $S_2$ approve only $e_1$ and $e_2$, respectively, among the selected candidates. Thus, the witness must be $v_5$, which forces $f_1, f_2$ to be in the committee. Consequently, the only possible reciprocal threshold-EJR committees are $W_3 = \{e_1, \dots, e_5, f_1, f_2, f_3\}$ and $W_4 = \{e_1, \dots, e_5, f_1, f_2, f_4\}$.

    It remains to show that both committees indeed satisfy reciprocal threshold-EJR. By symmetry, consider $W_3$. Every voter in $S_1$ and $S_2$ approves a selected quality-$1$ candidate, voter $v_5$ approves $f_1$ and $f_2$, and every voter in $R_j$ approves $e_j,f_1,f_2,f_3$. Groups contained in one of $S_1$, $S_2$, or $R_j$ are too small to raise a claim for two candidates, and their singleton claims are satisfied. Any group meeting two different $R$-blocks has common approvals contained in $\{f_1, \dots, f_4\}$ and contains a voter approving four selected candidates. Groups mixing $Q$ with $R$ have common approvals contained in $\{f_1, f_2\}$, which are both selected. Groups from $S_1 \cup R_1$ or $S_2 \cup R_2$ only approve one candidate in common, and all voters in these groups approve a quality-1 candidate. Finally, among groups contained in $S = S_1 \cup S_2 \cup Q$, the only eligible two-candidate claim is the claim of the full group $S$ for $\{a, b\}$, and this is witnessed by $v_5$ through $f_1$ and $f_2$. Thus, $W_3$ satisfies reciprocal threshold-EJR, and the same argument applies to $W_4$.

    Finally, both committees violate reciprocal value-EJR. The group $S = S_1 \cup S_2 \cup Q$ can claim $T = \{a, b\}$ as above, whose total quality is $h(T) = 1 + 0.8 = 1.8$. Under either $W_3$ or $W_4$, every voter in $S_1 \cup S_2$ has utility $1$, while $v_5$ has utility $h_{f_1} + h_{f_2} = 1.6$. Hence no voter in $S$ obtains utility at least $1.8$. Therefore, no committee satisfies both reciprocal threshold-EJR and reciprocal value-EJR.
\end{proof}

\begin{proposition}\label{prop:rtPJR_implies_rvPJR}
  Reciprocal threshold-PJR implies reciprocal value-PJR.
\end{proposition}

\begin{proof}
   Let $W$ be a reciprocal threshold-PJR committee.
  Consider any group of voters $S$ that can reciprocally afford a set $T \subseteq \appcut{S}$ of size $\ell$. Then $S$ can also reciprocally afford the set $$T^\star \in \argmax_{\substack{T\subseteq \appcut{S}\\ |T| = \ell}} \,\, \sum_{c \in T} h_c$$
  of $\ell$ commonly approved candidates with the highest quality, i.e., $\lvert S \rvert \ge \sum_{c \in T^\star}\frac{1}{h_c} \cdot \frac{n}{k}$. Define $c_1, \dots, c_\ell$ as the candidates in $T^\star$ in order of decreasing quality. Let $W_S=\appjoin{S} \cap W$. Note that reciprocal threshold-PJR with threshold $\tau = 0$ implies that $|W_S|\geq \ell$, as $S$ is large enough and $\ell$-cohesive. Again, define $d_1, \dots, d_\ell$ as the first $\ell$ candidates in $W_S$ in order of decreasing quality. In the following, we show that $h_{d_x}\geq h_{c_x}$ for all $x \in [\ell]$.
  Let $x \in [\ell]$. Note that $S$ is of size $\lvert S \rvert \ge \sum_{c \in \{c_1, \dots, c_x\}}\frac{1}{h_c} \cdot \frac{n}{k}$ and $x$-cohesive for the quality threshold $\tau_x = h_{c_x}$. Thus, reciprocal threshold-PJR implies that $|\appjoin{S} \cap W^{\ge \tau_x}| \ge x$. Therefore, the candidate with the $x$-th highest quality in $W_S$ needs to have quality at least $h_{c_x}$, so $h_{d_x} \ge h_{c_x}$. Since this holds for all $x \in [\ell]$, we have that
  $$
    h(\appjoin{S} \cap W)
    \ge \sum_{x \in [\ell]} h_{d_x}
    \geq \sum_{x \in [\ell]} h_{c_x}
    = \sum_{c \in T^\star} h_c
    = h(\appcut{S} \mid \ell),
  $$
  which completes the proof.
\end{proof}

The above proof works in particular for the special case of $\ell = 1$, giving the following corollary.

\begin{corollary}\label{cor:rtJR_implies_rvJR}
  Reciprocal threshold-JR implies reciprocal value-JR.
\end{corollary}

The incompatibility of reciprocal value- and reciprocal threshold-EJR can be strengthened to reciprocal value-EJR and reciprocal threshold-JR.

\begin{proposition}\label{prop:rvEJR_rtJR_incompatible}
    Reciprocal value-EJR and reciprocal threshold-JR are incompatible.
\end{proposition}

\begin{table}[H]
    \caption{Example instance showing that reciprocal value-EJR and reciprocal threshold-JR are incompatible. The only reciprocal threshold-JR committee is $\{c_1, c_2, c_3\}$ (choose threshold $\tau = 1$), which violates reciprocal value-EJR for the voter group $\{v_1, v_2, v_3\}$.}
    \centering
    \label{tab:rvEJR_rtJR_incompatible}
    \begin{tabular}{c >{\columncolor{gray!20}} c >{\columncolor{gray!20}}  c >{\columncolor{gray!20}}  c c c c}
        \toprule
        $k = 3$ & $c_1$ & $c_2$ & $c_3$ & $d_1$ & $d_2$ & $d_3$\\
        quality & $1$ & $1$ & $1$ & $0.75$ & $0.75$ & $0.75$\\
        \midrule
        $v_1$ & \checkmark & & & \checkmark & \checkmark & \checkmark \\
        $v_2$ & & \checkmark & & \checkmark & \checkmark & \checkmark \\
        $v_3$ & & & \checkmark & \checkmark & \checkmark & \checkmark \\
        \bottomrule
    \end{tabular}
\end{table}

\begin{proof}
      Consider the instance from \Cref{tab:rvEJR_rtJR_incompatible} with $k = 3$, voters $V = \{v_1, v_2, v_3\}$, and candidates $C = \{c_1, c_2, c_3, d_1, d_2, d_3\}$ where $h_{c_i} = 1$ for $i \in [3]$ and $h_{d_j} = 0.75$ for $j \in [3]$. Let $A_i = \{c_i\} \cup \{d_1, d_2, d_3\}$ for each $i \in [3]$.

      Any reciprocal threshold-JR committee $W$ must contain $c_i$ for every $i \in [3]$, since $\{v_i\}$ is $1$-large (as $\frac{n}{k} = 1$) and approves only $c_i$ at quality threshold $1$. Hence $W = \{c_1, c_2, c_3\}$ is the unique reciprocal threshold-JR committee.
      However, $V$ can reciprocally afford $T = \{d_1, d_2\}$, since $3 = |V| \ge \sum_{c \in T}\frac{1}{h_c} \cdot \frac{n}{k} = \frac{8}{3}$, so reciprocal value-EJR demands that some voter $i \in V$
    has utility $h(A_i \cap W) \ge \sum_{c \in T} h_c = 1.5$. Yet, each voter $v_i$ has utility $h(A_i \cap W) = 1$, violating reciprocal value-EJR.
\end{proof}

\subsubsection{Threshold- and Basic JR Notions}

It is easy to see that all reciprocal threshold-JR notions imply the corresponding reciprocal JR notions, as choosing the threshold $\tau = 0$ gives exactly the corresponding definition.

\begin{observation}\label{obs:rthreshold_implies_basic}
    Reciprocal threshold-EJR implies reciprocal EJR, reciprocal threshold-PJR implies reciprocal PJR and reciprocal threshold-JR implies reciprocal JR.
\end{observation}

Clearly, the other direction does not hold, in fact we can show that reciprocal EJR is incompatible with reciprocal threshold-JR.

\begin{proposition}\label{prop:rEJR_rtJR_incompatible}
  Reciprocal EJR is incompatible with reciprocal threshold-JR.
\end{proposition}

\begin{table}[H]
  \caption{Example instance showing that reciprocal threshold-JR and reciprocal EJR (also reciprocal value-JR and reciprocal EJR) are incompatible. The only reciprocal threshold-JR (reciprocal value-JR) committee is $\{c_1, \dots, c_5\}$, which violates reciprocal EJR for the voter group $\{v_1, \dots, v_5\}$.}
  \label{tab:rtJR_rEJR_incompatible}
  \centering
  \begin{tabular}{c >{\columncolor{gray!20}} c >{\columncolor{gray!20}} c >{\columncolor{gray!20}} c c c c}
    \toprule
    $k = 5$ & $c_1$ & $\cdots$ & $c_5$ & $d_1$ & $d_2$ \\
    quality & $1$ & $\cdots$ & $1$ & $0.4$ & $0.4$ \\
    \midrule
    $v_1$ & \checkmark & & & \checkmark & \checkmark \\
    $\vdots$ & & $\ddots$ & & $\vdots$ & $\vdots$ \\
    $v_5$ & & & \checkmark & \checkmark & \checkmark \\
    \bottomrule
  \end{tabular}
\end{table}

\begin{proof}
    Consider the instance from \Cref{tab:rtJR_rEJR_incompatible} with $k = 5$, voters $V = \{v_1, \dots, v_5\}$, and candidates $C = \{c_1, \dots, c_5, d_1, d_2\}$ where $h_{c_1} = \dots = h_{c_5} = 1$ and $h_{d_1} = h_{d_2} = 0.4$. Let $A_i = \{c_i\} \cup \{d_1, d_2\}$ for each $i \in [5]$.

    Any reciprocal threshold-JR committee $W$ must contain $c_i$ for every $i \in [5]$, since $\{v_i\}$ is $1$-large (as $\frac{n}{k} = 1$) and approves only $c_i$ at quality threshold $1$. Hence $W = \{c_1, \dots, c_5\}$ is the unique reciprocal threshold-JR committee.
    However, $V$ can reciprocally afford $T = \{d_1, d_2\}$ since $|V| \ge \sum_{c \in T}\frac{1}{h_c} \cdot \frac{n}{k} = 5$, so reciprocal EJR demands that some voter $i \in V$ has $|A_i \cap W| \geq 2$. Therefore, $W$ violates reciprocal EJR.
\end{proof}

\subsubsection{Value- and Basic JR Notions}

\begin{proposition}\label{prop:rvEJR_rPJR}
    Reciprocal value-EJR does not imply reciprocal PJR.
\end{proposition}
\begin{table}[H]
    \caption{Example instance showing that reciprocal value-EJR does not imply reciprocal PJR. The committee $W = \{c_3, c_4, c_5, c_6, c_7\}$ (shaded) satisfies reciprocal value-EJR, but violates reciprocal PJR for the group $S = \{v_1,\dots,v_5\}$ and the commonly approved set $T = \{c_1,c_2\}$.}
    \label{tab:rvPJR_rPJR}
    \centering
    \begin{tabular}{c c c >{\columncolor{gray!20}} c >{\columncolor{gray!20}} c >{\columncolor{gray!20}} c >{\columncolor{gray!20}} c >{\columncolor{gray!20}} c}
        \toprule
        $k = 5$ & $c_1$ & $c_2$ & $c_3$ & $c_4$ & $c_5$ & $c_6$ & $c_7$\\
        quality & $0.5$ & $0.5$ & $1$ & $1$ & $1$ & $1$ & $1$\\
        \midrule
        $v_1$ & \checkmark & \checkmark & \checkmark & & & & \\
        $v_2$ & \checkmark & \checkmark & \checkmark & & & & \\
        $v_3$ & \checkmark & \checkmark & \checkmark & & & & \\
        $v_4$ & \checkmark & \checkmark & & & & & \\
        $v_5$ & \checkmark & \checkmark & & & & & \\
        $v_6$ & & & & \checkmark & \checkmark & \checkmark & \checkmark \\
        \bottomrule
    \end{tabular}
\end{table}

\begin{proof}
    Consider the instance from \Cref{tab:rvPJR_rPJR} with $k = 5$, voters $v_1$ to $v_6$, and candidates $c_1$ to $c_7$ where $h_{c_1} = h_{c_2} = 0.5$ and $h_{c_3} = h_{c_4} = h_{c_5} = h_{c_6} = h_{c_7} = 1$. Let $A_1 = A_2 = A_3 = \{c_1, c_2, c_3\}$, $A_4 = A_5 = \{c_1, c_2\}$ and $A_6 = \{c_4, c_5, c_6, c_7\}$.

    The committee $W = \{c_3, c_4, c_5, c_6, c_7\}$ satisfies reciprocal value-EJR. To verify this, we check all candidate sets approved by at least two voters (since a single voter has no demand with $\frac{n}{k} = \frac{6}{5} > 1$). The singleton candidate sets $\{c_1\}$ and $\{c_2\}$ need a group of at least $\frac{1}{0.5} \frac{n}{k} = \frac{12}{5} > 2$ voters to demand a voter with utility at least $0.5$. Every $3$-sized subset of voters supporting $c_1$ or $c_2$ has one such voter. For $\{c_3\}$ two voters suffice, but each individual voter supporting $c_3$ has utility $1$. The set $\{c_1, c_2\}$ needs support of at least $(\frac{1}{0.5}+\frac{1}{0.5})\frac{n}{k} = \frac{24}{5} > 4$ voters and the only size $5$ group of voters supporting $c_1$ and $c_2$ has a voter approving the selected candidate $c_3$, with utility $h_{c_3} = 1 \ge 0.5 + 0.5$. For $\{c_1, c_3\}$ and $\{c_2, c_3\}$ a voter group needs to have size $(\frac{1}{0.5}+\frac{1}{1})\frac{n}{k} = \frac{18}{5} > 3$, but there are only $3$ voters supporting any of these sets. For $\{c_1,c_2,c_3\}$ the group size would need to be even larger. Thus $W$ satisfies reciprocal value-EJR.

    However, $W$ does not satisfy reciprocal PJR since the candidate set $\{c_1, c_2\}$ is supported by $5 > (\frac{1}{0.5}+\frac{1}{0.5})\frac{n}{k} = \frac{24}{5}$ voters, but $|W \cap \appjoin{S}| = 1 < 2$.
\end{proof}

Even though reciprocal value-EJR does not imply reciprocal PJR, the notions are always satisfiable together, which follows immediately from the compatibility of the non-reciprocal notions (see \Cref{prop:vEJR_PJR_compatible}), as each notion implies their corresponding reciprocal notion.

\begin{corollary}\label{cor:rvEJR_rPJR_compatible}
  Reciprocal value-EJR and reciprocal PJR are simultaneously satisfiable in every profile.
\end{corollary}

\begin{observation}\label{obs:rJR_rvJR}
  Reciprocal value-JR implies reciprocal JR.
\end{observation}
\begin{proof}
    Let $W$ be a committee that satisfies reciprocal value-JR. Then for every group of voters $S$ that can reciprocally afford $T$ of size $\lvert T \rvert \ge 1$ we have $h(\appjoin{S} \cap W) = \sum_{c \in \appjoin{S} \cap W} h_c > 0$ and thus $\appjoin{S} \cap W \neq \emptyset$.
\end{proof}

\begin{proposition}\label{prop:rEJR_rvJR_incompatible}
  Reciprocal EJR is incompatible with reciprocal value-JR.
\end{proposition}

\begin{proof}
    Consider the instance from \Cref{tab:rtJR_rEJR_incompatible} with $k = 5$, voters $V = \{v_1, \dots, v_5\}$, and candidates $C = \{c_1, \dots, c_5, d_1, d_2\}$ where $h_{c_1} = \dots = h_{c_5} = 1$ and $h_{d_1} = h_{d_2} = 0.4$. Let $A_i = \{c_i\} \cup \{d_1, d_2\}$ for each $i \in [5]$.

    Any reciprocal value-JR committee $W$ must contain $c_i$ for every $i \in [5]$, since $\{v_i\}$ is $1$-large and cohesive over $\{c_i\}$ with quality $1$. Hence $W = \{c_1, \dots, c_5\}$ is the unique reciprocal value-JR committee.
    However, $V$ can reciprocally afford $T = \{d_1, d_2\}$ since $|V| \ge \sum_{c \in T}\frac{1}{h_c} \cdot \frac{n}{k} = 5$, so reciprocal EJR demands that some voter $i \in V$ has $|A_i \cap W| \geq 2$. Therefore, $W$ violates reciprocal EJR.
\end{proof}

\subsection{Reciprocal Axiomatics}\label{app:reciprocal_axiomatics}
By demanding on party-list instances that only parties $C_j$ with at least $\frac{n}{k} \frac{1}{h_c}$ many supporters are eligible for representation, where $c= \argmax_{c\in C_j} h_c$, we obtain characterizations of reciprocal JR notions.

A proportionality notion $f$ satisfies \emph{singleton reciprocal party list representation}, if on all party list instances $\mathcal I$, with a singleton party $C_1 = \{c\}$ and at least $\frac{1}{h_c} \cdot \frac nk$ voters supporting this candidate, we have that $c\in W$ for all $W\in f(\mathcal I)$.

\begin{proposition}
  Let a proportionality axiom satisfy singleton reciprocal party list representation, monotonicity, independence of losers, and robustness to fully satisfied voters. Then, it is a refinement of reciprocal JR.
\end{proposition}
\begin{proof}
    We prove this claim by contraposition. Let $\mathcal I$ be any instance on which $W$ violates reciprocal JR. Assume $W\in f(\mathcal I)$. Then, there exists some $c\in C\setminus W$ and some $S\subseteq V[c]$ of size at least $\frac nk\frac{1}{h_c}$ such that $A_i\cap W = \emptyset$ for all $i\in S$. By independence of losers, we remove all candidates in $C\setminus (W \cup\{c\})$ from the instance to obtain an instance $\mathcal I_1$ with $W\in f(\mathcal I_1)$. We then apply robustness to fully satisfied voters to remove approvals of $c$ and monotonicity to add approvals of $W$ to obtain an instance $\mathcal I_2$ on which all voters in $V\setminus S$ satisfy $A_i = W$. Still $W\in f(\mathcal I_2)$, and since we did not alter the ballots in $S$, on this instance we have that $A_i = \{c\}$ for all $i\in S$. Therefore, we have a party-list instance on which the committee $W$ (and hence also the rule $f$) violates singleton reciprocal party list representation, as desired.
\end{proof}

Similarly, we can weaken the total and best-in-slot lower quota axioms and define \emph{reciprocal total party-list representation} and \emph{reciprocal best-in-slot party-list representation}, where parties are eligible for representation by one candidate if they have at least $\frac nk \frac{1}{h_c}$ supporters, where $h_c$ in both cases is the maximum quality achievable in the considered party.\footnote{Note that for both axioms, we consider parties of all sizes, not just consisting of a single candidate.}

This leads to characterizations of reciprocal value-JR and reciprocal threshold-JR.
\begin{proposition}
    The following statements hold:

    \begin{itemize}
        \item If $f$ satisfies reciprocal total party-list representation, monotonicity, independence of losers, and robustness to fully satisfied voters, then $f$ is a refinement of reciprocal value-JR. Further, reciprocal value-JR satisfies all these axioms.
        \item If $f$ satisfies reciprocal best-in-slot party-list representation, monotonicity, independence of losers, and robustness to fully satisfied voters, then $f$ is a refinement of reciprocal threshold-JR. Further, reciprocal threshold-JR satisfies all these axioms.
    \end{itemize}
\end{proposition}
\begin{proof}

    \begin{itemize}
    \item Let $W\in f(\mathcal I) \setminus \text{rvJR}(\mathcal I)$. Then, there is $S\subseteq V[c]$ for some $c\notin W$, with $S$ containing at least $\frac nk \frac 1{h_c}$ voters, such that $h(\appjoin{S} \cap W) < h_c$. We proceed as follows: apply independence of losers to remove all candidates in $C\setminus W$ from the instance, except for $c$;
    monotonicity to let all voters in $S$ approve of $(\appjoin{S} \cap W) \cup\{c\}$;
    robustness to fully satisfied voters and monotonicity to let each $i\in V\setminus S$ approve of precisely $W\setminus (\appjoin{S} \cap W)$. In the resulting party-list instance $\mathcal I^*$ with $V_1 = S$, $C_1 = (\appjoin{S} \cap W) \cup\{c\}$ $V_2 = V\setminus S$, $C_2 = W\setminus (\appjoin{S} \cap W)$ we still have $h(\appjoin{S}^*\cap W)= h(\appjoin{S} \cap W) < h_c = h_{\mathcal I^*}(C^*_S\mid 1)$. Here, $\appjoin{S}^*\cap W$ and $\appcut{S}^*$ are defined w.r.t. the profile $\mathcal I^*$ instead of profile $\mathcal I$.
    Therefore, on instance $\mathcal I^*$, the committee $W$ violates reciprocal total party-list representation. However, since we only applied well-behaved profile changes, $W\in X(\mathcal I^*)$

    \item Finally, let $W\in f(\mathcal I) \setminus  \text{rtJR}(\mathcal I)$. Then, there is  $c\notin W$ and $S\subseteq V[c]$ with $S$ containing at least $\frac nk \frac 1{h_c}$ voters, but $\lvert \appjoin{S} \cap W^{\ge h_c} \rvert = 0$. Set $\tau = h_c$. We proceed with similar operations as for value-JR: apply independence of losers to remove all candidates in $C\setminus W$ from the instance, except for $c$;
    monotonicity to let all voters in $S$ approve of $(\appjoin{S} \cap W) \cup\{c\}$;
    robustness to fully satisfied voters and monotonicity to let each $i\in V\setminus S$ approve of precisely $W\setminus (\appjoin{S} \cap W)$. In the resulting party-list instance $\mathcal I^*$ with $V_1 = S$, $C_1 = (\appjoin{S} \cap W) \cup\{c\}$, $V_2 = V\setminus S$, $C_2 = W\setminus (\appjoin{S} \cap W)$ we still have $\appjoin{S}^*\cap W^{\ge \tau} = \appjoin{S} \cap W^{\ge \tau} = \emptyset$. Here, $\appjoin{S}^*$ is defined w.r.t. the profile $\mathcal I^*$ instead of profile $\mathcal I$.
    Therefore, on instance $\mathcal I^*$, the committee $W$ violates reciprocal best-in-slot party-list representation. However, since we only applied well-behaved profile changes, $W\in f(\mathcal I^*)$.\qedhere
\end{itemize}
\end{proof}

\paragraph{PJR and EJR level reciprocal notions}
Recall that $S\subseteq V$ can \emph{reciprocally afford} a nonempty set $T\subseteq \appcut{S}$ if $\sum_{c\in T} \frac{1}{h_c} \le \frac kn \lvert S\rvert$.

\emph{Reciprocal lower quota} requires the following on party-list instances. Let $V_j$ be able to reciprocally afford $T\subseteq C_j$. Then, $\lvert C_j\cap W \rvert \ge \lvert T \rvert$.

\emph{Reciprocal total lower quota} requires the following on party-list instances. If $T\subseteq C_j$ is reciprocally affordable for $V_j$, then $h(C_j\cap W)\ge h(T)$.

\emph{Reciprocal best-in-slot lower quota} requires the following on party-list instances. If any candidate set $T\subseteq C_j$ is reciprocally affordable for $V_j$, then there exists $U\subseteq C_j\cap W$ s.t. $\lvert U \rvert = \lvert T \rvert$ and $h_c \ge h_d$ for all $c\in U$, $d\in C_j\setminus U$. Intuitively, this means that if $V_j$ can reciprocally afford a certain number of candidates, then $W$ should provide this number of candidates to $V_j$ in a best-in-slot fashion.

Indeed, it is easy to see that reciprocal PJR and reciprocal EJR satisfy reciprocal lower quota. Their value-versions satisfy reciprocal total lower quota, and their threshold-versions satisfy reciprocal best-in-slot lower quota.

Finally, we characterize reciprocal threshold-PJR+.%

\begin{theorem}
    Let $X$ satisfy independence of losers, robustness to fully satisfied voters, and monotonicity. If $X$ satisfies reciprocal best-in-slot lower quota, $X$ refines reciprocal threshold-PJR+. Vice versa, reciprocal threshold-PJR+ satisfies all these axioms.
\end{theorem}
\begin{proof}
    It is easy to verify that reciprocal $\thresholdPJR+$ satisfies the claimed properties.
    Let $W\in X(\mathcal I) \setminus \text{rtPJR+}(\mathcal I)$. Then, there is $d\in C\setminus W$ and (non-empty) $S\subseteq V[d]$ such that $\lvert S \rvert \frac{k}{n} \ge \frac{1}{h_d} + \sum_{c\in \appjoin{S} \cap W^{\ge h_d}}\frac{1}{h_c}$.
    We proceed with the exact same idea as for the threshold-PJR result: apply independence of losers to remove all candidates in $C\setminus W$ from the instance, except for $d$;
    monotonicity to let all voters in $S$ approve of $(\appjoin{S} \cap W) \cup\{d\}$;
    robustness to fully satisfied voters and monotonicity to let each $i\in V\setminus S$ approve of precisely $W\setminus (\appjoin{S} \cap W)$. In the resulting party-list instance $\mathcal I^*$ with $V_1 = S$, $C_1 = (\appjoin{S} \cap W) \cup\{d\}$, $V_2 = V\setminus S$, $C_2 = W\setminus \appjoin{S}$ we still have $\lvert S \rvert \frac{k}{n} \ge \frac{1}{h_d} + \sum_{c\in \appjoin{S}^* \cap W^{\ge h_d}}\frac{1}{h_c}$, where $\appjoin{S}^*$ denotes the union of approval ballots from voters in $S$ w.r.t. instance $\mathcal I^*$. In other words, $S$ can reciprocally afford $\{d\}\cup (\appjoin{S}^*\cap W^{\ge h_d}) = \{d\}\cup (\appjoin{S}\cap W^{\ge h_d})$, but only receives one candidate less than that from $W$ above threshold $h_d$. Therefore, on instance $\mathcal I^*$, the committee $W$ violates reciprocal best-in-slot lower quota, as if a set $U\subseteq W$ existed as required by the axiom, it would have to contain $d$. However, since we only applied well-behaved profile changes, $W\in X(\mathcal I^*)$.\qedhere
\end{proof}

\section{Quality-Price of Representation} \label{app:quality-price_of_representation}

In this section, we show that the quality-prices of representation of reciprocal EJR, reciprocal threshold-PJR+ and reciprocal value-EJR are $\frac{4}{3}$. To prove this, we first show a simple lower bound of $\frac{4}{3}$ and then identify a class of committees that approximate the optimal quality by a factor of $\frac{3}{4}$, namely all quality-greedy completions of reciprocally affordable sets. We then define rules that satisfy reciprocal EJR, reciprocal threshold-PJR+ and reciprocal value-EJR, respectively, while always outputting committees from that class, therefore achieving the optimal worst-case quality-approximation.

\recJRupperbound*

\begin{sideexample}{%
    \centering
    \begin{tabular}{c c c c}
        \toprule
        $k = 2$ & $c_1$ & $c_2$ & $c_3$\\
        quality & $\frac{1}{2} + \varepsilon$ & $1$ & $1$\\
        \midrule
        $v_1$     & \checkmark & & \\
        $\vdots$  & \vdots     & & \\
        $v_{n-1}$ & \checkmark & & \\
        $v_n$     & & \checkmark & \checkmark \\
        \bottomrule
    \end{tabular}
}
    Suppose there was always a reciprocal JR committee achieving a $\frac{3}{4} + \varepsilon$ quality-approximation for some $0 < \varepsilon \leq \frac{1}{4}$. Consider the instance on the right with $n = \lceil \frac{1}{2\varepsilon} \rceil + 1$ voters and $k = 2$. Reciprocal JR forces $c_1 \in W$ since all voters in the group $S = \{v_1, \dots, v_{n-1}\}$ approve $c_1$, and
    $$
        \lvert S \rvert = n - 1
        = \left\lceil \frac{1}{2\varepsilon} \right\rceil = \frac{\left\lceil\frac{1}{2\varepsilon}\right\rceil (1+2\varepsilon)}{1+2\varepsilon}
        \ge \frac{\left\lceil\frac{1}{2\varepsilon}\right\rceil+1}{1+2\varepsilon}
        = \tfrac{1}{h_{c_1}} \cdot \tfrac{n}{k}.
    $$
    Hence any reciprocal JR committee has quality at most $\frac{1}{2} + \varepsilon + 1 = \frac{3}{2} + \varepsilon$ (picking the quality-$1$ candidate $c_2$ or $c_3$ for the second slot). The optimal committee $\{c_2, c_3\}$ has quality $2$, giving quality-approximation ratio of at most $\frac{3/2 + \varepsilon}{2} = \frac{3}{4} + \frac{\varepsilon}{2} < \frac{3}{4} + \varepsilon$, contradicting the assumption.
\end{sideexample}

As reciprocal EJR, reciprocal threshold-PJR+ and reciprocal value-EJR all imply reciprocal JR (see \Cref{fig:relations_reciprocal}), the quality-prices of these notions must also be at least $\frac{4}{3}$.

\begin{corollary}\label{cor:rec_upper_bounds}
    The quality-prices of reciprocal EJR, reciprocal threshold-PJR+ and reciprocal value-EJR are at least $\frac{4}{3}$.
\end{corollary}

We want to show that this lower bound is tight for these three notions. To achieve this, we first identify a class of outcomes that approximate the optimal quality by a factor of $\frac{3}{4}$.

\QualityPriceOfRepresentation*

\begin{proof}
    Let $X \subseteq C$ be reciprocally affordable and $X^+$ be a quality-greedy completion of $X$. We show that $h(X^+) \ge \frac{3}{4}h(W^\star)$.
    We claim that we can assume without loss of generality that $X \cap W^\star = \emptyset$. If $\hat{X} = X \cap W^\star$ is non-empty, consider the modified instance with candidate set $C' = C \setminus \hat{X}$ and committee size $k' = k - |\hat{X}|$. The committee $X' = X \setminus \hat{X}$ is then reciprocally affordable with respect to $k'$, since each candidate costs at least one. When completing $X'$ quality-greedily to obtain ${X'}^+$ the quality of candidates selected by the greedy completion is the same as in the original instance, i.e., $h({X'}^+ \setminus X') = h(X^+ \setminus X)$, since both correspond to the sum of the top $k - |X| = k' - |X'|$ quality scores from candidates in $C \setminus X = C' \setminus X'$. Furthermore, the total quality of any quality-optimal committee ${W'}^\star$ of size $k'$ is given as $h({W'}^\star) = h(W^\star \setminus \hat{X})$. Then
    \begin{align*}
        \frac{h(X^+)}{h(W^\star)} = \frac{h(X^+ \setminus X) + h(X \setminus \hat{X}) + h(\hat{X})}{h(W^\star \setminus \hat{X}) + h(\hat{X})} &\ge \frac{h(X^+ \setminus X) + h(X \setminus \hat{X})}{h(W^\star \setminus \hat{X})} \\ &= \frac{h({X'}^+ \setminus X') + h(X')}{h({W'}^\star)} = \frac{h({X'}^+)}{h({W'}^\star)},
    \end{align*}
    i.e., the total quality ratio for the original instance is at least the ratio for the modified instance, where $X' \cap {W'}^\star = \emptyset$. It thereby suffices to show the theorem for instances where the intersection of $X$ and $W^\star$ is empty.

    We now turn to proving the theorem with the assumption $X \cap W^\star = \emptyset$. Let $t = |X|$.
    We know that the $(k-t)$ candidates $X^+ \setminus X$ chosen by the greedy completion have the same total quality as the $(k-t)$ highest quality candidates from $W^\star$. Hence, their average quality is at least the average quality of the candidates in $W^\star$, i.e., $h(X^+ \setminus X) \ge \frac{k-t}{k} h(W^\star)$. Furthermore, by the AM-HM-inequality and the reciprocal affordability of $X$, we can bound the total quality of $X$ by $h(X) = \sum_{c \in X} h_c \ge \frac{t^2}{\sum_{c \in X} \frac{1}{h_c}} \ge \frac{t^2}{k}$.
    The total quality ratio can then be expressed as
    $$
        f(t)
        := \frac{h(X^+)}{h(W^\star)} \ge \frac{\frac{t^2}{k} + \frac{k-t}{k} h(W^\star)}{h(W^\star)}
        = 1 + \frac{t^2 - t \cdot h(W^\star)}{k \cdot h(W^\star)}.
    $$
    Setting the derivative to zero yields a single extreme point with $2t - h(W^\star) = 0$ and thus at $t = \frac{h(W^\star)}{2}$. Since the second derivative is positive, this point is a minimum. Plugging $t = \frac{h(W^\star)}{2}$ into $f$ yields a lower bound on the ratio of
    $$
        f\Big(\frac{h(W^\star)}{2}\Big)
        \ge 1 + \frac{\big(\frac{h(W^\star)}{2}\big)^2 - \frac{h(W^\star)}{2} \cdot h(W^\star)}{k \cdot h(W^\star)}
        = 1 - \frac{h(W^\star)}{4k}
        \ge 1 - \frac{k}{4k}
        = \frac{3}{4}.
    $$
\end{proof}

\subsection{The Quality-Price of Reciprocal EJR}
We show that the quality-price of reciprocal EJR is $\frac{4}{3}$, by defining reciprocal MES, a rule that always outputs a reciprocal EJR committee that is a quality-greedy completion of a reciprocally affordable set of candidates, and then applying \Cref{thm:quality_approximation}. A key observation is that reciprocal EJR corresponds exactly to the notion of EJR defined in \cite{peters2021proportional} for the participatory budgeting setting, when the prices of each candidate $c$ are set to $\frac{1}{h_c}$. We can then make use of their algorithms to construct a mechanism satisfying reciprocal EJR.
For completeness, we give the definition of their \textit{Method of Equal Shares} (MES) for our setting here and call it reciprocal MES.

\textbf{Reciprocal MES}.
Each voter $i \in V$ starts with budget $b_i \coloneqq \frac{k}{n}$ and each candidate $c \in C$ has cost $\frac{1}{h_c}$. The rule starts with the empty committee $W \coloneqq \emptyset$ and proceeds in rounds. In each round, an unselected candidate $c \in C \setminus W$ is called \emph{$\rho$-reciprocally affordable} if
$$
    \sum_{i \in V[c]} \min(b_i, \rho) \geq \frac{1}{h_c},
$$
i.e., $c$ is reciprocally affordable for $V[c]$ while no voter pays more than $\rho$.
While there exists some reciprocally affordable candidate, pick a candidate $c$ that is $\rho$-reciprocally affordable for a minimum $\rho \geq 0$, add them to $W$ and charge $p_i(c) \coloneqq \min(b_i, \rho)$ to each $i \in V[c]$ by updating $b_i \coloneqq b_i - p_i(c)$. Afterwards, while $\lvert W \rvert < k$, iteratively add the highest quality candidates from $C \setminus W$ to $W$.

\begin{proposition}[\citet{peters2021proportional}]\label{prop:MESrec_recEJR}
  Reciprocal MES satisfies reciprocal EJR.
\end{proposition}

Combining the lower bound from \Cref{cor:rec_upper_bounds} with \Cref{thm:quality_approximation} and the fact that reciprocal MES selects a quality-greedy completion of a reciprocally affordable set yields the following corollary, which gives a $\frac{4}{3}$-quality price of representation of reciprocal EJR.

\begin{corollary}\label{cor:MESrec_34}
    The quality-price of reciprocal EJR is $\frac{4}{3}$.
\end{corollary}

\subsection{The Quality-Price of Reciprocal Threshold-PJR+}

We can adapt the definition of EAR to the reciprocal framework to construct rules that always return reciprocal threshold-PJR+ committees that are quality-greedy completions of reciprocally affordable sets of candidates.

\textbf{Reciprocal EAR.} Start with an empty committee $W$ and budgets $b_i = \frac{k}{n}$ for each voter $i \in V$. Let $h^1 > h^2 > \dots > h^r$ be the distinct quality scores. For $j = 1, \ldots, r$, consider the set $C^j \coloneqq \{c\in C\setminus W \mid h_c \geq h^j \,\land\, \sum_{i\in V[c]} b_i \geq \frac{1}{h_c}\}$ of candidates above threshold $h_j$ that can be reciprocally afforded by their supporters, given their current budgets. While $C^j \neq \emptyset$, iteratively add a candidate $c \in C^j$ to $W$ and let voters in $V[c]$ pay the price $\frac{1}{h_c}$ for $c$ arbitrarily (without overdrawing their budget). Track payments $p_i(c)$, update each $b_i$ by subtracting $p_i(c)$, and refresh $C^j$. Once no candidate in $C^r$ is affordable, complete the committee quality-greedily.

The proof that any reciprocal EAR satisfies reciprocal threshold-PJR+ works the same as for the non-reciprocal equivalents.
We first define the following reciprocal price system notion.

\begin{definition}
    For a given instance and set of candidates $X$, a reciprocal price system $(p, (p_i)_{i \in V})$ is a tuple of a price $p \le 1$ and payment functions $p_i: C \rightarrow \mathbb{R}_{\ge 0}$ for every voter $i \in V$ such that
    \begin{enumerate}
        \item $i\notin V[c]\Rightarrow p_i(c) = 0$,
        \item $\sum_{c\in X} p_i(c) \le \frac{k}{n}$ for all $i \in V$,
        \item $\sum_{i\in V} p_i(c) = \frac{p}{h_c}$ for all $c \in X$, and
        \item $\sum_{i\in V} p_i(d) = 0$ for all $d \in C\setminus X$.
    \end{enumerate}
\end{definition}

We further define the notion of reciprocal threshold-maximal affordability.

\begin{definition}[Reciprocal threshold-maximal affordability]\label{def:reciprocal_threshold_affordable}
    For a given instance, a set of candidates $X$ is called \emph{reciprocal threshold-maximally affordable} iff there exists a reciprocal price system $(p, (p_i)_{i \in V})$, such that for all $h \in [0,1]$ and $d \in C^{\ge h} \setminus X$ it holds that
    $$
        \sum_{i \in V[d]} (\frac{k}{n} - \sum_{c\in X^{\geq h}} p_i(c)) < \frac{p}{h_d}.
    $$
\end{definition}

\begin{proposition} \label{prop:reciprocal_threshold-maximally-affordable_reciprocal_threshold-PJR+}
    Any completion of a reciprocal threshold-maximally affordable set of candidates satisfies reciprocal threshold-PJR+.
\end{proposition}
\begin{proof}
    Let $X \subseteq C$ of size $\le k$ be reciprocal threshold-maximally affordable certified by some reciprocal price system $(p, (p_i)_{i \in V})$ and let $W$ be an arbitrary completion of $X$.
    Assume for contradiction that there exist a candidate $d \notin W$ and a group of supporting voters $S\subseteq V[d]$ with $\lvert S \rvert \frac{k}{n} \ge \frac{1}{h_d} + \sum_{c\in \appjoin{S} \cap W^{\ge h_d}} \frac{1}{h_c}$.
    Since voters from $S$ can only pay for candidates in $\appjoin{S}$, the voters collectively have a remaining budget of $\lvert S \rvert \frac{k}{n} - \sum_{c\in \appjoin{S} \cap W^{\ge h_d}} \frac{1}{h_c} \ge \frac{1}{h_d} \ge \frac{p}{h_d}$, a violation of threshold-maximal affordability.
\end{proof}

We can now show that any reciprocal EAR indeed satisfies reciprocal threshold-PJR+.

\begin{restatable}{proposition}{recEARThresholdPJR}\label{prop:reciprocal_EAR_reciprocal_threshold-PJR+}
   If $f$ is a reciprocal EAR, it satisfies reciprocal threshold-PJR+.
\end{restatable}

\begin{proof}[Proof of \Cref{prop:reciprocal_EAR_reciprocal_threshold-PJR+}.]
    Let $W$ be the committee obtained by a reciprocal EAR. Let $X$ be the uncompleted set of candidates from the reciprocal EAR and $(1, (p_i)_{i \in V})$ the corresponding reciprocal price system. We show that $X$ is reciprocal threshold-maximally affordable. If for some $h$, a candidate $d \in C\setminus X$ with $h_d \ge h$ violates the condition, i.e., $\sum_{i\in V[d]} (\frac{k}{n} - \sum_{c\in X^{\geq h}} p_i(c)) \ge \frac{1}{h_d}$, then we claim that the algorithm would not have stopped the first phase with $X$: after each step during the algorithm with working committee $X'$, we have that $b_i = \frac{k}{n} - \sum_{c\in X'} p_i(c)$. Therefore, consider step $j$ of the algorithm where $h_d = h^j$. Since $h^j = h_d \ge h$, we have $\sum_{i\in V[d]} {b_i} = \sum_{i\in V[d]} (\frac{k}{n} - \sum_{c\in X^{\geq h_d}} p_i(c)) \ge \sum_{i\in V[d]} (\frac{k}{n} - \sum_{c\in X^{\geq h}} p_i(c)) \ge \frac{1}{h_d}$. However, since $d \in C^j$, step $j$ would not have terminated.
    Since $W$ is a completion of the reciprocal threshold-maximally affordable set $X$, it follows from \Cref{prop:reciprocal_threshold-maximally-affordable_reciprocal_threshold-PJR+} that $W$ satisfies reciprocal threshold-PJR+.
\end{proof}

Clearly, any reciprocal EAR elects a quality-greedy completion of a reciprocally affordable set, therefore we obtain a quality-price of $\frac{4}{3}$ for reciprocal threshold-PJR+.

\begin{corollary}\label{cor:RecEAR_34}
    The quality-price of reciprocal threshold-PJR+ is $\frac{4}{3}$.
\end{corollary}

\subsection{The Quality-Price of Reciprocal Value-EJR}

In this section, we introduce the reciprocal equivalent of value-GCR and show that it satisfies reciprocal value-EJR, while always returning quality-greedy completions of reciprocally affordable sets.

\textbf{Reciprocal Value-GCR.} Start with the empty committee $W = \emptyset$ and budgets $b_i = \frac{k}{n}$. We say that $(S, T)$ \emph{witnesses a reciprocal value-EJR violation} if $S \subseteq V$ can reciprocally afford $T$, but for all $i \in S$ we have $h(A_i\cap W) < h(T)$. In the first phase we iteratively choose a reciprocal value-EJR violation witness $(S, T)$ maximizing $h(T)$. We add all candidates from $T \setminus W$ to $W$, and deduct $\frac {\sum_{c \in T \setminus W} \frac{1}{h_c}}{\lvert S\rvert}$ from the budget of each $i \in S$. Proceed until no reciprocal value-EJR witnesses exist. For the second phase, as long as some candidate can be reciprocally afforded by their supporting voters, i.e., $c \in C\setminus W$ satisfies $\sum_{i \in V[c]} b_i \ge \frac{1}{h_c}$, we add $c$ to $W$ and deduct a budget of $\frac{1}{h_c}$ from $V[c]$ arbitrarily (without charging a voter more than their budget) and repeat until no such candidate exists. Finally, if $|W| < k$, we complete $W$ quality-greedily.

\begin{theorem} \label{thm:reciprocal-value-GCR_reciprocal-value-EJR}
    Reciprocal value-GCR satisfies reciprocal value-EJR.
\end{theorem}

\begin{proof}
    Consider any profile $A$ and target size $k$. We first show that reciprocal value-GCR terminates with a committee $W$ of size at most $k$.

    Let $r$ be the number of iterations of phase 1.
    Let $(S_x, T_x)$ be the witness in iteration $x \in [r]$. Moreover, let $W_x$ denote the set of candidates after iteration $x$. We claim that all $S_x$ are disjoint. Indeed, once $T_x \setminus W_{x-1}$ is added to $W_{x-1}$, for all $x' > x$ it holds that $h(A_i \cap W_{x'-1}) \ge h(A_i \cap W_x) \ge h(T_x) \ge h(T_{x'})$ for all $i \in S_x$, and therefore $i \notin S_{x'}$ for all $i \in S_x$.

    We now consider the budgets from this process. Each voter $i$ starts with a budget of $b_i = \frac{k}{n}$ and since the voter sets are disjoint, before iteration $x$ the budget of all voters in $S_x$ is still $\frac{k}{n}$. We split the cost of $\sum_{c \in T \setminus W} \frac{1}{h_c}$ equally over $|S_x| \ge \sum_{c \in T \setminus W} \frac{1}{h_c} \frac{n}{k}$ voters, thus each voter is charged $\frac{\sum_{c \in T \setminus W} \frac{1}{h_c}}{|S_x|} \le \frac{k}{n}$. Thus no voter can have a negative budget after iteration $r$. Phase 2 distributes the cost of selected candidates over supporting voters without exceeding their budget. Let $W'$ be the committee after phase 2. Then for each candidate $c \in W'$ their cost of $\frac{1}{h_c}$ was distributed over the voters without exceeding any voter's budget. The size of $W'$ is therefore bounded by the sum of initial budgets $\sum_{i \in V} b_i = k$. The completion step does not increase the size of $W'$ above $k$ by definition, therefore $|W| \le k$.

   The satisfaction of reciprocal value-EJR follows immediately from the termination condition of phase 1 and $W \supseteq W_r$.
\end{proof}

Clearly, reciprocal value-GCR elects a quality-greedy completion of a reciprocally affordable set, therefore we obtain a quality-price of $\frac{4}{3}$ for reciprocal value-EJR.

\begin{corollary}\label{cor:RecGCR_34}
    The quality-price of reciprocal value-EJR is $\frac{4}{3}$.
\end{corollary}

\section{Welfare-Price of Representation}\label{app:welfare_tradeoffs}
\newcommand{\greedyJR}{\ensuremath{\mathrm{greedy\text{-}JR}}}
In this section, we write $u(W) \coloneqq \sum_{i\in V} h(A_i \cap W)$ for the welfare of a committee $W$ and $u_i(W) \coloneqq h(A_i \cap W)$ for the utility derived by some voter $i\in V$. We denote by $\mathrm{OPT} \coloneqq u(W^\star)$ the welfare of a welfare-optimal committee $W^\star$. We analyze the welfare-price that a proportionality notion provides.

\begin{definition}
    A proportionality notion $\mathcal X$ \emph{admits a welfare approximation} of $\alpha \le 1$ if for every instance $\mathcal I$, some committee $W \in \mathcal X(\mathcal I)$ satisfies $u(W) \ge \alpha \cdot \mathrm{OPT}$.
\end{definition}

The weakest reciprocal-style notion, reciprocal JR, already admits a strong worst-case bound. Note that the classical worst-case examples translate directly: setting $h_c = 1$ for all $c \in C$ collapses reciprocal JR to JR, and quality-based and standard welfare coincide. Moreover, on the instance used in the proof of \citet[Theorem~9]{DBLP:journals/ai/LacknerS20}, JR itself forces the candidate of every singleton party into the committee, so their bound for weakly proportional rules also applies to every committee satisfying JR.

\begin{observation}
  Reciprocal JR admits a welfare approximation of at most
  $\frac{2}{\lfloor\sqrt{k}\rfloor} - \frac{1}{k}$.
\end{observation}

For standard JR, we obtain a worse bound in our setting.

\begin{observation}\label{obs:WelfareWorstCaseJR}
  For any $\varepsilon \in (0,1]$, JR admits a welfare approximation of at most
  $\frac{1}{k} + \frac{(k-1)\varepsilon}{k}$.
\end{observation}
\begin{proof}
  Consider an instance with $n = k$ voters, where each of the first $k-1$ voters
  approves only their own quality-$\varepsilon$ candidate, and the last voter approves only
  $k$ quality-$1$ candidates. JR forces
  selecting one representative per cohesive group, contributing
  $(k-1)\varepsilon$, leaving only one slot for the quality-$1$ block. The
  optimal committee picks all $k$ quality-$1$ candidates, achieving welfare
  proportional to $k$. The ratio is
  $\frac{(k-1)\varepsilon + 1}{k} \to \frac{1}{k}$ as $\varepsilon \to 0$.
\end{proof}

We now turn to comparing the welfare loss imposed by value-JR and
threshold-JR. The worst-case bound for JR from
\Cref{obs:WelfareWorstCaseJR} also implies the same bound for value-JR.
However, note that in this example we need to choose the minimum quality to
be in $\mathcal O(\frac{1}{k})$, so qualities will get arbitrarily small for
large $k$. For threshold-JR, due to its lexicographic nature, this bound
remains even when we have minimum quality $1-\varepsilon$ for arbitrarily
small $\varepsilon>0$:

\begin{table}
  \caption{Instance from \Cref{ex:welfare_price_tJR_vs_vJR}.}
  \label{tab:welfare_price_tJR_vs_vJR}
  \centering
  \begin{tabular}{c c c c c c c c c c}
    \toprule
    $k = n$ & $c_1$ & $\cdots$ & $c_n$ & $c_{n+1}$ & $\cdots$ & $c_{n+k}$\\
    quality &  $1$ & $\cdots$ & $1$ & $1-\varepsilon$ & $\cdots$ & $1-\varepsilon$\\
    \midrule
    $v_1$     & \checkmark & & & \checkmark & $\cdots$ & \checkmark\\
    $\vdots$  & & $\ddots$ & & $\vdots$ & $\ddots$ & $\vdots$\\
    $v_n$     & & & \checkmark & \checkmark & $\cdots$ & \checkmark\\
    \bottomrule
  \end{tabular}
\end{table}

\begin{example} \label{ex:welfare_price_tJR_vs_vJR}
    Consider an instance with $k=n$, $k$ unanimously approved candidates of
    quality $1-\varepsilon$ and for each voter $i\in V$ a distinct candidate
    $c_i$ with quality $1$ (cf. \Cref{tab:welfare_price_tJR_vs_vJR}). Then, threshold-JR forces the selection of the committee of the distinct candidates, giving only a welfare
    approximation of $\frac{k}{(1-\varepsilon)k^2}$, which tends to
    $\frac{1}{k}$ as $\varepsilon \to 0$.
\end{example}

In contrast, we can show that a stronger bound is achievable for value-JR, depending on the minimum quality.

\begin{proposition}\label{prop:minRelWelfareBound}
    Assume that $k\geq 4$. Given an instance with minimum quality $h^-$, a
    value-JR committee $W$ with
    $$
    u(W) \geq \frac{h^-}{8\sqrt{k}} \cdot \mathrm{OPT}
    $$
    can be computed in polynomial time.
\end{proposition}
\begin{proof}
    For candidate $c\in C$, define their \emph{value} as
    $v_c\coloneqq h_c\cdot |V[c]|$. We first define the following greedy-JR
    procedure, which, given $X\subseteq C$, computes a committee
    $W=\greedyJR(X)$ as follows: We initialize $W=X$. We call a candidate
    $d\in C \setminus W$ \emph{blocking} if there exists a set
    $S\subseteq V[d]$ of \emph{unhappy} supporters with
    $|S|\geq \frac{n}{k}$ and $u_i(W)< h_d$ for all $i\in S$. We say that
    $S$ is \emph{served} by adding $d$. As long as such a candidate exists,
    choose a blocking $d\in C\setminus W$ with maximum $h_d$ and add $d$ to
    $W$. If no such candidate exists, we return $W$.
    It is easy to see that for any $X\subseteq C$, the output
    $\greedyJR(X)$ satisfies value-JR. Moreover, it holds that
    $|\greedyJR(\emptyset)|\leq k$, since the sets of unhappy supporters for
    each added blocking candidate are disjoint, as we choose the blocking
    candidates with maximum quality.
    Given an instance with minimum quality $h^-$, we compute the committee
    $W$ as follows.
    \begin{enumerate}
        \item Let $r=\lfloor \sqrt{k} \rfloor$ and let $X\subseteq C$ be the
        set of the $r$ candidates of highest value, ties broken arbitrarily.
        \item If $|\greedyJR(X)|\leq k$, then return a completion of $\greedyJR(X)$.
        \item Otherwise, compute $W'=\greedyJR(\emptyset)$ and return the
        completion of $W'$ to size $k$ with the remaining candidates of
        maximum value.
    \end{enumerate}
    It is clear that the returned committee has size at most $k$ and
    satisfies value-JR. It remains to show the welfare approximation.
    First, assume that we return $\greedyJR(X)$ in step 2. Since the
    welfare of any committee is at most the sum of the $k$ highest values
    and $u(X)$ is the sum of the $r$ highest values, we have
    $u(X)\geq \frac{r}{k}\cdot\mathrm{OPT}$. Then, we are done, since
    $$
    u(\greedyJR(X))\geq u(X)\geq \frac{r}{k}  \cdot \mathrm{OPT}
    = \frac{\lfloor \sqrt{k} \rfloor}{k} \cdot \mathrm{OPT}
    \geq \frac{h^-}{2\sqrt{k}} \cdot \mathrm{OPT},
    $$
    where the last inequality holds since $h^-\leq 1$ and $k\geq 4$ implies
    $\lfloor \sqrt{k} \rfloor \geq \frac{\sqrt{k}}{2}$.
    Thus, it remains to consider the case that $|\greedyJR(X)|> k$. If this
    is the case, we can derive the following useful bound on
    $\mathrm{OPT}$: Let $U\subseteq V$ be the set of unhappy voters that
    were served by adding the $|\greedyJR(X)| - |X| > k -r$ candidates
    during the run of greedy-JR with initial committee $X$. We have that
    $|U|\geq (|\greedyJR(X)| - |X|) \frac{n}{k} > (k -r) \frac{n}{k}$ and
    hence $|V\setminus U| \leq r \frac{n}{k}$.
    Observe that $u_i(X)<1$ for all $i\in U$, as they otherwise cannot be
    unhappy. For $i \in V\setminus U$, we have
    $u_i(X)\leq h(X) \leq |X| \leq r$.
    It follows with $r^2 \leq k$ that
    \begin{equation}\label{valueJRWelfare:eq:OPT_Bound}
        \mathrm{OPT} \leq \frac{k}{r} u(X)
        \leq \frac{k}{r}\Bigl(|U| + |V\setminus U| \cdot r\Bigr)
        \leq \frac{k}{r}\Bigl(n +  \frac{r^2n}{k}\Bigr)
        \leq \frac{k}{r} \cdot 2n \leq \frac{2k}{\sqrt{k}} \cdot 2n
        = 4n\sqrt{k}.
    \end{equation}
    Finally, we bound the welfare of the committee $W$ obtained by
    $\greedyJR(\emptyset)$ completed to size $k$ with the remaining
    candidates of maximum value.
    First, assume that $|\greedyJR(\emptyset)|\geq \frac{k}{2}$. We have
    $u(W)\geq u(\greedyJR(\emptyset)) \geq h^- \cdot
    |\greedyJR(\emptyset)| \cdot \frac{n}{k} \geq \frac{h^-\cdot n}{2}$.
    Thus, we have with \Cref{valueJRWelfare:eq:OPT_Bound} that
    $$
    \frac{u(W)}{\mathrm{OPT}} \geq \frac{h^-\cdot n}{2 \cdot 4n\sqrt{k}}
    = \frac{h^-}{8\sqrt{k}}.
    $$
    Finally, assume that $|\greedyJR(\emptyset)|<\frac{k}{2}$. As we
    complete $\greedyJR(\emptyset)$ to size $k$ with the at least
    $\lceil \frac{k}{2} \rceil$ remaining candidates of maximum value, we
    have that $\frac{u(W)}{\mathrm{OPT}}\geq \frac{1}{2}$: for
    $\ell \in [\lceil \frac{k}{2} \rceil]$, the candidate in $W$ with the
    $\ell$-th highest value has at least the value of the $\ell$-th highest
    value candidate from any welfare-optimal committee.
\end{proof}

\end{document}